\documentclass[a4paper,twocolumn,11pt,accepted=2025-07-29]{quantumarticle}
\pdfoutput=1

\usepackage[english]{babel}
\usepackage{amsmath}
\usepackage{amssymb}
\usepackage{amsthm}
\usepackage{amsfonts}
\usepackage{graphicx,adjustbox}
\usepackage[colorlinks=true, allcolors=blue]{hyperref}
\usepackage{braket}
\newcommand{\vect}[1]{\boldsymbol{#1}}

\usepackage[table,xcdraw]{xcolor}
\usepackage[compat=0.6]{yquant}
\usetikzlibrary{quotes, fit, calc, matrix, positioning, arrows.meta, intersections, through, backgrounds, patterns, graphs.standard}
\usetikzlibrary{datavisualization.formats.functions}
\usepackage{pgfplots}
\usepackage{tikz}
\usepackage[capitalise]{cleveref}
\usepackage[normalem]{ulem}
\usepackage{multirow}
\usepackage{float}
\usepackage[numbers,sort&compress]{natbib}
\usepackage{makecell}

\hypersetup{
    colorlinks=true,       
    linkcolor=cyan,          
    citecolor=magenta,        
    filecolor=magenta,      
    urlcolor=cyan,           
    runcolor=cyan
}

\newtheorem{mytheorem}{Theorem}
\newtheorem{mylemma}{Lemma}
\newtheorem{mycorollary}{Corollary}
\newtheorem{myproposition}{Proposition}

\newtheorem{mydefinition}{Definition}

\newcommand{\txg}[1]{\textcolor{gray}{#1}}
\newcommand{\mbf}[1]{\mathbf{#1}}
\newcommand{\fM}{\mathsf{M}}
\newcommand{\fN}{\mathsf{N}}
\newcommand{\eps}{\ensuremath{\epsilon}} 
\newcommand{\mc}{\mathcal}

\newcommand{\bes} {\begin{subequations}}
\newcommand{\ees} {\end{subequations}}
\newcommand{\beq}{\begin{equation}}
\newcommand{\eeq}{\end{equation}}

\def\ox{\otimes}
\def\>{\rangle}
\def\<{\langle}

\newcommand{\eff}{\overset{\mathrm{g}}{=}}
\definecolor{red}{RGB}{192, 0, 0}
\definecolor{blue}{RGB}{46, 95, 127}
\definecolor{green}{RGB}{52, 105, 40}

\DeclareMathOperator{\rank}{rank}

\begin{document}

\title{Families of $d=2$ 2D subsystem stabilizer codes for universal Hamiltonian quantum computation with two-body interactions}

\author{Phattharaporn Singkanipa}
\affiliation{Center for Quantum Information Science \& Technology}
\affiliation{Department of Physics \& Astronomy}
\orcid{0009-0006-3612-4293}

\author{Zihan Xia}
\affiliation{Center for Quantum Information Science \& Technology}
\affiliation{Department of Electrical \& Computer Engineering}
\orcid{0000-0001-7575-9423}

\author{Daniel A. Lidar}
\affiliation{Center for Quantum Information Science \& Technology}
\affiliation{Department of Physics \& Astronomy}
\affiliation{Department of Electrical \& Computer Engineering}
\affiliation{Department of Chemistry\\
University of Southern California, Los Angeles,
California 90089, USA}
\orcid{0000-0002-1671-1515}

\maketitle

\begin{abstract}
  In the absence of fault tolerant quantum error correction for analog, Hamiltonian quantum computation, error suppression via energy penalties is an effective alternative. We construct families of distance-$2$ stabilizer subsystem codes we call ``trapezoid codes'', that are tailored for energy-penalty schemes. We identify a family of codes achieving the maximum code rate, and by slightly relaxing this constraint, uncover a broader range of codes with enhanced physical locality, thus increasing their practical applicability. Additionally, we provide an algorithm to map the required qubit connectivity graph into graphs compatible with the locality constraints of quantum hardware. Finally, we provide a systematic framework to evaluate the performance of these codes in terms of code rate, physical locality, graph properties, and penalty gap, enabling an informed selection of error-suppression codes for specific quantum computing applications. We identify the $[[4k+2,2k,g,2]]$ family of subsystem codes as optimal in terms of code rate and penalty gap scaling. 
\end{abstract}

\section{Introduction}

Hamiltonian quantum computation, which encompasses approaches such as adiabatic quantum computing (AQC)~\cite{Farhi:00,aharonov_adiabatic_2007,Albash-Lidar:RMP}, quantum annealing~\cite{kadowaki_quantum_1998,Hauke:2019aa,crosson2020prospects}, holonomic quantum computing~\cite{HQC,Duan:2001ff,Zhang:2023aa}, quantum simulation~\cite{Bernien:2017aa,king2022coherent,altmanQuantumSimulatorsArchitectures2019}, and continuous-time quantum walks~\cite{Farhi:1998aa,Childs:2013kx,Mulken:2011aa}, operates by continuously modifying the parameters of a Hamiltonian. In the adiabatic case, for example, the goal is to maintain a quantum system in its ground state while evolving the system’s Hamiltonian from a simple initial form to a complex final one. The solution to a given problem is encoded in the ground state of the final Hamiltonian by design. However, integrating fault tolerant quantum error correction into Hamiltonian-based computation poses a significant challenge due to its analog nature, which complicates the introduction of low-entropy ancilla qubits throughout the computation, as well as syndrome measurements, among other issues. Solutions have been proposed in some cases, e.g., holonomic quantum computation~\cite{Oreshkov:2009bl}, but this involves introducing elements of the circuit model.

To protect Hamiltonian computation, in particular the adiabatic model, Ref.~\cite{jordan2006error} introduced a method that employs stabilizer subspace codes~\cite{Gottesman:1996fk,Calderbank:98} to suppress decoherence. In this approach, the Hamiltonian is encoded by replacing each of its Pauli operators with the corresponding logical Pauli operator of the chosen code, and the code stabilizers are used to construct a penalty Hamiltonian that is added and suppresses the errors the code is designed to detect. Since the penalty Hamiltonian is a sum of commuting terms, it is automatically gapped~\cite{Bravyi:2003tx}.
This results in the suppression of thermal excitations out of the ground subspace, similar to the error suppression provided by topological quantum computing~\cite{Dennis:02,Alicki:07a}. The approach pioneered in Ref.~\cite{jordan2006error} (see also Refs.~\cite{Young:2013fk,Young:13,Sarovar:2013kx,Ganti:13,Bookatz:2014uq,Marvian:2016kb,Marvian:2017aa,Lidar:2019ab}) involves $4$-local interactions, where by `$n$-local' we mean an $n$-body (or weight-$n$) interaction, regardless of the geometry of the interaction graph; later we refer to the latter as `geometric locality.' Since $4$-local interactions are not naturally available (though possible to generate~\cite{Mizel:2004ve,Gurian:2012aa,Dai:2017aa,Chirolli:2024aa}), a practical implementation becomes difficult and an alternative based on naturally available $2$-local interactions is desirable. Unfortunately, a purely $2$-local scheme based on subspace codes that suppresses arbitrary $1$-local errors is impossible~\cite{Marvian:2014nr}.

Circumventing this no-go theorem, it was shown~\cite{Jiang:2015kx,Marvian-Lidar:16} that one can construct non-commuting $2$-local Hamiltonians capable of suppressing general $1$-local errors using stabilizer \emph{subsystem} codes \cite{poulin_stabilizer_2005}. However, the Hamiltonians found in Refs.~\cite{Jiang:2015kx,Marvian-Lidar:16} are insufficient for universality. This obstacle was overcome in Ref.~\cite{marvian2019robust}, which presented a $2$-local construction of a subsystem code whose error-suppressing gauge operators, single-qubit logical operators, and products of two-qubit logical operators (of the same type) are all $2$-local, and which showed how this code can be used to encode universal AQC while suppressing all $1$-local errors. This universality result still leaves several important unanswered questions related primarily to the optimal achievable code rate, geometric locality, and the embedding of the interaction graph in physical hardware connectivity graphs. 

To properly explain these issues, let us first provide some more of the technical background. We wish to protect Hamiltonian quantum computation performed by a system Hamiltonian $H_S(t)$ against the system-bath interaction. The encoded Hamiltonian, $H^{\text{enc}}_S(t)$, is constructed by replacing every operator in $H_S(t)$ with the corresponding logical operators. To protect the computational space, a penalty term $\eps_P H_P$ is added to the encoded Hamiltonian, where $\eps_P>0$ is the energy penalty. The penalty Hamiltonian, $H_P$, is constructed using the elements of the gauge group of the subsystem code, $H_P=-\sum_{g_i\in \mathcal{G}}g_i$. In the large penalty limit, the system becomes decoupled from the bath and performs the desired encoded computation~\cite{Marvian-Lidar:16}, even though in the subsystem case $[H^{\text{enc}}_S(t),H_P]\neq 0$ (unlike the subspace case, where the encoded Hamiltonian commutes with the penalty term). Since the penalty Hamiltonian is now a sum of non-commuting gauge operators, an important disadvantage of subsystem codes is that (unlike in the subspace codes case) their penalty Hamiltonians are no longer necessarily gapped, a fact that complicates their analysis.

Of particular interest is a family of $[[6k,2k,2]]$ stabilizer codes introduced as subspace codes in Ref.~\cite{Ganti:13} and studied as subsystem codes by Marvian and Lloyd~\cite{marvian2019robust}, which enable universal encoded computation within this energy-penalty framework. 
The $[[6k,2k,2]]$ code family studied in \cite{marvian2019robust} can encode a geometrically local Hamiltonian into a new, geometrically local encoded Hamiltonian using all $2$-local operators with a code rate of $2k$ logical qubit per $6k$ physical qubits (i.e., a rate of $1/3$). A natural question that arises is whether this rate can be improved while maintaining the highly desirable properties of $2$-local interactions and geometric locality. Various additional tradeoffs arise that contribute to the practical implementation of subsystem codes.

Here, we study subsystem codes within the energy-penalty framework and systematically study these tradeoffs in terms of a set of criteria outlined below. A key tool in our study is Bravyi's `$A$ matrix' construction of subsystem codes~\cite{Bravyi2011}. As our central result, we generalize the $[[6k,2k,2]]$ code to a new family of $[[4k+2l,2k,g,2]]$ codes, for $k\in\mathbb{Z}^+$, $l\in\{1,\cdots,k\}$, for which the the best rate achievable is $k/(2k+1) \ge 1/3$.  
We demonstrate that every code from this family can have all $2$-local one-qubit logical operators and two-qubit logical operators. Additionally, we introduce an optimization algorithm to map the graph derived from logical operators and gauges to one aligned with hardware connectivity, employing the total Manhattan distance as a performance metric for the resulting graphs. 

An interesting question is whether the methods we develop here can be extended to more general codes, with distance greater than $2$. Unfortunately, codes with distance $d > 2$ would make it impossible to work exclusively with $2$-local Hamiltonians, which is a key requirement that we impose and center our work around. Hence, we do not pursue such an extension here.

The structure of this paper is as follows: \cref{sec:desiderata} gives an overview of our construction. To establish terminology and notation, \cref{sec:background} provides a brief introduction to subspace stabilizer codes, subsystem stabilizer codes and their construction using Bravyi's $A$ matrix. In \cref{sec:extended}, we construct the new family of codes mentioned above along with the corresponding set of optimal set of logical operators, and calculate the energy gap of the penalty Hamiltonian. \cref{sec:optimiza} presents our optimization algorithm for alignment with hardware connectivity. \cref{sec:example} provides two detailed examples illustrating the entire pipeline of our construction and results. \cref{sec:discussion} discusses the various criteria and metrics guiding our code selection. Finally, \cref{sec:conclusion} summarizes our findings. The Appendix contains technical details that complement the main text.

\section{Construction and Desiderata}
\label{sec:desiderata}

We first briefly describe the construction we employ in this work in very general terms, with details and full definitions provided later. 

We start with a system Hamiltonian $H_S(t)$ defined in terms of qubits occupying the vertices $V_S$ of a graph $G_S(V_S,E_S)$ whose edges $E_S$ support the interactions between the qubits:
\begin{align}
\label{eq:H-s}
    H_S(t)&=\sum_i a_i X_i +\sum_i b_i Z_i \notag\\
    &+\sum_{(i,j)\in E^{XX}_S} c_{ij} X_iX_j +\sum_{(i,j)\in E^{ZZ}_S} d_{ij} Z_iZ_j ,
\end{align}
where $X_i$ and $Z_i$ are the Pauli $X$ and $Z$ matrices acting on the $i$'th qubit, with identity acting on all other qubits. The sets $E_S^{XX}$ and $E_S^{ZZ}$ denote edges supporting the $XX$ and $ZZ$ interactions, respectively. Hamiltonians of the form given in \cref{eq:H-s} are universal for AQC, where the coefficients $a_i$, $b_i$, $c_{ij}$, and $d_{ij}$ are all assumed to be smoothly time-dependent and fully controllable~\cite{Biamonte:07,Albash-Lidar:RMP}. A more restricted form with all $c_{ij}=0$ (transverse field Ising model) is commonly used in quantum annealing, where the ground state at the final time is the solution to an optimization problem, typically involving an adiabatic evolution from large initial $a_i$ and small $b_i,d_{ij}$ to small final $a_i$ and large $b_i,d_{ij}$~\cite{kadowaki_quantum_1998,Dwave,q108,Boixo:2014yu,Albash:2017aa,Mandra:2017ab,kowalsky20213regular}.

Second, to protect the computation, this system Hamiltonian is replaced with an encoded Hamiltonian $\hat{H}_S(t)$ where every physical Pauli operator $\sigma_i^\alpha$ is replaced with a dressed logical operator $\hat{\sigma}_i^\alpha$ while preserving the same graph $G_S$ (we clarify the `dressed' terminology in \cref{sec:stab_code}). The vertices of $G_{S}$ are logical qubits. 

Third, the dressed logical operators and gauge group elements are all replaced by $2$-local physical operators. The latter occupy the vertices $V$ of a new graph we call the `induced graph' $G(V,E)$ [We use the standard notation $G(V,E)$ to denote a graph $G$ with vertices $V$ and edge set $E$.]. Unlike in the second step, the vertices of $G$ are physical qubits. To this is added a penalty term $\eps_P H_P$, where the penalty Hamiltonian $H_P$, consisting of a sum of gauge group elements, also lives on the induced graph $G$. 

Finally, we map the induced graph to a new graph we call the `mapped graph' $G_{\text{m}}$, which represents the hardware graph of a physical device. The vertices of $G_{\text{m}}$ are physical qubits. This step can be considered as a permutation of physical qubits on a device such that the connectivity of the device aligns more with that of $G$.

Next, we state five criteria for code selection: code rate, physical locality, induced graph properties, mapped graph properties, and penalty gap. These criteria motivate this work in terms of selecting subsystem codes that optimally balance the various desiderata that guide penalty-protected Hamiltonian computing. 

In stating these criteria we need two technical terms: (1) the \emph{degree} of a qubit is the number of edges from this qubit to other qubits in the same graph; (2) the \emph{geometric locality} of a pair of qubits is $1$ plus the shortest path between them, i.e., $1$ plus the smallest number of edges separating them. E.g., two qubits separated by a single edge are geometrically $2$-local.

\begin{enumerate}

\item Code rate:
With a fixed code distance $d=2$ and a 2D geometry, we aim to maximize the code rate $k/n$ of $[[n,k,d]]$ (subsystem) stabilizer codes. 

\item Physical Locality:
Ideally, at most two-body interactions should be used, i.e., the Hamiltonian should have at most $2$-local terms even after encoding. If this criterion turns out to be impossible to satisfy, then the physical locality should be minimized, i.e., the number of $k$-local terms with $k>2$ should be as small as possible. 

\item Induced graph properties: The induced graph is geometrically $2$-local and should be of fixed and low degree. These requirements arise from the typical constraints of solid-state implementations of quantum computers; the degree requirement can be relaxed for some systems such as trapped atoms or ions, which can support all-to-all qubit connectivity.

The induced graph should ideally be balanced. Here, by `balanced' we mean that each qubit has roughly the same degree.

\item Mapped graph properties:
In general, the induced graph will be quite different from the qubit connectivity graphs of actual devices. To account for this, we map the induced graph to a variety of more practical graphs that correspond to feasible connectivity graphs, in some cases corresponding to actual experimental implementations. Like the induced graph, the mapped graph should be geometrically local, of low degree, and balanced.

\item Penalty Gap:
The `penalty gap' is the energy difference between the ground state and the first excited state of the penalty Hamiltonian $H_P$. Since the encoded Hamiltonian performs the same desired computation as the original system Hamiltonian only in the large penalty limit, a larger penalty gap is preferred. The size of the penalty gap will depend on the specific code.

\end{enumerate}

We revisit these criteria in \cref{sec:discussion} and comment on how the different codes we study satisfy them. First, we provide the required background on subsystem codes.

\section{Background}
\label{sec:background}

Consider a system of $n$ physical qubits with Hilbert space $\mathcal{H}=(\mathbb{C}^{2})^{\otimes n}$, and the Pauli group $\mathcal{P}_n$, generated by all $n$-fold tensor products of the $2\times 2$ identity operator $I$ and the Pauli matrices
\begin{align}
    X=\left[\begin{array}{cc}
       0  &  1\\
       1  & 0
    \end{array}\right],Y=\left[\begin{array}{cc}
       0  &  -i\\
       i  & 0
    \end{array}\right],Z=\left[\begin{array}{cc}
       1  &  0\\
       0  & -1
    \end{array}\right] .
\end{align}
The elements of $\mathcal{P}_n$ can be expressed as $\alpha P_1\otimes P_2\otimes\cdots\otimes P_n$, with $P_j\in\{I,X,Y,Z\}$ and $\alpha\in\{\pm 1, \pm i\}$. The `weight' $|P|$ of a Pauli operator $P$ is the number of non-identity one-qubit Pauli matrices comprising $P$. 

An $[[n,k,d]]$ subspace code encodes $k$ logical qubits into $n>k$ physical qubits~\cite{Knill:1997kx}. The parameter $d$ is the code distance: the code can correct $\lfloor \frac{d-1}{2} \rfloor$ errors and detect $d-1$ errors.

\subsection{Stabilizer Subspace Codes}
\label{sec:stab_code_space}

A stabilizer subspace code is defined by $n-k$ independent, mutually commuting stabilizers $S_j\in \mathcal{P}_n$, $j\in\{1,\cdots,n-k\}$~\cite{Gottesman:1996fk,Calderbank:98}. 
The stabilizers generate an order $2^{n-k}$ Abelian stabilizer group $\mathcal{S}=\langle S_1, \cdots, S_{n-k}\rangle$, where throughout we use angular brackets to denote a generating set of a group. An important subspace is the codespace $\mathcal{C}=\{\ket{\psi}\in\mathcal{H}: S\ket{\psi}=\ket{\psi}~\forall S\in \mathcal{S}\}$, consisting of all states that remain invariant under the action of all elements of the stabilizer group. Consequently, the Hilbert space can be expressed as a direct sum of $\mathcal{C}$ and its orthogonal complement, $\mathcal{C}^\perp$, i.e., $\mathcal{H}=\mathcal{C}\oplus \mathcal{C}^\perp$. 

The stabilizer group $\mathcal{S}$ excludes $-I$ and forms an Abelian subgroup of $\mathcal{P}_n$. The codespace is the simultaneous $+1$ eigenspace of $\mathcal{S}$. The normalizer (also the centralizer $C_{\mathcal{P}_n}(\mathcal{S})$ in this setting\footnote{The centralizer group of a subset ${V}\subset {G}$ in a group ${G}$ is $C_{{G}}({V})=\{ g\in {G}: [g,v_i]=0\ \ \forall v_i\in {V}\}$. The normalizer is $N_{{G}}({V})=\{ g\in {G}: [g,{V}]=0\}$.}) 
$N_{\mathcal{P}_n}(\mathcal{S})$ preserves the codespace $\mathcal{C}$ and encompasses all stabilizers and logical operators. The latter enact non-identity transformations between code states.
Thus, we can decompose the normalizer as $N_{\mathcal{P}_n}(\mathcal{S})=\langle S_1, \cdots, S_{n-k},\bar{X}_1,\bar{Z}_1,\cdots,\bar{X}_k,\bar{Z}_k\rangle$, where $\bar{X}_j$ and $\bar{Z}_j$ are, respectively, logical $X$ and $Z$ operators on logical qubit $j$. 

The code distance $d$ is the minimum weight of a logical operator, i.e., $d=\min\limits_{P\in N_{\mathcal{P}_n}(\mathcal{S})\setminus\mathcal{S}}|P|$. Note that logical operators are equivalent under multiplication by stabilizer elements, so the minimization includes this operation.

\subsection{Stabilizer Subsystem Codes}
\label{sec:stab_code}

An $[[n,k,g,d]]$ stabilizer subsystem code~\cite{poulin_stabilizer_2005} can be understood as an $[[n,k+g,d]]$ stabilizer subspace code where $g$ of the logical qubits and the corresponding logical operators are not used. More precisely, in the subsystem case, we select a subset of $k$ logical qubits from a total of $k+g$ logical qubits of a subspace code, while the remaining $g$ qubits become gauge qubits. This results in a division of the codespace $\mathcal{C}$ into two distinct subsystems $\mathcal{A}$ and $\mathcal{B}$, corresponding to the logical and gauge qubits, respectively. The system Hilbert space decomposes accordingly as $\mathcal{H}=\mathcal{A}\otimes\mathcal{B}\oplus\mathcal{C}^\perp$. 

The logical operators pertaining to the logical qubits used to encode information act as the identity on $\mathcal{B}$ (i.e., on the gauge qubits) and form a non-Abelian group $\mc{L} = \< \bar{X}_{1},\bar{Z}_{1},\cdots, \bar{X}_{k},\bar{Z}_{k}\>$, where $\bar{X}_{i}$ and $\bar{Z}_{i}$ are the logical Pauli operators on the $i$'th logical qubit. These are known as \emph{bare logical operators}; they preserve the code space $\mc{C}$ and act trivially on the gauge qubits. Likewise, the logical operators pertaining to the gauge qubits act trivially on $\mc{A}$ (i.e., on the logical qubits) and form a non-Abelian group $\mc{K} = \< \tilde{X}_{1},\tilde{Z}_{1},\cdots, \tilde{X}_{g},\tilde{Z}_{g}\> \equiv \< \bar{X}_{k+1},\bar{Z}_{k+1},\cdots, \bar{X}_{k+g},\bar{Z}_{k+g}\>$. From here on, we use the bar and tilde decorations to denote the logical and gauge operators, respectively.

The sets $\{\mc{S},\mc{L},\mc{K}\}$ mutually commute.
Subsystem codes are further characterized by their \emph{gauge group} $\mathcal{G} = \<S_1, \cdots, S_{n-k},\tilde{X}_{1},\tilde{Z}_{1},\cdots, \tilde{X}_{g},\tilde{Z}_{g}\>$, which is non-Abelian unless $g=0$. Thus, subsystem codes with an Abelian gauge group are equivalent to subspace stabilizer codes.

Following~\cite{Bravyi2011}, given the gauge group, we may characterize the stabilizers as the elements of $\mc{P}_n$ that commute with every element in $\mc{G}$ and are also in $\mc{G}$, i.e., 
\begin{align}
\label{eq:stabs}
    \mathcal{S} = \mathcal{G} \cap C_{\mathcal{P}_n}(\mathcal{G}) .
\end{align}
We may likewise characterize the bare logical operators as nontrivial logical operators that are elements of $C_{\mathcal{P}_n}(\mc{G})$ but not in $\mc{G}$, i.e., $\mc{L} = C_{\mathcal{P}_n}(\mathcal{G})\setminus\mathcal{G}$. Note that the centralizer of the stabilizer group is $C_{\mathcal{P}_n}(\mathcal{S})=\langle\mathcal{G},\bar{X}_1, \bar{Z}_1, \cdots, \bar{X}_k,\bar{Z}_k\rangle$. The elements of $C_{\mathcal{P}_n}(\mathcal{S})\setminus\mathcal{G}$ are called \textit{dressed logical operators}, as they may act non-trivially on the gauge qubits: dressed logical operators can be written as a product of a bare logical operator and a gauge operator. From here on, we use the bar and hat decorations to denote the bare logical and dressed logical operators, respectively.
The code distance parameter $d$ is the minimum weight of all nontrivial dressed logical operators.

\subsection{Bravyi's $A$ matrix}
\label{sec:bravyi}

In this section, we briefly review a generalization of the two-dimensional (2D) Bacon-Shor code~\cite{Bacon:05,Aliferis:07} by closely following the method introduced in Ref.~\cite{Bravyi2011}. In Bravyi's approach, a square binary $A$ matrix is mapped to a stabilizer subsystem code, where 1-entries in the matrix represent physical qubits. Let $|A|$ denote the Hamming weight (number of $1$'s) in $A$, and let $d_{\mathrm{row}}$ and $d_{\mathrm{col}}$ denote the minimum weight of the non-zero vectors in the row space and column space of $A$, respectively, where linear subspaces are defined over the binary field $\mathbb{F}_2$.

\begin{mytheorem}[Rephrased Theorem 2 in \cite{Bravyi2011}]
\label{th:th2-ref1}

Let $A$ be an arbitrary $m\times m$ binary matrix and $\mathcal{A}$ be the subsystem code associated with $A$. Then $\mathcal{A}$ encodes $k = \rank (A)$ qubits into $n = |A|$ qubits with minimum distance $d = \min (d_{\mathrm{row}},d_{\mathrm{col}})$.
\end{mytheorem}

Using \cref{th:th2-ref1}, we can relate the parameters $[[n,k,d]]$ of an error correcting code as in \cref{tab:nkd}.

\begin{table*}[ht]
    \centering
    \begin{tabular}{ |c|c|c| } 
     \hline
     Code parameters & Properties of $A$ & Parameter meaning \\ 
      \hline
     $n$& $|A|$ & Number of physical qubits \\ 
     $k$ & $\rank(A)$ & Number of logical qubits \\ 
     $d$ & $\min(d_{\mathrm{row}},d_{\mathrm{col}})$ & Distance of the code \\
     \hline
    \end{tabular}
    \caption{Properties of the $A$ matrix and their relation to the parameters of the corresponding error correcting code $[[n,k,d]]$.}
    \label{tab:nkd}
\end{table*}

We also need to identify logical operators and stabilizers for the code, in order to perform calculations and detect errors.

\begin{mydefinition}[Gauge group]
\label{def:gauge}
A subsystem code $\mathcal{A}$ associated with an $A$ matrix has a gauge group $\mathcal{G}$ generated by the  following operators:
    \begin{enumerate}
    \item Every pair of qubits $j,j'$ located in the same row of $A$
contributes a generator $X_jX_{j'}$.
    \item Every pair of qubits $i,i'$ located in the same column of $A$ contributes a generator $Z_iZ_{i'}$.
\end{enumerate}
\end{mydefinition}

In anticipation of the construction of $X$-type and $Z$-type logical operators below, we define:

\begin{mydefinition}
\label{th:row-col-op}
Row and column operators:
 \begin{enumerate}
    \item $R_i=\prod_{j:A_{i,j}=1}Z_{i,j}$ is a \emph{row operator} acting via $Z$ on every qubit in the $i$'th row.
    \item $C_j=\prod_{i:A_{i,j}=1}X_{i,j}$ is a \emph{column operator} acting via $X$ on every qubit in the $j$'th column.
\end{enumerate}
\end{mydefinition}

Having identified the gauge group $\mathcal{G}$, the stabilizer group $\mathcal{S}$ is determined through \cref{eq:stabs}, where the centralizer $C_{\mathcal{P}_n}(\mathcal{G})$ is generated by the row and column operators, as shown in \cite[Prop.~1]{Bravyi2011}:
\begin{align}
    C_{\mathcal{P}_n}(\mathcal{G}) &= \langle R_1, \cdots, R_m, C_1, \cdots, C_m \rangle.
\end{align}

The $A$ matrix represents the commutation relations of the row and column operators. We have $A_{ij}=1$ if $R_i$ and $C_j$ overlap on one qubit (anticommute) and $A_{ij}=0$ if they do not overlap (commute), namely
\begin{align} \label{eq:row_col_comm}
    R_i C_j = (-1)^{A_{ij}} C_j R_i.
\end{align}

With row and column operators, we can represent the $X$-type and $Z$-type operators by bitstrings of length $m$ as follows:

\begin{mydefinition}
\label{th:x-z-op}
    $X$-type operators and $Z$-type operators:
\begin{align}
\label{eq:prod}
    P^X(\vect{x}) = \prod_{j=0}^{m} C_j^{x_j},~~ P^Z(\vect{z}) = \prod_{j=0}^{m} R_j^{z_j},
\end{align}
where $\vect{x}$ and $\vect{z}$ are bitstrings of length $m$. 
\end{mydefinition}
Thus, a 1-entry in the $j$'th position of the bitstring means the $j$'th column or row is involved in the corresponding $X$-type or $Z$-type operator.

Using \cref{eq:row_col_comm}, we can obtain the commutation relations between $X$-type operators and $Z$-type operators:
\begin{align}
    P^X(\vect{x}) P^Z(\vect{z}) &= (-1)^{\vect{z}^T A \vect{x}} P^Z(\vect{z}) P^X(\vect{x}) .
\end{align}

Any stabilizer or logical operator can be expressed using the form in \cref{th:x-z-op}. Specifically, using \cref{eq:stabs} we can find stabilizers by calculating the kernels of $A$ and $A^T$, i.e.:

\begin{myproposition}[\cite{Bravyi2011}]
\label{prop:stabs}
$P^X(\vect{x})$ is an $X$-type stabilizer iff $\vect{x}\in\mathrm{Ker}(A)$ and $P^Z(\vect{z})$ is a $Z$-type stabilizer iff $\vect{z}\in\mathrm{Ker}(A^T)$. Consequently, the number of $X$- and $Z$-type stabilizers is the nullity (dimension of the kernel) of $A$ and $A^T$, respectively.
\end{myproposition}

Using standard Gram-Schmidt orthogonalization one can choose $k$ pairs $\{P_a^X,P_a^Z\}_{a=1}^k$ as bare logical operators subject to the following conditions:
    \begin{enumerate}
    \item $P^X_a P^Z_b = (-1)^{\delta_{ab}} P^Z_b P^X_a$
    \item $C_{\mathcal{P}_n}(\mathcal{G})=\langle\mathcal{S},P^X_1,\cdots,P^X_k,P^Z_1,\cdots,P^Z_k\rangle$
\end{enumerate}
Proofs can be found in Ref.~\cite{Bravyi2011}.

This construction only works for the bare logical operators. Going beyond \cite{Bravyi2011}, we next show how to modify it for the dressed logical operators. 

\begin{mydefinition}
\label{th:row-col-op-dressed}
Reversed row and column operators:
 \begin{enumerate}
    \item $\bar{R}_i=\Pi_{j:A_{i,j}=1} X_{i,j}$ is a \emph{row operator} acting via $X$ on every qubit in the $i$'th row.
    \item $\bar{C}_j=\Pi_{i:A_{i,j}=1} Z_{i,j}$ is a \emph{column operator} acting via $Z$ on every qubit in the $j$'th column.
\end{enumerate}
\end{mydefinition}

These are simply the row and column operators with the roles of $X$ and $Z$ exchanged compared to \cref{th:row-col-op}. These operators are gauge operators which can be generated by the gauge group given in \cref{def:gauge}. Similar to \cref{eq:row_col_comm}, it follows that
\begin{align} 
\label{eq:row_col_comm-bar}
    \bar{R}_i \bar{C}_j = (-1)^{A_{ij}} \bar{C}_j \bar{R}_i.
\end{align}

To establish a commutation relation for dressed logical operators, we make one assumption: Each row and column of the $A$ matrix contains an even number of qubits, which holds for the $A$ matrix in our construction (see \cref{lem:define-A}). Under this assumption, we have $[\bar{R}_i,R_j]=[\bar{C}_i,C_j]=0$, as their overlaps occur on an even number of qubits.
Also, $[\bar{R}_i,C_j]=[\bar{C}_i,R_j]=0$ because they are the same type of operators.

In analogy to \cref{th:x-z-op}, let us define a gauge-multiplied version of these operators.

\begin{mydefinition}
\label{th:x-z-op-dressed}
    $X$-type operators and $Z$-type operators with gauge multiplication:
\begin{align}
\label{eq:prod2}
    Q^X(\vect{x},\bar{\vect{x}}) = \prod_{j=0}^{m} C_j^{x_j}\bar{R}_j^{\bar{x}_j}, Q^Z(\vect{z},\bar{\vect{z}}) = \prod_{j=0}^{m} R_j^{z_j} \bar{C}_j^{\bar{z}_j},
\end{align}
where $\bar{\vect{x}}$ and $\bar{\vect{z}}$ are bitstrings of length $m$. 
\end{mydefinition}

We can obtain the commutation relations between $X$-type and $Z$-type operators with gauge multiplication.

\begin{myproposition}
$X$-type and $Z$-type operators with gauge multiplication obey the following commutation relations:
    \begin{align}
    &Q^X(\vect{x},\bar{\vect{x}}) Q^Z(\vect{z},\bar{\vect{z}}) \\
    &\quad= (-1)^{\vect{z}^T A \vect{x}+\bar{\vect{z}}^T A^T \bar{\vect{x}}} Q^Z(\vect{z},\bar{\vect{z}}) Q^X(\vect{x},\bar{\vect{x}}) .\notag
\end{align}
\label{prop:dressed-commutation}
\end{myproposition}
\begin{proof}
Note that using \cref{eq:row_col_comm} we have $C_i^{x_i}  R_j^{z_j} =(-1)^{z_jA_{ji}x_i}R_j^{z_j}C_i^{x_i}$, and using \cref{eq:row_col_comm-bar} we have $\bar{R}_i^{\bar{x}_i}\bar{C}_j^{\bar{z}_j}=(-1)^{\bar{x}_iA_{ij}\bar{z}_j}\bar{C}_j^{\bar{z}_j}\bar{R}_i^{\bar{x}_i}$. Therefore:
\bes
        \begin{align}
            &Q^X(\vect{x},\bar{\vect{x}})Q^Z(\vect{z},\bar{\vect{z}}) 
            = \prod_{i=0}^{m} C_i^{x_i}\bar{R}_i^{\bar{x}_i} \prod_{j=0}^{m} R_j^{z_j} \bar{C}_j^{\bar{z}_j}\\
            &\qquad= \prod_{i,j=0}^{m} C_i^{x_i}  R_j^{z_j} \bar{R}_i^{\bar{x}_i}\bar{C}_j^{\bar{z}_j}\\
            &\qquad= \prod_{i,j=0}^{m} (-1)^{z_jA_{ji}x_i}R_j^{z_j}C_i^{x_i}  (-1)^{\bar{z}_jA^T_{ji}\bar{x}_i} \bar{C}_j^{\bar{z}_j}\bar{R}_i^{\bar{x}_i}\\
            &\qquad=(-1)^{\vect{z}^T A \vect{x}+\bar{\vect{z}}^T A^T \bar{\vect{x}}}  \prod_{j=0}^{m} R_j^{z_j}\bar{C}_j^{\bar{z}_j}\prod_{i=0}^{m}C_i^{x_i}   \bar{R}_i^{\bar{x}_i}\\
            &\qquad=(-1)^{\vect{z}^T A \vect{x}+\bar{\vect{z}}^T A^T \bar{\vect{x}}} Q^Z(\vect{z},\bar{\vect{z}})Q^X(\vect{x},\bar{\vect{x}}).
        \end{align}
\ees
\end{proof}

Thus, one can choose $k$ pairs $\{Q_a^X,Q_a^Z\}_{a=1}^k$ as dressed logical operators subject to the following conditions:
    \begin{enumerate}
    \item $Q_a^X Q_b^Z = (-1)^{\delta_{ab}} Q^Z_b Q^X_a$
    \item $C_{\mathcal{P}_n}(\mathcal{S})=\langle\mathcal{G},Q^X_1,\cdots,Q^X_k,Q^Z_1,\cdots,Q^Z_k\rangle$
\end{enumerate}

Moreover, row and column permutations of the $A$ matrix leave the subsystem code $\mc{A}$ invariant. Specifically, we show in \cref{app:A-perm} that while the gauge generators are modified under permutations, the gauge group and the stabilizers remain invariant.

\section{
Trapezoid subsystem codes}
\label{sec:extended}

In this section, we introduce a family of $[[4k+2l,2k,g,2]]$ and $[[4k+2l-2,2k-1,g,2]]$ codes for even and odd numbers of logical qubits, respectively. We call these `trapezoid' codes due to the structure of their $A$ matrices, which are defined as follows:

\begin{mydefinition}[$A$ matrix for trapezoid codes]
\label{lem:define-A}
    The $m\times m$ $A$ matrix of the family of trapezoid codes is parameterized by the integers $m$ and $l$, where $l\le \lceil \frac{m-1}{2}\rceil$. Its entries are zero unless specified otherwise below. Its matrix elements are given by: 
    \begin{itemize}
        \item first column: $A_{i,1}=1$ for $1 \le i\le 2l$,
        \item last row: $A_{m,i}=1$ for $m-2l+1 \le i\le m$,
        \item superdiagonal: $A_{i,i+1}=1$ for $ 1\le i \le m-1$,
        \item lower diagonal: $A_{i,i-2l+1}=1$ for $ 2l+1 \le i \le m-1$,
    \end{itemize}
    where $A_{i,j}$ is the element in the $i$'th row and $j$'th column, $i,j\in \{1,\cdots,m\}$.
 \end{mydefinition}

This results in the rotated isosceles trapezoidal structure illustrated in \cref{eq:general-A} (the trapezoidal shape outlined by 1-entries), in which we have indicated the row and column indices where the top of the trapezoid forms. The bottom side of the trapezoid aligns with the superdiagonal of the $A$ matrix and its legs are on the left and the bottom side of the $A$ matrix with length $2l$. 
\begin{align}
\label{eq:general-A}
&\qquad\quad A= \notag\\ 
&\begin{array}{llcccccllllll}
    &\vline&1&1&&&&&&&&&\\
    &\vline&1 &&1&&&&&&\\
    &\vline&\vdots &&&\ddots&&&&&\\
    \scriptstyle{2l}&\vline&1&&&&1&&&&&\\ 
    \scriptstyle{2l+1}&\vline&&1&&&&1&&&\\
    &\vline&&&\ddots&&&&\ddots&&&\\
    &\vline&&&&1&&&&1&\\
    \scriptstyle{m-1}&\vline&&&&&1&&&&1\\
    \scriptstyle{m}&\vline&&&&&&1&\cdots&1&1\\
    \hline
    &\vline&&&&&&\hspace{-.5cm}\scriptstyle{m-2l+1}&\hspace{.5cm}&&\hspace{-.03cm}\scriptstyle{m}&\\
\end{array}
\end{align}

To identify parameters of the trapezoid codes, we use \cref{tab:nkd} together with the following observations:
\begin{itemize}
\item The last row is the sum of all the previous rows and all other rows are linearly independent, thus we have $\rank(A)=m-1$. 
\item $|A| = 2*(m-1)+ m-(m-2l+1)+1 = 2m+2l-2$. 
\item Let $d^{(r)}_{i}$ and $d^{(c)}_{i}$ denote the Hamming weights of the $i$'th row and column, respectively. \cref{lem:define-A} implies that $d^c_i=d^r_i=2$ for every $i$ except $d^c_1=d^r_m=2l$, hence, $\min(d_{\mathrm{row}},d_{\mathrm{col}})=2$. 
\end{itemize}
Thus, the $m\times m$ $A$ matrix corresponds to codes with parameters $[[4k+2l,2k,g,2]]$ for $m=2k+1$ and $[[4k+2l-2,2k-1,g,2]]$ for $m=2k$. 

We refer to the $m=2k+1$ and $m=2k$ cases as the odd-$m$ and even-$m$ codes, respectively. This is not to be confused with the even and odd numbers of logical qubits, which is $m-1$. We show below that 
\bes
\label{eq:gvalues}
\begin{align}
    g&=2k+2l-2\quad \text{for odd }m\\
    g&=2k+2l-3\quad \text{for even }m .
\end{align}
\ees

We note that our construction generalizes the $[[6k,2k,2]]$ code introduced in Ref.~\cite{marvian2019robust}, which is the special case $l=k$ of the odd-$m$ trapezoid code, with $g=4k-2$. 

\begin{table*}[!ht]
    \centering
    \begin{tabular}{|l|l|l|}
    \hline
    Notation & Meaning & Example\\
    \hline
         $A$ matrix & $m \times m$ Bravyi's $A$ matrix in \cite{Bravyi2011} & \cref{eq:general-A} \\
         $\bar{X}_i, \bar{Z}_i$ & Bare logical operators $X,Z$ for logical qubit $i$ &\\
         $\hat{X}_i, \hat{Z}_i$ & Dressed logical operators $X,Z$ for logical qubit $i$&\\
         $\vect{X}, \vect{Z}$& $(m-1) \times m$ operator matrix parameterizing bare logical operators & \cref{eq:log-op}\\
         $\bar{\vect{X}}, \bar{\vect{Z}}$& \makecell[l]{$(m-1) \times m$ gauge matrix parameterizing gauges \\ \qquad\qquad\qquad\qquad\qquad\quad to reduce bare to dressed logical operators} & \cref{eq:log-op-bar}\\
    \hline
    \end{tabular}
    \caption{Summary of notations used in this work and their examples.}
    \label{tab:notations}
\end{table*}

\subsection{Properties of trapezoid codes}
\label{sec:props}

\subsubsection{Rate}
\label{sec:props-rate}

The odd-$m$ and even-$m$ codes have different code rates, i.e.,
\begin{equation}
r_{\mathrm{odd}}=\frac{k}{2k+l}\ ,\quad r_{\mathrm{even}}=\frac{2k-1}{4k+2l-2}.
    \label{eq:rate}
\end{equation}
The code rate is maximized for $l=1$ in both cases. Since the $[[6k,2k,2]]$ code corresponds to $l=k$, its rate is always surpassed by any trapezoid code with $l<k$.

For the same $l$, $r_{\mathrm{odd}}>r_{\mathrm{even}}$. For the same number of physical qubits, also, $r_{\mathrm{odd}}>r_{\mathrm{even}}$. 
Thus, the odd-$m$ codes are usually desirable since they have a higher code rate. In the following discussion, we focus our attention on the odd-$m$ codes unless explicitly noted otherwise. The proofs for the even-$m$ cases follow the same reasoning.

\subsubsection{Stabilizers}

From the rank-nullity theorem, $\rank(A)+\mathrm{nullity}(A)=m$. Hence, using \cref{prop:stabs} and $\rank(A)=m-1=2k$, the number of $X$ and $Z$ stabilizers is $\mathrm{nullity}(A)=m-(m-1)=1$. This shows that each of the codes in the trapezoid family has $|\mathcal{S}|=2$ stabilizers (one $S_X$ and one $S_Z$). 

\begin{mylemma}
\label{lem:stab}
    $S_X=\prod_{i=1}^m C_i$ and $S_Z=\prod_{i=1}^m R_i$ are valid stabilizers of the $A$ matrix of the trapezoid codes family.
\end{mylemma}
\begin{proof}
Using \cref{prop:stabs}, the stabilizers of $X$-type ($Z$-type) can be identified with the right (left) null vectors of $A$. Since $S_X$ ($S_Z$) is a product of all columns (rows), it can be parameterized using a bitstring $\vect{x}=1^m$ ($\vect{z}=1^m$), as in \cref{eq:prod}. Since every row of $A$ has an even number of 1-entries, $A\vect{x}=\vect{0}$. Since every column of $A$ has an even number of entries of $1$, $\vect{z}^TA=\vect{0}^T$.
\end{proof}

Recall that the $1$-entries in an $A$ matrix represent all the physical qubits. It follows that the two stabilizers are the all-$X$ and all-$Z$ Pauli operators. This holds for both odd-$m$ and even-$m$ codes.
\cref{lem:stab} shows that the stabilizers are the same as those of the family of $[[n,n-2,3]]$ subspace stabilizer codes, whose two stabilizers are also the all-$X$ and all-$Z$ Pauli operators.  
Our family of subsystem trapezoid codes generalizes this family of subspace codes by discarding some gauge qubits and imposing a 2D geometry. 

Note that the weight of the stabilizers grows linearly as the number of physical qubits increases. Therefore, using the stabilizers as penalty terms~\cite{jordan2006error} would require at least $n$-body interactions in this case, which is impractical for $n>2$. One of the major advantages of using subsystem codes is that we can replace these high weight stabilizers by $2$-local gauge operators as penalty terms.

\subsubsection{Gauge group}
\label{sec:gauge-group}

For a subspace code with parameters $[[n,k'+g,d]]$ we have $|\mathcal{S}| = n-(k'+g)$~\cite{Gottesman:1996fk}. The same equality holds for the corresponding subsystem code with parameters $[[n,k',g,d]]$ (where $g$ is the number of unused logical qubits or gauge qubits). Thus, for the odd-$m$ trapezoid family, for which $n=4k+2l$ and $k'=2k$, we have $g=2k+2l-2$. Likewise, for even-$m$ codes we have $n=4k+2l-2$ and $k'=2k-1$, so $g=2k+2l-3$. This confirms \cref{eq:gvalues}.

The gauge group $\mathcal{G}$ is generated by the logical operators of the gauge qubits (i.e., $2g$ generators) and the two stabilizers. Hence, $|\mathcal{G}|=2g+2=4(k+l)-2$ (odd-$m$) or $|\mathcal{G}|=4(k+l)-4$ (even-$m$).

\subsection{Example: the $k=3$ family}
For the odd-$m$ codes, we have three $m\times m$ $A$ matrices where $m=2k+1=7$ for $1\le l\le 3$ as shown in \cref{eq:odd-m-A}. 
\begin{widetext}
\begin{align}
\label{eq:odd-m-A}
&\left[\begin{array}{ccccccc}
      1&1&&&&&\\
    1&&1&&&&\\
    &1&&1&&&\\
    &&1&&1&&\\
    &&&1&&1&\\
    &&&&1&&1\\
    &&&&&1&1\\
\end{array}
   \right]
&&\left[\begin{array}{ccccccc}
    1&1&&&&&\\
    1&&1&&&&\\
    1&&&1&&&\\
    1&&&&1&&\\
    &1&&&&1&\\
    &&1&&&&1\\
    &&&1&1&1&1\\
    \end{array}    \right]
    &&\left[\begin{array}{ccccccc}
    1&1&&&&&\\
    1&&1&&&&\\
    1&&&1&&&\\
    1&&&&1&&\\
    1&&&&&1&\\
    1&&&&&&1\\
    &1&1&1&1&1&1\\
    \end{array}    \right]\notag \\
    &\quad l=1:[[14,6,6,2]]&&\quad l=2:[[16,6,8,2]] &&\quad l=3:[[18,6,10,2]]
\end{align}
\end{widetext}

For the even-$m$ codes, we have another three $m\times m$ $A$ matrices where $m=2k=6$ for $1\le l\le 3$ as shown in \cref{eq:even-m-A}.
\begin{align}
\label{eq:even-m-A}
&\left[\begin{array}{cccccc}
    1&1&&&&\\
    1&&1&&&\\
    &1&&1&&\\
    &&1&&1&\\
    &&&1&&1\\
    &&&&1&1\\
    \end{array}    \right]
&&\left[\begin{array}{cccccc}
    1&1&&&&\\
    1&&1&&&\\
    1&&&1&&\\
    1&&&&1&\\
    &1&&&&1\\
    &&1&1&1&1\\
    \end{array}    \right]\notag\\
&\quad l=1: [[12,5,5,2]]&& \quad l=2:  [[14,5,7,2]]\notag\\
& \left[\begin{array}{cccccc}
    1&1&&&&\\
    1&&1&&&\\
    1&&&1&&\\
    1&&&&1&\\
    1&&&&&1\\
    1&1&1&1&1&1\\
    \end{array}    \right]\\
    &\quad l=3: [[16,5,9,2]]\notag
\end{align}

\subsection{Construction of logical operators}\label{sec:optimal_logi}

The choice of logical operators is not unique. Here, we give a set of logical operators as two $(m-1)\times m$ matrices $\vect{X}$ and $\vect{Z}$, where each row of the matrices represents the bitstring corresponding to a logical operator. Let 
\beq
\label{eq:J_i}
J_i=\lfloor\frac{m-i-1}{2l} \rfloor\ , \quad \bar{J}_i=\lfloor\frac{i-1}{2l}\rfloor .
\eeq 
We define:
\bes
\label{eq:logicalXZ}
\begin{align}
\label{eq:logicalX}
    \vect{X}_{i,i+1+2lj}&=1 \notag \\
    &\text{ for } 1\leq i\leq m-1 \text{ and } 0\leq j \leq J_i, \\
\label{eq:logicalZ}
    \vect{Z}_{i,i-2lj}&=1 \notag \\
    & \text{ for } 1\leq i\leq m-1 \text{ and } 0\leq j \leq \bar{J}_i.
\end{align}
\ees

For the three $m=7$ codes in \cref{eq:odd-m-A}, the logical operators constructed from \cref{eq:logicalXZ} are given in the \emph{operator matrices} below:
\begin{widetext}
\bes
\label{eq:log-op}
    \begin{align}
        \vect{X} &= \left[\begin{array}{ccccccc}
        \txg{0}&\mbf{1}&\txg{0}&\mbf{1}&\txg{0}&\mbf{1}&\txg{0}  \\
        \txg{0}&\txg{0}&\mbf{1}&\txg{0}&\mbf{1}&\txg{0}&\mbf{1}  \\
        \txg{0}&\txg{0}&\txg{0}&\mbf{1}&\txg{0}&\mbf{1}&\txg{0}  \\
        \txg{0}&\txg{0}&\txg{0}&\txg{0}&\mbf{1}&\txg{0}&\mbf{1}  \\
        \txg{0}&\txg{0}&\txg{0}&\txg{0}&\txg{0}&\mbf{1}&\txg{0}  \\
        \txg{0}&\txg{0}&\txg{0}&\txg{0}&\txg{0}&\txg{0}&\mbf{1}
    \end{array}\right],~ \vect{Z} = \left[\begin{array}{ccccccc}
         \mbf{1}&\txg{0}&\txg{0}&\txg{0}&\txg{0}&\txg{0}&\txg{0} \\
         \txg{0}&\mbf{1}&\txg{0}&\txg{0}&\txg{0}&\txg{0}&\txg{0} \\
         \mbf{1}&\txg{0}&\mbf{1}&\txg{0}&\txg{0}&\txg{0}&\txg{0} \\
         \txg{0}&\mbf{1}&\txg{0}&\mbf{1}&\txg{0}&\txg{0}&\txg{0} \\
         \mbf{1}&\txg{0}&\mbf{1}&\txg{0}&\mbf{1}&\txg{0}&\txg{0} \\
         \txg{0}&\mbf{1}&\txg{0}&\mbf{1}&\txg{0}&\mbf{1}&\txg{0}
    \end{array}\right] \notag\\
    \label{eq:log-op-a}
    &\qquad\qquad\qquad\qquad\qquad\quad l=1\\
    \vect{X} &= \left[\begin{array}{ccccccc}
        \txg{0}&\mbf{1}&\txg{0}&\txg{0}&\txg{0}&\mbf{1}&\txg{0}  \\
        \txg{0}&\txg{0}&\mbf{1}&\txg{0}&\txg{0}&\txg{0}&\mbf{1}  \\
        \txg{0}&\txg{0}&\txg{0}&\mbf{1}&\txg{0}&\txg{0}&\txg{0}  \\
        \txg{0}&\txg{0}&\txg{0}&\txg{0}&\mbf{1}&\txg{0}&\txg{0}  \\
        \txg{0}&\txg{0}&\txg{0}&\txg{0}&\txg{0}&\mbf{1}&\txg{0}  \\
        \txg{0}&\txg{0}&\txg{0}&\txg{0}&\txg{0}&\txg{0}&\mbf{1}
    \end{array}\right],~ \vect{Z} = \left[\begin{array}{ccccccc}
         \mbf{1}&\txg{0}&\txg{0}&\txg{0}&\txg{0}&\txg{0}&\txg{0} \\
         \txg{0}&\mbf{1}&\txg{0}&\txg{0}&\txg{0}&\txg{0}&\txg{0} \\
         \txg{0}&\txg{0}&\mbf{1}&\txg{0}&\txg{0}&\txg{0}&\txg{0} \\
         \txg{0}&\txg{0}&\txg{0}&\mbf{1}&\txg{0}&\txg{0}&\txg{0} \\
         \mbf{1}&\txg{0}&\txg{0}&\txg{0}&\mbf{1}&\txg{0}&\txg{0} \\
         \txg{0}&\mbf{1}&\txg{0}&\txg{0}&\txg{0}&\mbf{1}&\txg{0}
    \end{array}\right] \notag\\
    \label{eq:log-op-b}
    &\qquad\qquad\qquad\qquad\qquad\quad l=2\\   
    \vect{X} &= \left[\begin{array}{ccccccc}
        \txg{0}&\mbf{1}&\txg{0}&\txg{0}&\txg{0}&\txg{0}&\txg{0}  \\
        \txg{0}&\txg{0}&\mbf{1}&\txg{0}&\txg{0}&\txg{0}&\txg{0}  \\
        \txg{0}&\txg{0}&\txg{0}&\mbf{1}&\txg{0}&\txg{0}&\txg{0}  \\
        \txg{0}&\txg{0}&\txg{0}&\txg{0}&\mbf{1}&\txg{0}&\txg{0}  \\
        \txg{0}&\txg{0}&\txg{0}&\txg{0}&\txg{0}&\mbf{1}&\txg{0}  \\
        \txg{0}&\txg{0}&\txg{0}&\txg{0}&\txg{0}&\txg{0}&\mbf{1}
    \end{array}\right],~ \vect{Z} = \left[\begin{array}{ccccccc}
         \mbf{1}&\txg{0}&\txg{0}&\txg{0}&\txg{0}&\txg{0}&\txg{0} \\
         \txg{0}&\mbf{1}&\txg{0}&\txg{0}&\txg{0}&\txg{0}&\txg{0} \\
         \txg{0}&\txg{0}&\mbf{1}&\txg{0}&\txg{0}&\txg{0}&\txg{0} \\
         \txg{0}&\txg{0}&\txg{0}&\mbf{1}&\txg{0}&\txg{0}&\txg{0} \\
         \txg{0}&\txg{0}&\txg{0}&\txg{0}&\mbf{1}&\txg{0}&\txg{0} \\
         \txg{0}&\txg{0}&\txg{0}&\txg{0}&\txg{0}&\mbf{1}&\txg{0}
    \end{array}\right]\notag\\
    &\qquad\qquad\qquad\qquad\qquad\quad l=3
    \label{eq:log-op-c}
\end{align}
\ees
\end{widetext}
and gauge operators to reduce the weights of these logical operators can be represented in another two $(m-1)\times m$ \emph{gauge matrices} below:
\begin{widetext}
\bes
\label{eq:log-op-bar}
    \begin{align}
        \bar{\vect{X}} &= \left[\begin{array}{ccccccc}
        \txg{0}&\txg{0}&\mbf{1}&\txg{0}&\mbf{1}&\txg{0}&\txg{0}  \\
        \txg{0}&\txg{0}&\txg{0}&\mbf{1}&\txg{0}&\mbf{1}&\txg{0} \\
        \txg{0}&\txg{0}&\txg{0}&\txg{0}&\mbf{1}&\txg{0}&\txg{0}  \\
        \txg{0}&\txg{0}&\txg{0}&\txg{0}&\txg{0}&\mbf{1}&\txg{0}  \\
        \txg{0}&\txg{0}&\txg{0}&\txg{0}&\txg{0}&\txg{0}&\txg{0}  \\
        \txg{0}&\txg{0}&\txg{0}&\txg{0}&\txg{0}&\txg{0}&\txg{0}
    \end{array}\right],~ \bar{\vect{Z}} = \left[\begin{array}{ccccccc}
         \txg{0}&\txg{0}&\txg{0}&\txg{0}&\txg{0}&\txg{0}&\txg{0} \\
         \txg{0}&\txg{0}&\txg{0}&\txg{0}&\txg{0}&\txg{0}&\txg{0} \\
         \txg{0}&\mbf{1}&\txg{0}&\txg{0}&\txg{0}&\txg{0}&\txg{0} \\
         \txg{0}&\txg{0}&\mbf{1}&\txg{0}&\txg{0}&\txg{0}&\txg{0} \\
         \txg{0}&\mbf{1}&\txg{0}&\mbf{1}&\txg{0}&\txg{0}&\txg{0} \\
         \txg{0}&\txg{0}&\mbf{1}&\txg{0}&\mbf{1}&\txg{0}&\txg{0}
    \end{array}\right] \notag\\
    \label{eq:log-op-a-bar}
    &\qquad\qquad\qquad\qquad\qquad\quad l=1\\
    \bar{\vect{X}} &= \left[\begin{array}{ccccccc}
        \txg{0}&\txg{0}&\txg{0}&\txg{0}&\mbf{1}&\txg{0}&\txg{0}  \\
        \txg{0}&\txg{0}&\txg{0}&\txg{0}&\txg{0}&\mbf{1}&\txg{0} \\
        \txg{0}&\txg{0}&\txg{0}&\txg{0}&\txg{0}&\txg{0}&\txg{0}  \\
        \txg{0}&\txg{0}&\txg{0}&\txg{0}&\txg{0}&\txg{0}&\txg{0}  \\
        \txg{0}&\txg{0}&\txg{0}&\txg{0}&\txg{0}&\txg{0}&\txg{0}  \\
        \txg{0}&\txg{0}&\txg{0}&\txg{0}&\txg{0}&\txg{0}&\txg{0}
    \end{array}\right],~ \bar{\vect{Z}} = \left[\begin{array}{ccccccc}
         \txg{0}&\txg{0}&\txg{0}&\txg{0}&\txg{0}&\txg{0}&\txg{0} \\
         \txg{0}&\txg{0}&\txg{0}&\txg{0}&\txg{0}&\txg{0}&\txg{0} \\
         \txg{0}&\txg{0}&\txg{0}&\txg{0}&\txg{0}&\txg{0}&\txg{0} \\
         \txg{0}&\txg{0}&\txg{0}&\txg{0}&\txg{0}&\txg{0}&\txg{0} \\
         \txg{0}&\mbf{1}&\txg{0}&\txg{0}&\txg{0}&\txg{0}&\txg{0} \\
         \txg{0}&\txg{0}&\mbf{1}&\txg{0}&\txg{0}&\txg{0}&\txg{0}
    \end{array}\right] \notag\\
    &\qquad\qquad\qquad\qquad\qquad\quad l=2
    \label{eq:log-op-b-bar}
\end{align}
\ees
\end{widetext}
The $\bar{\vect{X}}$ and $\bar{\vect{Z}}$ for $l=3$ code are two zero matrices since all bare logical operators are $2$-local by construction.

We can construct a set of dressed, $2$-local logical operators using \cref{th:x-z-op-dressed}, where the bitstrings $\vect{x},\bar{\vect{x}}$ and $\vect{z},\bar{\vect{z}}$ are taken from each row of $\vect{X},\bar{\vect{X}}$ and $\vect{Z},\bar{\vect{Z}}$, respectively. An example is illustrated in \cref{fig:example-dressed}. 

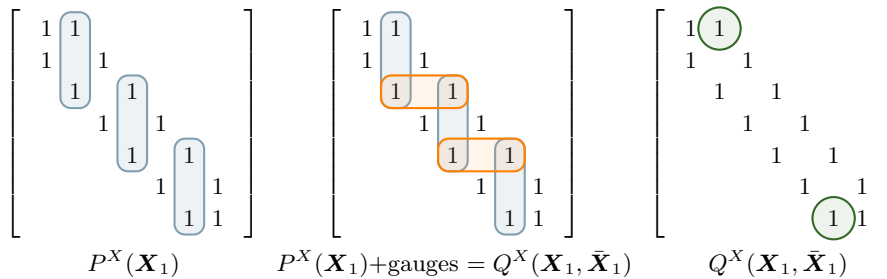
\begin{figure*}[t]
\centering
\begin{tikzpicture}
    \matrix (m)[
    matrix of math nodes,
    nodes in empty cells,
    left delimiter=\lbrack,
    right delimiter=\rbrack
    ] {
     1&1&&&&&\\
    1&&1&&&&\\
    &1&&1&&&\\
    &&1&&1&&\\
    &&&1&&1&\\
    &&&&1&&1\\
    &&&&&1&1\\
    } ;
    \node [below of= m-7-3, xshift=0.4cm, node distance = 0.5cm] {$P^X(\vect{X}_1)$}; 
    \begin{scope}[rounded corners,fill=blue!80!white,fill opacity=0.1,draw=blue,draw opacity=0.6,thick]
    \filldraw (m-1-2.north west)  rectangle (m-3-2.south east);
    \filldraw (m-3-4.north west)  rectangle (m-5-4.south east);
    \filldraw (m-5-6.north west)  rectangle (m-7-6.south east);
   \end{scope}
    \end{tikzpicture}
    \begin{tikzpicture}
    \matrix (m)[
    matrix of math nodes,
    nodes in empty cells,
    left delimiter=\lbrack,
    right delimiter=\rbrack
    ] {
     1&1&&&&&\\
    1&&1&&&&\\
    &1&&1&&&\\
    &&1&&1&&\\
    &&&1&&1&\\
    &&&&1&&1\\
    &&&&&1&1\\
    } ;
    \node [below of= m-7-3, xshift=0.4cm, node distance = 0.5cm] {$P^X(\vect{X}_1)+$gauges $=Q^X(\vect{X}_1,\bar{\vect{X}}_1)$}; 
    \begin{scope}[rounded corners,fill=blue!80!white,fill opacity=0.1,draw=blue,draw opacity=0.6,thick]
    \filldraw (m-1-2.north west)  rectangle (m-3-2.south east);
    \filldraw (m-3-4.north west)  rectangle (m-5-4.south east);
    \filldraw (m-5-6.north west)  rectangle (m-7-6.south east);
   \end{scope}
    \begin{scope}[rounded corners,fill=orange!80!white,fill opacity=0.1,draw=orange,thick]
    \filldraw (m-3-2.north west)  rectangle (m-3-4.south east);
    \filldraw (m-5-4.north west)  rectangle (m-5-6.south east);
   \end{scope}
    \end{tikzpicture}
    \begin{tikzpicture}
    \matrix (m)[
    matrix of math nodes,
    nodes in empty cells,
    left delimiter=\lbrack,
    right delimiter=\rbrack
    ] {
     1&1&&&&&\\
    1&&1&&&&\\
    &1&&1&&&\\
    &&1&&1&&\\
    &&&1&&1&\\
    &&&&1&&1\\
    &&&&&1&1\\
    } ;
    \node [below of= m-7-3, xshift=0.4cm, node distance = 0.5cm] { $Q^X(\vect{X}_1,\bar{\vect{X}}_1)$}; 
   \begin{scope}[fill=green!80!white,fill opacity=0.1,draw=green,thick]
    \filldraw (m-1-2)  circle (8pt);
    \filldraw (m-7-6)  circle (8pt);
   \end{scope}
    \end{tikzpicture}
    \caption{Construction of the dressed logical operator $\hat{X}_1$ for the $[[14,6,6,2]]$ code. $\vect{X}_1$ from \cref{eq:log-op-a} has 1-entries in the 2, 4, and 6 positions, corresponding to those columns of $X$ operators in the $A$ matrix (blue boxes). This is $P^X(\vect{X}_1)$, a six-local operator. Next, use $\bar{\vect{X}}_1$ from \cref{eq:log-op-a-bar}, which has 1-entries in the 3 and 5 positions, corresponding to those rows of $X$ gauges in the $A$ matrix (orange boxes). The blue and orange operators make up the $Q^X(\vect{X}_1,\bar{\vect{X}}_1)$ operator.  After multiplying all the $X$ operators, double operators on the same qubit cancel. Hence, it simplifies to $2$-local dressed logical operator $\hat{X}_1$ (green circles).}
    \label{fig:example-dressed}
\end{figure*}

\begin{mydefinition}[Dressed, $2$-local logical operators for trapezoid codes]
\label{con:dressed2local}
Dressed logical operators are constructed according to \cref{th:x-z-op-dressed}, i.e., $\hat{X}_i=Q^X(\vect{X}_i,\bar{\vect{X}}_i)$ and $\hat{Z}_i=Q^Z(\vect{Z}_i,\bar{\vect{Z}}_i)$ where $\vect{X}_i,\bar{\vect{X}}_i,\vect{Z}_i$ and $\bar{\vect{Z}}_i$ are the $i$'th row of $\vect{X},\bar{\vect{X}},\vect{Z}$ and $\bar{\vect{Z}}$, respectively. By construction, $Q^X(\vect{X}_i,\bar{\vect{X}}_i)$ ($Q^Z(\vect{Z}_i,\bar{\vect{Z}}_i)$) is $P^X(\vect{X}_i)$ ($P^Z(\vect{Z}_i)$) multiplied by appropriate gauge operators that reduce it to $2$-local.
\end{mydefinition}

\cref{tab:notations} summarizes frequently occurring notations we use in this work. Note that the reduction to a $2$-local operator is possible as stated since, by \cref{def:gauge}, every pair of qubits $c,c'$ located in the same row of $A$ contributes a $2$-local gauge generator $X_cX_{c'}$. For $l=k$, the logical operator is $2$-local by construction, hence, no further reduction is needed. Without loss of generality, the indexing in the proofs below will be for odd-$m$ codes. 

\begin{mylemma}
\label{lemma:valid_logical}
    The set of logical operators given in \cref{con:dressed2local} is valid, i.e., for all $i,j$, (1) $\hat{X}_i$ and $\hat{Z}_i$ commute with the two stabilizers, and (2) $\{\hat{X}_i,\hat{Z}_i\}=0$ and $[\hat{X}_i,\hat{Z}_j]=0$ when $i\neq j$.
\end{mylemma}

\begin{proof}
    To prove (1), recall from \cref{lem:stab} that $S_X$ is a product of all $4k+2l$ physical $X$ operators. Every dressed $\hat{Z}_i$ has two physical $Z$ operators. Let $\hat{Z}_i=Z_{i_1}Z_{i_2}$. We have $[S_X,\hat{Z}_i]=[X_{i_1}X_{i_2},Z_{i_1}Z_{i_2}]=0$. Using a similar argument for $S_Z$, it follows that $[S_X,\hat{Z}_i]=[S_Z,\hat{X}_i]=0$.

    To prove (2), consider an operator $P^X(\vect{X}_i)$ and the corresponding $A$ matrix. Each pair of columns of $P^X(\vect{X}_i)$, $(a+1)$'th column and $(b+1)$'s column with $|a-b|=2l$ has two $X$ operators in the same row. Multiplying these rows by $X$-type gauge operators and canceling out $X$ operators on the same qubits, we are left with one $X$ operator in the $(i+1)$'th column, $X_{i,i+1}$, and another $X$ operator in the $q_i$'th column, $X_{m,q_i}$, where $q_i=i+1+2lJ_i$. Thus, $\hat{X}_i = X_{i,i+1}X_{m,q_i}$. 
    
    A similar argument holds for $\hat{Z}_i$. Multiplying $Z$-type gauge (column) operators and canceling out $Z$ operators on the same qubits, we are left with one $Z$ operator in the $i$'th row, $Z_{i,i+1}$, and another $Z$ operator in the $\bar{q}_i$'th row, $Z_{\bar{q}_i,1}$, where $\bar{q}_i = i-2l\bar{J}_i$. Thus, $\hat{Z}_i = Z_{i,i+1}Z_{\bar{q}_i,1}$. 
    
    It is straightforward to see that $\{\hat{X}_i,\hat{Z}_i\}=0$ since they overlap on exactly one qubit, $(i,i+1)$, and that $[\hat{X}_i,\hat{Z}_j]=0$ when $i\neq j$ since the two operators do not overlap.
\end{proof}

We let $A\eff B$ denote $A=B$ up to multiplication by gauge operators, i.e., $A$ and $B$ have the same operation on the logical qubits. 

\begin{mylemma}
\label{lem:2_local_logical}
    The set of dressed logical operators given above is optimal in terms of physical locality, i.e., for all $i$ and $j$, the operators $\hat{X}_i,\hat{Z}_i,\hat{X}_i\hat{X}_j$, and $\hat{Z}_i\hat{Z}_j$ are $2$-local.
\end{mylemma}

Informally, the reason is that any operation on the gauge qubits can be ignored or discarded. Thus, even though the representation of a logical operator such as $\hat{X}_i\hat{X}_j$ is a $4$-local physical operator, two of the physical operators act on the gauge qubits and can hence be ignored. The formal proof follows.

\begin{proof}
    In \cref{lemma:valid_logical}, we showed that for all $i, \hat{X}_i$ and $\hat{Z}_i$ are $2$-local. Now, we  show that for all $i$ and $j, \hat{X}_i\hat{X}_j$ and $\hat{Z}_i\hat{Z}_j$ are also $2$-local. First, using \cref{eq:J_i} note that $q_i=i+1+2lJ_i=m-1$ when $i$ is odd and $q_i=m$ when $i$ is even. When $i$ and $j$ have the same parity, $\hat{X}_i\hat{X}_j=X_{i,i+1}X_{j,j+1}$ since $\hat{X}_i = X_{i,i+1}X_{m,q_i}$ and the $X_{m,q_i}$'s from the two terms cancel. When $i$ and $j$ have a different parity, we have $\hat{X}_i\hat{X}_j=X_{i,i+1}X_{j,j+1}X_{m,q_i}X_{m,q_j}\eff X_{i,i+1}X_{j,j+1}$ since $X_{m,q_i}X_{m,q_j}$ constitutes a gauge operator (as in \cref{def:gauge}). A similar argument holds for $\hat{Z}_i\hat{Z}_j$. Hence, $\hat{Z}_i\hat{Z}_j\eff Z_{i,i+1}Z_{j,j+1}$. 
\end{proof}

The logical operators constructed above are dressed logical operators. Encoding the problem Hamiltonian using dressed logical operators is sufficient to protect the computational space in the large penalty limit~\cite{marvian2019robust}. We show in \cref{lem:bare-not-2} that bare logical operators cannot all be $2$-local.

\begin{mylemma}
\label{lem:bare-not-2}
    Bare logical operators $\bar{X}$ and $\bar{Z}$ cannot all be $2$-local for a code in the trapezoid family with $l<k$.
\end{mylemma}

\begin{proof}
First, we show that an operator is $2$-local if the corresponding bitstring has Hamming weight $1$, which forms rows in the operator matrices of size $(m-1)\times m$. Next, we show that the operator matrices of rank $m-1$ consisting of only two local operators must be an $(m-1)\times (m-1)$ identity padded by an all-zero column up to permutation, as in \cref{eq:log-op-c}. Finally, we show that the operators found in these matrices do not satisfy the commutation condition for logical operators of codes with $l<k$.

First, a $1$-entry in the bitstring in the operator matrix means that a column or a row of the $A$ matrix contributes to the corresponding logical operator. From our construction, each row or column of the $A$ matrix has at least two physical qubits. Since no gauge cancellation is allowed for bare logical operators, a bitstring in 
the operator matrix corresponding to a $2$-local operator can only have Hamming weight $1$.

Next, as shown in the previous paragraph, there is one $1$-entry in each row of the $(m-1)\times m$ operator matrix. Since we need $m-1$ logical operators, the rank of the operator matrix must be $m-1$, which implies that there is one $1$-entry in each column of the operator matrix. Thus, the only way to construct the operator matrix is to have an $(m-1)\times(m-1)$ identity padded by a $(m-1)\times 1$ column of zeros up to permutations. This is because each of the $m-1$ rows and columns must be linearly independent and have Hamming weight $1$.

Finally, if we were to use the operator matrix found in the previous paragraph to construct the logical operators for $l<k$ codes, we could always find two $\bar{X}$ operators that anticommute with the same $\bar{Z}$ operator, which is a contradiction. For example, $\bar{X}_1$ and $\bar{X}_{2l+1}$ both anticommute with $\bar{Z}_{2l+1}$, as seen from using logical operators in \cref{eq:log-op-c} for $l=1$ and $l=2$ codes in \cref{eq:odd-m-A}.

With these conditions, one or more bare logical operators cannot be $2$-local for any $l<k$ code that satisfies the commutation relation.
\end{proof}

In order to find the dressed logical operators, a brute-force approach is to search the entire space of the logical operators multiplied by gauge operators. However, this incurs an exponential cost in $m$. We provide a method for reducing the search space dimension to polynomial in $m$ in \cref{app:reduce_log}.

\subsection{Hamiltonian penalty}
\label{sec:ham-penalty}

The goal is to protect a computation performed by a system Hamiltonian $H_S(t)$ in \cref{eq:H-s} 
against the system-bath interaction described by $I_S\ox H_B+H_{SB}$ where $H_{SB}=\sum_\alpha \sigma_\alpha \otimes B_{\alpha}$ and $\sigma_\alpha$ is a $1$-local system operator (i.e., the $\alpha$ index runs over the physical qubits and $x,y,z$ in the Pauli basis). $H_B$ is the pure-bath Hamiltonian.

The encoded Hamiltonian, $\bar{H}_S(t)$, is constructed by replacing every operator in $H_S(t)$ with the corresponding bare logical operators, i.e.,
\begin{align}
\label{eq:H-bar}
    \bar{H}_S(t)&=\sum_i a_i \bar{X}_i +\sum_i b_i \bar{Z}_i \notag\\
    &+\sum_{(i,j)\in E^{XX}_S} c_{ij} \bar{X}_i\bar{X}_j +\sum_{(i,j)\in E^{ZZ}_S} d_{ij} \bar{Z}_i\bar{Z}_j ,
\end{align}
where $\bar{X}_i$ and $\bar{Z}_i$ are bare logical operators. In this manner, $\bar{H}_S(t)$ is universal for AQC in the code subspace.

We refer to the reducible logical operators as those whose interactions can be reduced to $2$-local upon multiplication by one or more gauge operators. We multiply these operators in $\bar{H}_S(t)$ by appropriate gauge operators to reduce the many-body interactions down to $2$-local and write the new Hamiltonian as $\hat{H}_S(t)$:
\bes
\label{eq:H-hat}
\begin{align}
    &\hat{H}_S(t)
    =\sum_i a_i \bar{X}_i g_i^x +\sum_i b_i \bar{Z}_i g_i^z\notag\\
    \label{eq:H-hat-a}
    &+\sum_{(i,j)\in E^{XX}_S} c_{ij} \bar{X}_i\bar{X}_j g^x_{ij} +\sum_{(i,j)\in E^{ZZ}_S} d_{ij} \bar{Z}_i\bar{Z}_j g^z_{ij} ,\\
        &=\sum_{(i,j)\in E^{X}} a'_{(i,j)} X_{i}X_{j} +\sum_{(i,j)\in E^{Z}} b'_{(i,j)} Z_{i}Z_{j}  \notag\\
        \label{eq:H-hat-b}
        & +\sum_{(i,j)\in E^{XX}} c'_{(i,j)} X_iX_j+\sum_{(i,j)\in E^{ZZ}} d'_{(i,j)} Z_iZ_j.
\end{align}
\ees
In \cref{eq:H-hat-b}, $E^{X}$, $E^{Z}$ , $E^{XX}$, and $E^{ZZ}$ represent the sets of edges in the induced graph $G$ supporting the $\hat{X}$, $\hat{Z}$, $\hat{X}\hat{X}$, and $\hat{Z}\hat{Z}$ interactions, respectively. The operators $X$ and $Z$ are physical operators. This means that $\hat{H}_S(t)$ can be implemented using only $2$-local interactions. The primed coefficients ($a'_i$, etc.) are related to the unprimed ones ($a_i$, etc.) via $\alpha_i$ in $\Pi_0g_i\Pi_0=\alpha_i\Pi_0$, where $\Pi_0$ is the projector to the ground subspace of $H_P$, as explained in detail in \cref{app:penalty}. 

A penalty term $\eps_PH_P$ is added to the encoded Hamiltonian to protect the computational space. The penalty Hamiltonian, $H_P$, is constructed using the generators $\mathcal{G'}$ of the gauge group $\mathcal{G}$, i.e.,
\begin{equation}
    H_P=-\sum_{g_i\in \mathcal{G'}}g_i.
\end{equation} 
The ground subspace of the penalty Hamiltonian is in the codespace~\cite{burton2018spectra}. Thus, the penalty Hamiltonian penalizes transitions into or out of the codespace. In general, $[\hat{H}_S(t),H_P]\neq 0$ due to the gauge terms in \cref{eq:H-hat}, which may interfere with the computation in the codespace. However, Ref.~\cite{marvian2019robust} showed that by rescaling the coefficients, the coupling effect can be compensated. We use a slightly modified version; details can be found in \cref{app:penalty}.
The total Hamiltonian now is:
\begin{equation}
\label{eq:H_total}
    H(t)=[\hat{H}_S(t)+\eps_P H_P]\ox I_B+I_S\ox H_B+H_{SB}.
\end{equation}

The `penalty gap' is the energy difference between the ground state and the first excited state of $H_P$. The separation between the codespace and the lowest excited subspace depends on this gap, the energy penalty $\eps_P$, and the geometry of the induced graph. We calculate the penalty gap for various code parameters for the induced graphs shown in \cref{tikz:gauge_graphs-xz-color}, as well as for additional induced graphs with other values of $k$ and $l$. The calculations and detailed proofs for the numerics in this section are provided in \cref{app:penalty}. 

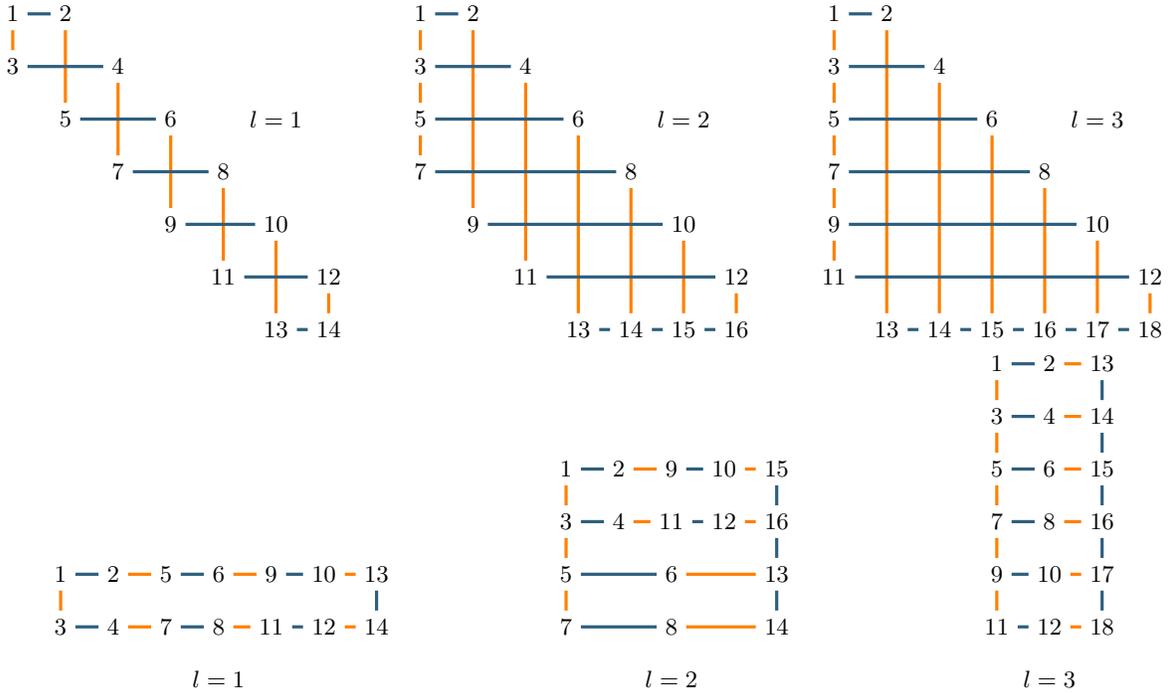
\begin{figure*}[t]
\centering
\begin{tikzpicture}[scale=0.7]
        \begin{scope}
            \node (1) at (0,0) {1};
            \node (2) at (1,0) {2};
            \node (3) at (0,-1) {3};
            \node (4) at (2,-1) {4};
            \node (5) at (1,-2) {5};
            \node (6) at (3,-2) {6};
            \node (7) at (2,-3) {7};
            \node (8) at (4,-3) {8};
            \node (9) at (3,-4) {9};
            \node (10) at (5,-4) {10};
            \node (11) at (4,-5) {11};
            \node (12) at (6,-5) {12};
            \node (13) at (5,-6) {13};
            \node (14) at (6,-6) {14};
        \end{scope}

        \begin{scope}[draw=orange, very thick]
            \draw[-] (10) -- (13);
            \draw[-] (12) -- (14);
            \draw[-] (2) -- (5);
            \draw[-] (4) -- (7);
            \draw[-] (6) -- (9);
            \draw[-] (8) -- (11);
            \draw[-] (1) -- (3);
        \end{scope}

        \begin{scope}[draw=blue, very thick]
            \draw[-] (1) -- (2);
            \draw[-] (3) -- (4);
            \draw[-] (5) -- (6);
            \draw[-] (7) -- (8);
            \draw[-] (9) -- (10);
            \draw[-] (11) -- (12);
            \draw[-] (13) -- (14);
        \end{scope}
        
        \node (l) at (5,-2) {$l=1$} ;
    \end{tikzpicture}
    \qquad
    \begin{tikzpicture}[scale=0.7]
        \begin{scope}
            \node (1) at (0,0) {1};
            \node (2) at (1,0) {2};
            \node (3) at (0,-1) {3};
            \node (4) at (2,-1) {4};
            \node (5) at (0,-2) {5};
            \node (6) at (3,-2) {6};
            \node (7) at (0,-3) {7};
            \node (8) at (4,-3) {8};
            \node (9) at (1,-4) {9};
            \node (10) at (5,-4) {10};
            \node (11) at (2,-5) {11};
            \node (12) at (6,-5) {12};
            \node (13) at (3,-6) {13};
            \node (14) at (4,-6) {14};
            \node (15) at (5,-6) {15};
            \node (16) at (6,-6) {16};
        \end{scope}

        \begin{scope}[draw=orange, very thick]
            \draw[-] (2) -- (9);
            \draw[-] (4) -- (11);
            \draw[-] (10) -- (15);
            \draw[-] (12) -- (16);
            \draw[-] (6) -- (13);
            \draw[-] (8) -- (14);
            \draw[-] (1) -- (3) -- (5) -- (7);
        \end{scope}

        \begin{scope}[draw=blue, very thick]
            \draw[-] (1) -- (2);
            \draw[-] (3) -- (4);
            \draw[-] (5) -- (6);
            \draw[-] (7) -- (8);
            \draw[-] (9) -- (10);
            \draw[-] (11) -- (12);
            \draw[-] (13) -- (14) -- (15) -- (16);
        \end{scope}
        
        \node (l) at (5, -2) {$l=2$} ;
    \end{tikzpicture}
    \qquad
    \begin{tikzpicture}[scale=0.7]
        \begin{scope}
            \node (1) at (0,0) {1};
            \node (2) at (1,0) {2};
            \node (3) at (0,-1) {3};
            \node (4) at (2,-1) {4};
            \node (5) at (0,-2) {5};
            \node (6) at (3,-2) {6};
            \node (7) at (0,-3) {7};
            \node (8) at (4,-3) {8};
            \node (9) at (0,-4) {9};
            \node (10) at (5,-4) {10};
            \node (11) at (0,-5) {11};
            \node (12) at (6,-5) {12};
            \node (13) at (1,-6) {13};
            \node (14) at (2,-6) {14};
            \node (15) at (3,-6) {15};
            \node (16) at (4,-6) {16};
            \node (17) at (5,-6) {17};
            \node (18) at (6,-6) {18};
        \end{scope}

        \begin{scope}[draw=orange, very thick]
            \draw[-] (2) -- (13);
            \draw[-] (4) -- (14);
            \draw[-] (6) -- (15);
            \draw[-] (8) -- (16);
            \draw[-] (10) -- (17);
            \draw[-] (12) -- (18);
            \draw[-] (1) -- (3) -- (5) -- (7) -- (9) -- (11);
        \end{scope}

        \begin{scope}[draw=blue, very thick]
            \draw[-] (1) -- (2);
            \draw[-] (3) -- (4);
            \draw[-] (5) -- (6);
            \draw[-] (7) -- (8);
            \draw[-] (9) -- (10);
            \draw[-] (11) -- (12);
            \draw[-] (13) -- (14) -- (15) -- (16) -- (17) -- (18);
        \end{scope}
        
        \node (l) at (5,-2) {$l=3$} ;
    \end{tikzpicture}
    \qquad
\begin{tikzpicture}[scale=0.7]
        \begin{scope}
            \node (1) at (0,0) {1};
            \node (2) at (1,0) {2};
            \node (5) at (2,0) {5};
            \node (6) at (3,0) {6};
            \node (9) at (4,0) {9};
            \node (10) at (5,0) {10};
            \node (13) at (6,0) {13};
            
            \node (3) at (0,-1) {3};
            \node (4) at (1,-1) {4};
            \node (7) at (2,-1) {7};
            \node (8) at (3,-1) {8};
            \node (11) at (4,-1) {11};
            \node (12) at (5,-1) {12};
            \node (14) at (6,-1) {14};
        \end{scope}

        \begin{scope}[draw=orange, very thick]
            \draw[-] (10) -- (13);
            \draw[-] (12) -- (14);
            \draw[-] (2) -- (5);
            \draw[-] (4) -- (7);
            \draw[-] (6) -- (9);
            \draw[-] (8) -- (11);
            \draw[-] (1) -- (3);
        \end{scope}

        \begin{scope}[draw=blue, very thick]
            \draw[-] (1) -- (2);
            \draw[-] (3) -- (4);
            \draw[-] (5) -- (6);
            \draw[-] (7) -- (8);
            \draw[-] (9) -- (10);
            \draw[-] (11) -- (12);
            \draw[-] (13) -- (14);
        \end{scope}
        
        \node (l) at (3,-2) {$l=1$} ;
    \end{tikzpicture}
    \qquad\qquad\qquad
    \begin{tikzpicture}[scale=0.7]
        \begin{scope}
            \node (1) at (0,1) {1};
            \node (2) at (1,1) {2};
            \node (9) at (2,1) {9};
            \node (10) at (3,1) {10};
            \node (15) at (4,1) {15};
            
            \node (3) at (0,0) {3};
            \node (4) at (1,0) {4};
            \node (11) at (2,0) {11};
            \node (12) at (3,0) {12};
            \node (16) at (4,0) {16};
            
            \node (5) at (0,-1) {5};
            \node (6) at (2,-1) {6};
            \node (13) at (4,-1) {13};
            
            \node (7) at (0,-2) {7};
            \node (8) at (2,-2) {8};
            \node (14) at (4,-2) {14};
        \end{scope}

        \begin{scope}[draw=orange, very thick]
            \draw[-] (2) -- (9);
            \draw[-] (4) -- (11);
            \draw[-] (10) -- (15);
            \draw[-] (12) -- (16);
            \draw[-] (6) -- (13);
            \draw[-] (8) -- (14);
            \draw[-] (1) -- (3) -- (5) -- (7);
        \end{scope}

        \begin{scope}[draw=blue, very thick]
            \draw[-] (1) -- (2);
            \draw[-] (3) -- (4);
            \draw[-] (5) -- (6);
            \draw[-] (7) -- (8);
            \draw[-] (9) -- (10);
            \draw[-] (11) -- (12);
            \draw[-] (15) -- (16) -- (13) -- (14);
        \end{scope}
        
        \node (l) at (2,-3) {$l=2$} ;
    \end{tikzpicture}
    \qquad\qquad\qquad\quad
    \begin{tikzpicture}[scale=0.7]
        \begin{scope}
            \node (1) at (0,2) {1};
            \node (2) at (1,2) {2};
            \node (13) at (2,2) {13};
            
            \node (3) at (0,1) {3};
            \node (4) at (1,1) {4};
            \node (14) at (2,1) {14};
            
            \node (5) at (0,0) {5};
            \node (6) at (1,0) {6};
            \node (15) at (2,0) {15};
            
            \node (7) at (0,-1) {7};
            \node (8) at (1,-1) {8};
            \node (16) at (2,-1) {16};

            \node (9) at (0,-2) {9};
            \node (10) at (1,-2) {10};
            \node (17) at (2,-2) {17};
            
            \node (11) at (0,-3) {11};
            \node (12) at (1,-3) {12};
            \node (18) at (2,-3) {18};
        \end{scope}

        \begin{scope}[draw=orange, very thick]
            \draw[-] (2) -- (13);
            \draw[-] (4) -- (14);
            \draw[-] (6) -- (15);
            \draw[-] (8) -- (16);
            \draw[-] (10) -- (17);
            \draw[-] (12) -- (18);
            \draw[-] (1) -- (3) -- (5) -- (7) -- (9) -- (11);
        \end{scope}

        \begin{scope}[draw=blue, very thick]
            \draw[-] (1) -- (2);
            \draw[-] (3) -- (4);
            \draw[-] (5) -- (6);
            \draw[-] (7) -- (8);
            \draw[-] (9) -- (10);
            \draw[-] (11) -- (12);
            \draw[-] (13) -- (14) -- (15) -- (16) -- (17) -- (18);
        \end{scope}
        
        \node (l) at (1,-4) {$l=3$} ;
    \end{tikzpicture}
    \caption{Gauge connectivity for $k=3$ codes $(l=1,2,3)$. 
    Blue (orange) lines are $XX$ ($ZZ$) gauge operators. The top row arranges qubits in 2D in the form of the $A$ matrix. The bottom row is a layout that minimizes the geometric distance between qubits. The $l=1$ case is always a chain for any $k$.}
    \label{tikz:gauge_graphs-xz-color}
\end{figure*}

\cref{fig:gap} shows the results of numerical calculations of the penalty gap.
We plot the penalty gap for codes with $l=1,2,\dots,7$ and $m=2,3,\dots,20$, where $m\times m$ is the size of the $A$ matrix. The parameters $m$ and $k$ are related via $m=2k+1$ for odd-$m$ codes and $m=2k$ for even-$m$ codes. Hence, $m$ scales with the number of logical qubits $k$ while the number of physical qubits scales with $l$. 

\cref{fig:gap} (left) (supported by a careful analysis of fitting function in \cref{app:scaling}) shows that for a fixed $l$, the gap closes as a power-law for $l=1,2$ with respect to the number of logical qubits (recall that $k=\lfloor m/2\rfloor$), and exponentially for $l\ge 3$. \cref{fig:gap} (middle) shows the scaling as a function of the number of logical qubits for $l(k)=k-l'$ where $l'=0,1,\dots,7$; the gap closes exponentially for all $l(k)$. This gap closure behavior can be explained on the basis of the qubit connectivity of $H_P$ in \cref{tikz:gauge_graphs-xz-color}. 

The connectivity for $l=1$ is chain-like for every $m$ (see \cref{tikz:gauge_graphs-xz-color}). The penalty Hamiltonian then reduces to the (critical) one-dimensional $90^\circ$ quantum compass model: $H_P = -\sum_{i=1}^{m/2} X_{2i-1}X_{2i} + Z_{2i}Z_{2i+1}$ (identifying qubit $m+1$ with qubit $1$), which is dual to the one-dimensional transverse
field Ising model~\cite{Nussinov:2015aa}, whose gap is well known to scale as $1/m$~\cite{PFEUTY197079}. Fitting the $l=1$ results, we find that the gap scales as $1/m^\nu$, where $\nu=1.032\pm 0.004$ (see \cref{tab:params-by-m}). 
Our numerical result is in agreement with the analytical results of Ref.~\cite{Brzezicki:2007aa} for the one-dimensional compass model at the critical point ($\alpha=1$ in their notation), who showed that the gap between the two lowest energy states in the ground state sector scales as $1/m$. This power-law dependence is desirable because it means that the penalty gap ensures protection for all problems whose gap closes faster than $1/m$, as is typically the case in AQC~\cite{Albash-Lidar:RMP}.

For $l=2$ we find two one-dimensional compass models connected via two links (see \cref{tikz:gauge_graphs-xz-color} for $k=3$ and \cref{tikz:gauge_graphs} for $k=4$). Despite a slight curvature at small $m$ visible in \cref{fig:gap} (left), the gap scaling is best fit by a power law (see \cref{tab:each-l-by-m}).

On the other hand, starting from $l=k$, we observe a ladder-like connectivity similar to the compass model \cite{PhysRevB.87.214421,PhysRevB.80.014405}, whose gap decays exponentially with system size \cite{PhysRevB.72.024448}. The difference between the $l=k$ connectivity and that of the standard 2D compass model is the missing middle leg of the ladder (see $l=k=3$ in \cref{tikz:gauge_graphs-xz-color} and $l=k=4$ in \cref{tikz:gauge_graphs}). It is reasonable to expect that despite this difference, $H_P$'s spectrum is closer to the 2D than the 1D compass model, which leads to the penalty gap decaying exponentially rather than as a power law. This is confirmed numerically in \cref{fig:gap} (middle), and we find that the gap scales as $e^{(-0.45\pm 0.01)m}$ (for full details of the fitting parameters see \cref{tab:params-by-m-lp}). We note that this conclusion differs from \cite{marvian2019robust}, who claimed that the penalty gap closes polynomially.

For $2<l<k$ we find a connectivity that is intermediate between the ladder-like structure of the $l=k$ case and the double one-dimensional compass model of the $l=2$ case (see $l=3$ in \cref{tikz:gauge_graphs}). This structure is closer to a 2D rather than 1D connectivity, suggesting that the gap closure for $2<l<k$ is closer to the exponentially decaying gap of the 2D compass model than the power law decay of the 1D model. This is borne out by our numerics, which exhibit a more exponential-like decay the higher the value of $l$ (see \cref{tab:each-l-by-m}).

\cref{fig:gap} (right) shows that the gap decreases exponentially as a function of $l$ when $m$ is fixed. This suggests that codes with smaller $l$ (smaller number of physical qubits) provide an exponentially larger penalty gap. These codes with smaller $l$ also have better code rates, which is desirable. However, as we discuss in the next section, there is a tradeoff between these advantages offered by small-$l$ codes and the higher degree of physical qubit connectivity they require when we also incorporate the encoded system Hamiltonian $\hat{H}_S$.

\begin{figure*}[t]
\centering
     \includegraphics[width=0.32\linewidth]{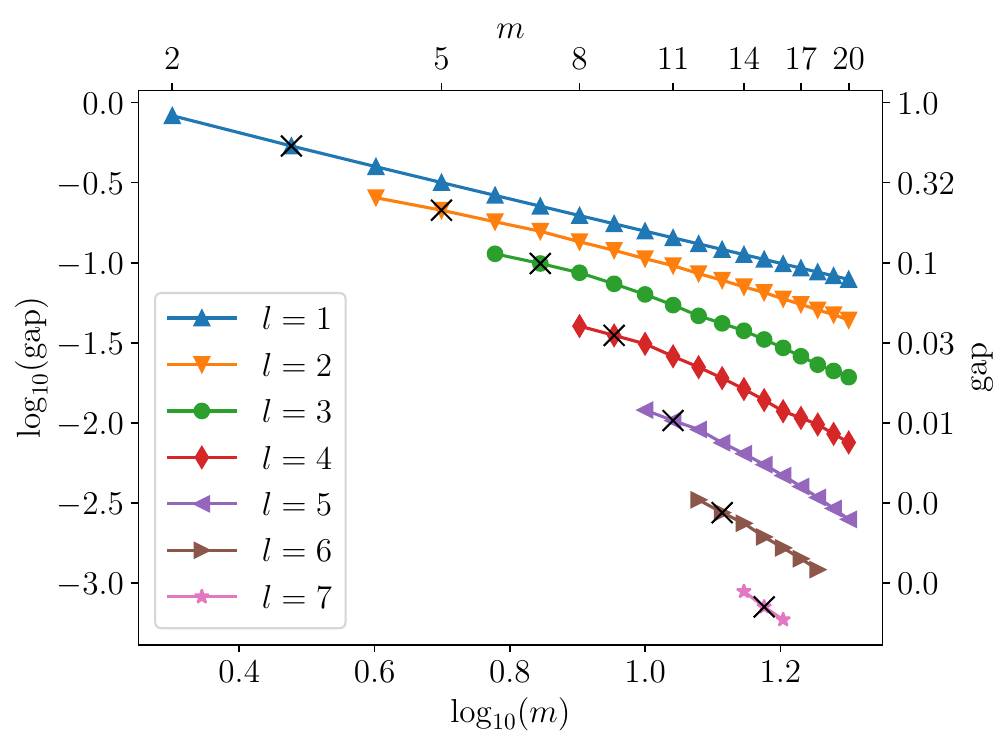}
     \includegraphics[width=0.33\linewidth]{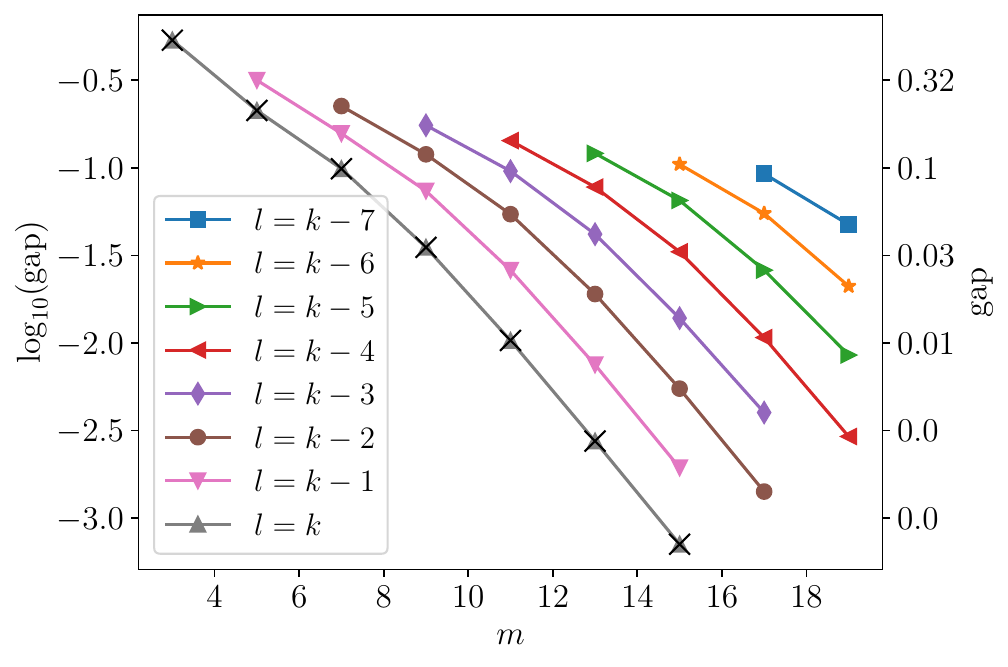}
        \includegraphics[width=0.33\linewidth]{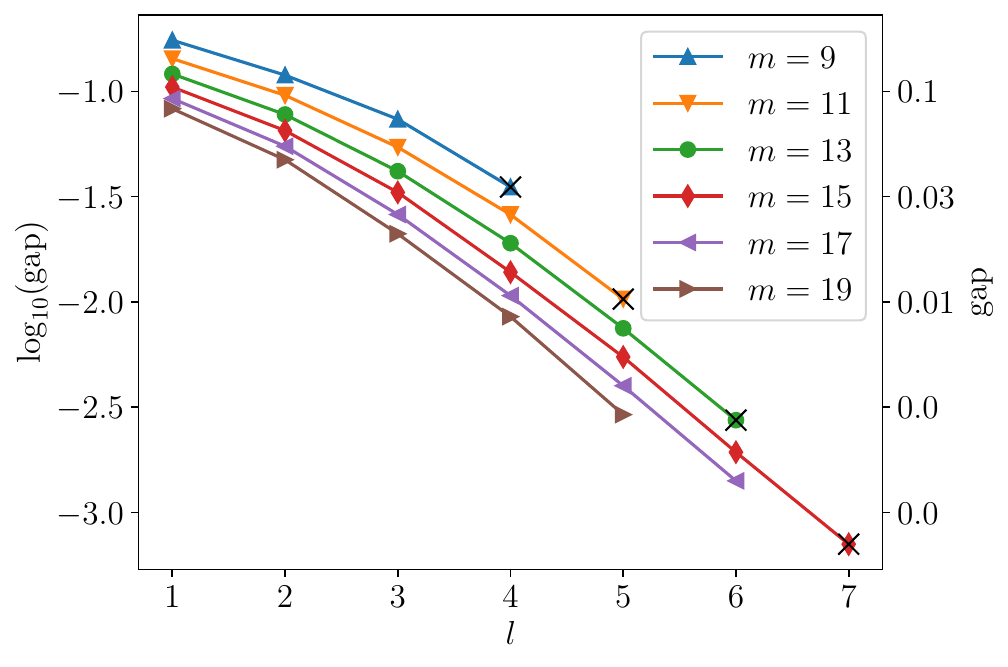}
\caption{Penalty Hamiltonian gap as a function of the code matrix size $m$ for different values of the code parameter $l$ (left, middle) and as a function of $l$ for different values of $m$ (right). The number of logical qubits is $k=\lfloor m/2 \rfloor$ and the number of physical qubits is $4k+2l$ for odd $m$ or $4k+2(l-1)$ for even $m$. The black crosses are the results for the odd-$m$ codes with $l=k$~\cite{marvian2019robust}. The gap is largest for $l=1$ (which also maximizes the code rate) and decays more slowly than for $l>1$. The decay with $m$ at fixed $l$ follows a power-law for $l=1,2$ and is exponential for $l>2$ (left). The decay with $m$ at fixed $l=k-l'$ is exponential in $m=2k+1$ (middle). The decay as a function of $l$ at fixed $m$ is exponential in all cases (right). See \cref{app:scaling} for details. Missing data points for higher $l$ and $m$ are due to insufficient memory (we used 256GB GPUs).}
\label{fig:gap}
\end{figure*}

\section{Qubit Mapping}
\label{sec:optimiza}

The induced graphs $G$ that arise from the logical operators and gauges constructed in \cref{sec:optimal_logi}, have certain features in common. The central part of the induced graph forms a fully connected subgraph with $m-1$ vertices (corresponds to the physical qubits on the superdiagonal in $A$ matrix), representing the connectivity needed for two-qubit logical operators $\hat{X}\hat{X}$ and $\hat{Z}\hat{Z}$, as shown by the yellow edges in \cref{fig:dressed-gauge-connect}. Additionally, qubits located in the first column and the last row of the $A$ matrix are each coupled to $\lfloor\frac{m-1}{2l}\rfloor$ or $\lceil\frac{m-1}{2l}\rceil$ vertices in the fully connected subgraph in an alternating manner, accounting for the single-qubit logical operators $\hat{X}$ and $\hat{Z}$, as illustrated by the blue edges in \cref{fig:dressed-gauge-connect}. Finally, the qubits in the first column and the last row form two cyclic components, representing the connectivity required for gauge generators, corresponding to part of the green edges in \cref{fig:dressed-gauge-connect}.

The cyclic couplings among the qubits in the first column and the last row exist since we select nearest-neighbor gauge generators. The structure of these two components depends on the choice of gauge generators. For example, if the penalty Hamiltonian includes the entire gauge group, these two components will be fully connected subgraphs, each with $2l$ vertices.

For the case $l=1$ (which maximizes the code rate), $G$ consists of $m-1$ vertices (blue nodes in \cref{fig:dressed-gauge-connect}) with degrees of $m+1$ or $m+2$, $m-3$ vertices (green nodes in \cref{fig:dressed-gauge-connect}) with degrees of $2$, and $4$ vertices (light and dark brown nodes in \cref{fig:dressed-gauge-connect}) with degrees of $(m+1)/2$. At the other extreme, when $l=k$, $G$ includes one fully connected subgraph and two cyclic graphs, each with $m-1$ vertices. In this case, the induced graph consists of $m-1$ vertices (blue nodes in \cref{fig:dressed-gauge-connect}) with degrees of $m+1$ and $2m-2$ vertices (light and dark brown nodes in \cref{fig:dressed-gauge-connect}) with degrees of $3$, resulting in a more balanced induced graph. 

Moving from $l=k$ to $l=1$, more gauge cancellation is involved to ensure all logical operators become $2$-local. Consequently, the connectivity induced by the logical operators overlaps less with that induced by the gauges. (As seen in \cref{fig:dressed-gauge-connect}, the number of green nodes and green edges excluding those connecting brown nodes increases from right to left.) In the $l=1$ case, while the number of physical qubits required for the connectivity induced by the logical operators is smaller than in the $l=k$ case, additional qubits and edges are needed for the gauges, forming a graph with one degree more than that of the $l=k$ case. On the other hand, when $l=k$, the logical operators and gauges require the same connectivity, making the graph more balanced. 

\begin{figure*}
    \centering
     \includegraphics[width=0.31\linewidth]{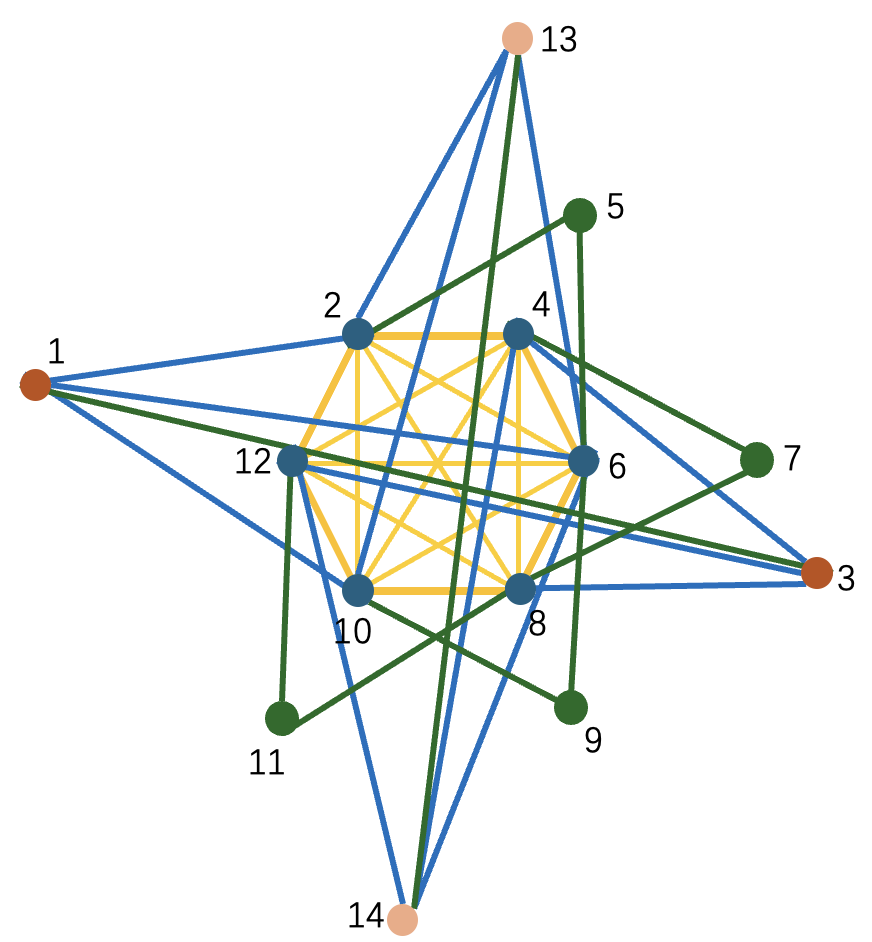}
         \includegraphics[width=0.31\linewidth]{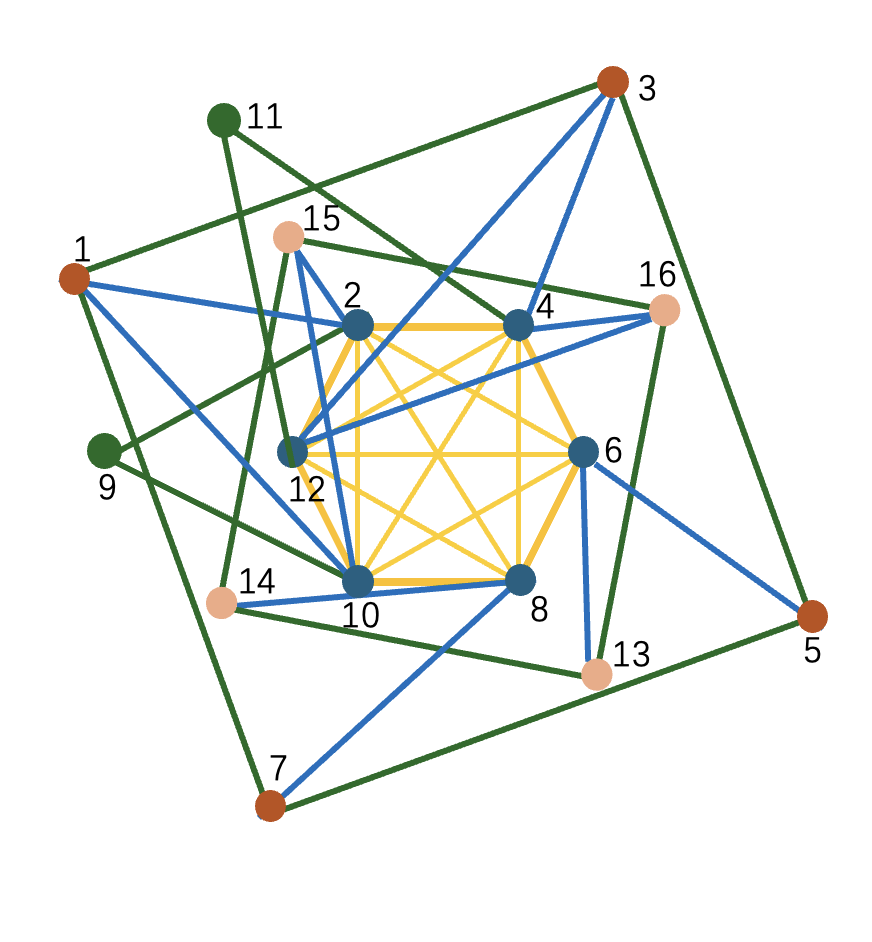}
         \includegraphics[width=0.31\linewidth]{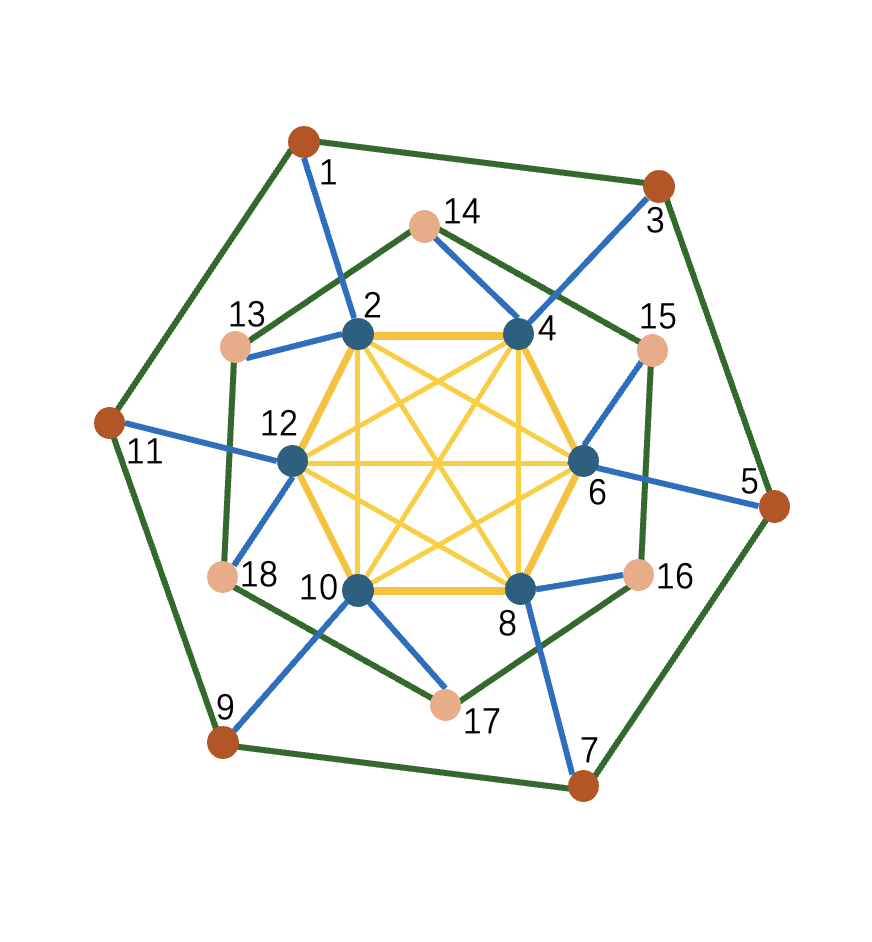}
        \caption{The induced graphs required to support the single-qubit logical operators, two-qubit logical operators, and gauge operators are shown for codes with $k=3$ and $l=1,2,3$, whose matrix representations are given in \cref{eq:log-op}. Each node represents a physical qubit, whose index [as in \cref{eq:A-matrix-number-1}] is a node label. The dark blue nodes correspond to qubits located on the superdiagonal, which must be coupled to support the two-qubit logical operators $\hat{X}\hat{X}$ and $\hat{Z}\hat{Z}$, with the yellow edges representing these couplings. The light brown nodes represent the qubits in the last row, and the dark brown nodes represent the qubits in the first column. The dark blue edges indicate the connectivity required for $\hat{X}$ and $\hat{Z}$. The green nodes represent qubits along the lower diagonal, with green edges showing the additional couplings required for gauge operators. The graphs shown in \cref{tikz:gauge_graphs-xz-color} are subgraphs of those shown here, representing just the connectivity required for the penalty Hamiltonians.}
        \label{fig:dressed-gauge-connect}
\end{figure*}

The degree of the induced graph increases linearly with the number of logical qubits, which is undesirable for practical implementation. Therefore, we aim to map it to a new graph that is more feasible or aligns well with the connectivity of existing hardware. 

\subsection{Mapping induced graph to mapped graph}
Denote the set of physical qubits as $V$. Connect the physical qubits to support all logical operators $\hat{X}$, $\hat{Z}$, $\hat{X}\hat{X}$, $\hat{Z}\hat{Z}$, and gauge generators, resulting in the induced graph $G$ illustrated in \cref{fig:dressed-gauge-connect}. The objective is to map the induced graph $G(V,E)$ to another graph $G_{\mathrm{m}}(V_{\mathrm{m}},E_{\mathrm{m}})$, where $G_{\mathrm{m}}$ represents an existing hardware connectivity graph or a more feasible hardware graph, we call it the mapped graph. The distance between two vertices in $G_{\mathrm{m}}$ is quantified using the Manhattan distance.

\begin{mydefinition}[Manhattan Distance] 
    Consider a graph $G(V,E)$. The Manhattan distance between $x,y\in V$, denoted as $m(x,y)$, is the length of the shortest path from $x$ to $y$. In other words, $m(x,y)$ is the smallest number of edges in $E$ necessary to connect $x$ and $y$.
\end{mydefinition}

Since we use undirected graphs, we have $m(x,y)=m(y,x)$. In addition, we define the total mapped Manhattan distance, or just the total Manhattan distance for short, as

\begin{mydefinition}[Total Manhattan Distance] \label{def:total-manhattan}
    Given two graphs $G(V,E)$ and $G_{\mathrm{m}}(V_{\mathrm{m}},E_{\mathrm{m}})$ and a map from $V$ to $V_{\mathrm{m}}$, the total Manhattan distance of $G_{\mathrm{m}}$ given $G$ is defined as: 
 \beq
    m^E_{\mathrm{total}}(E_{\mathrm{m}})=\sum_{(x,y)\in E} m_{E_{\mathrm{m}}}(x',y'), 
\eeq
where $x',y'\in V_{\mathrm{m}}$ are the mapped vertices corresponding to $x,y\in V$, and $m_{E_{\mathrm{m}}}(x',y')$ denotes the Manhattan distance between $x'$ and $y'$ in $G_{\mathrm{m}}$.
\end{mydefinition}

In the analog quantum computation model, since SWAP gates are not explicitly applied, the interactions in \cref{eq:H-hat-b} are implemented simultaneously. Therefore, although some edges are shared by both the logical operators and the gauges, they are counted only once.
The goal is to find a mapping from $V$ to $V_{\mathrm{m}}$ that minimizes the total Manhattan distance, equivalent to minimizing the number of SWAP gates required in the circuit model. 

Another useful metric is the following:
\begin{mydefinition}[Average Manhattan distance] The average Manhattan distance is the average over the original edges
\beq
\tilde{m}=\frac{m^E_{\text{total}}(E')}{|E|}=\frac{\sum_{(x,y)\in E} m_{E'}(x',y')}{|E|} .
\eeq
\end{mydefinition}

A smaller average Manhattan distance is preferred. 

\subsection{Optimal mappings from a brute-force search}
\label{subsec:brute-map}

We present the mapping results for four cases with $m\in\{4,5\}$ and $l\in\{1,2\}$, and explore seven potential hardware layouts: Union Jack, square (Rigetti's Ankaa), triangular, IBM's heavy-hex, hexagonal, kagome, and Rigetti's Aspen, as depicted in \cref{tikz:hardware_graphs}.
Given the NP-hardness of the mapping problem, we employ a brute-force search to find the optimal mappings for modest problem sizes. For larger problem sizes, the mixed-integer programming method provided in \cref{app:MIPO} can be employed. Detailed optimal mapping results obtained through brute-force search are tabulated in \cref{table:mapping_results}.

\begin{figure}[h!]
\centering
\begin{tikzpicture}[scale=0.7]
        \begin{scope}
            \node (3) at (-1,0) {3};
            \node (4) at (0,0) {4};
            \node (5) at (1,0) {5};
            
            \node (6) at (-1,-1) {6};
            \node (7) at (0,-1) {7};
            \node (8) at (1,-1) {8};

            \node (0) at (-1,1) {0};
            \node (1) at (0,1) {1};
            \node (2) at (1,1) {2};
            
            \node (9) at (-1,-2) {9};
            \node (10) at (0,-2) {10};
            \node (11) at (1,-2) {11};
        \end{scope}

        \begin{scope}[draw=gray]
            \draw[-] (0) -- (1) -- (2);
            \draw[-] (3) -- (4) -- (5);
            \draw[-] (6) -- (7) -- (8);
            \draw[-] (9) -- (10) -- (11);
            
            \draw[-] (0) -- (3) -- (6) -- (9);
            \draw[-] (1) -- (4) -- (7) -- (10);
            \draw[-] (2) -- (5) -- (8) -- (11);

            \draw[-] (0) -- (4) -- (8);
            \draw[-] (1) -- (5);
            \draw[-] (3) -- (7) -- (11);
            \draw[-] (6) -- (10);
            
            \draw[-] (1) -- (3);
            \draw[-] (2) -- (4) -- (6);
            \draw[-] (5) -- (7) -- (9);
            \draw[-] (8) -- (10);
        \end{scope}
        \node (l) at (0,2) {Union Jack} ;
    \end{tikzpicture}
    \qquad
    \begin{tikzpicture}[scale=0.7]
        \begin{scope}
            \node (3) at (-1,0) {3};
            \node (4) at (0,0) {4};
            \node (5) at (1,0) {5};
            
            \node (6) at (-1,-1) {6};
            \node (7) at (0,-1) {7};
            \node (8) at (1,-1) {8};

            \node (0) at (-1,1) {0};
            \node (1) at (0,1) {1};
            \node (2) at (1,1) {2};
            
            \node (9) at (-1,-2) {9};
            \node (10) at (0,-2) {10};
            \node (11) at (1,-2) {11};
        \end{scope}

        \begin{scope}[draw=gray]
            \draw[-] (0) -- (1) -- (2);
            \draw[-] (3) -- (4) -- (5);
            \draw[-] (6) -- (7) -- (8);
            \draw[-] (9) -- (10) -- (11);
            
            \draw[-] (0) -- (3) -- (6) -- (9);
            \draw[-] (1) -- (4) -- (7) -- (10);
            \draw[-] (2) -- (5) -- (8) -- (11);
        \end{scope}
        \node (l) at (0,2) {Square} ;
    \end{tikzpicture}
    \qquad
    \begin{tikzpicture}[scale=0.7]
        \begin{scope}
            \node (0) at (-1.1547,1) {0};
            \node (1) at (0,1) {1};
            \node (2) at (1.1547,1) {2};
            \node (3) at (2.3094,1) {3};

            \node (4) at (-1.73205,0) {4};
            \node (5) at (-0.57735,0) {5};
            \node (6) at (0.57735,0) {6};
            \node (7) at (1.73205,0) {7};

            \node (8) at (-2.3094,-1) {8};
            \node (9) at (-1.1547,-1) {9};
            \node (10) at (0,-1) {10};
            \node (11) at (1.1547,-1) {11};
        \end{scope}

        \begin{scope}[draw=gray]
            \draw[-] (0) -- (1) -- (2) -- (3);
            \draw[-] (4) -- (5) -- (6) -- (7);
            \draw[-] (8) -- (9) -- (10) -- (11);
            
            \draw[-] (0) -- (4) -- (8);
            \draw[-] (1) -- (5)-- (9);
            \draw[-] (2) -- (6)-- (10);
            \draw[-] (3) -- (7)-- (11);

            \draw[-] (0) -- (5) -- (10);
            \draw[-] (1) -- (6) -- (11);
            \draw[-] (4) -- (9);
            \draw[-] (2) -- (7);
        \end{scope}
        \node (l) at (0,2) {Triangular} ;
    \end{tikzpicture}
    \quad
    \begin{tikzpicture}[scale=0.7]
        \begin{scope}
            \node (0) at (0,1) {0};
            \node (1) at (-0.5,0.866025) {1};
            \node (2) at (-1,0.57735) {2};
            \node (3) at (-1,0) {3};

            \node (4) at (-1,-0.57735) {4};
            \node (5) at (-0.5,-0.866025) {5};
            
            \node (6) at (1,0) {6};
            \node (7) at (1,0.57735) {7};
            \node (8) at (0.5,0.866025) {8};
            
            \node (9) at (0,1.5) {9};
            \node (10) at (-1.5,0.866025) {10};
            \node (11) at (-1.5,-0.866025) {11};
            \node (12) at (1.5,0.866025) {12};
        \end{scope}

        \begin{scope}[draw=gray]
            \draw[-] (6) -- (7) -- (8) -- (0) -- (1) -- (2) -- (3) -- (4) -- (5);
            \draw[-] (0) -- (9);
            \draw[-] (2) -- (10);
            \draw[-] (4) -- (11);
            \draw[-] (7) -- (12);
        \end{scope}
        \node (l) at (0,2.5) {Heavy-hex} ;
    \end{tikzpicture} \\
    
    \begin{tikzpicture}[scale=0.7]
        \begin{scope}
            \node (0) at (-0.5, 0.866025) {0};
            \node (1) at (0.5, 0.866025) {1};
            \node (2) at (1, 0) {2};
            \node (3) at (0.5, -0.866025) {3};
            \node (4) at (-0.5, -0.866025) {4};
            \node (5) at (-1, 0) {5};

            \node (6) at (-1, 1.73205) {6};
            \node (7) at (1, 1.73205) {7};
            \node (8) at (2, 0) {8};
            \node (9) at (1, -1.73205) {9};
            \node (10) at (-1, -1.73205) {10};
            \node (11) at (-2,0) {11};
        \end{scope}

        \begin{scope}[draw=gray]
            \draw[-] (0) -- (1) -- (2) -- (3) -- (4) -- (5) -- (0);
            \draw[-] (0) -- (6);
            \draw[-] (1) -- (7);
            \draw[-] (2) -- (8);
            \draw[-] (3) -- (9);
            \draw[-] (4) -- (10);
            \draw[-] (5) -- (11);
        \end{scope}
        \node (l) at (0,2.5) {Hexagonal} ;
    \end{tikzpicture}
    \quad
    \begin{tikzpicture}[scale=0.7]
        \begin{scope}
            \node (0) at (-0.5, 0.866025) {0};
            \node (1) at (0.5, 0.866025) {1};
            \node (2) at (1, 0) {2};
            \node (3) at (0.5, -0.866025) {3};
            \node (4) at (-0.5, -0.866025) {4};
            \node (5) at (-1, 0) {5};

            \node (6) at (0, 1.73205) {6};
            \node (7) at (1.5, 0.866025) {7};
            \node (8) at (1.5, -0.866025) {8};
            \node (9) at (0, -1.73205) {9};
            \node (10) at (-1.5, -0.866025) {10};
            \node (11) at (-1.5, 0.866025) {11};
        \end{scope}

        \begin{scope}[draw=gray]
            \draw[-] (0) -- (1) -- (2) -- (3) -- (4) -- (5) -- (0);
            \draw[-] (0) -- (6) -- (1);
            \draw[-] (1) -- (7) -- (2);
            \draw[-] (2) -- (8) -- (3);
            \draw[-] (3) -- (9) -- (4);
            \draw[-] (4) -- (10) -- (5);
            \draw[-] (5) -- (11) -- (0);
        \end{scope}
        \node (l) at (0,2.5) {Kagome} ;
    \end{tikzpicture}
    \quad
    \begin{tikzpicture}[scale=0.7]
        \begin{scope}
            \node (0) at (-2, 1.866025) {0};
            \node (1) at (-1, 1.866025) {1};
            \node (2) at (-2, 0.866025) {2};
            \node (3) at (-1, 0.866025) {3};
            
            \node (4) at (-0.5, 0.5) {4};
            \node (5) at (0.5, 0.5) {5};
            \node (10) at (-0.5, -0.5) {10};
            \node (11) at (0.5, -0.5) {11};

            \node (6) at (1, 1.866025) {6};
            \node (7) at (2, 1.866025) {7};
            \node (8) at (1, 0.866025) {8};
            \node (9) at (2, 0.866025) {9};

            \node (12) at (-1, -0.866025) {12};
            \node (13) at (1, -0.866025) {13};
            
        \end{scope}

        \begin{scope}[draw=gray]
            \draw[-] (3) -- (1) -- (0) -- (2) -- (3) -- (4);
            \draw[-] (8) -- (6) -- (7) -- (9) -- (8) -- (5);
            \draw[-] (4) -- (5) -- (11) -- (10) -- (4);
            \draw[-] (10) -- (12);
            \draw[-] (11) -- (13);
        \end{scope}
        \node (l) at (0,3) {Rigetti's Aspen} ;
    \end{tikzpicture}
    \caption{The seven mapped graphs ($G_{\mathrm{m}}$) we consider. Numbers denote indices of vertices, which are used in the mapping. Solid lines denote enabled connectivity. We choose small portions of periodically repeating patterns in each graph type due to computational limitations.}
    \label{tikz:hardware_graphs}
\end{figure}

\begin{table*}
    \begin{tabular}{ |p{1.6cm}||p{2.7cm}|p{3.4cm}|p{3.3cm}| p{3.9cm}| }
         \hline
         \multicolumn{5}{|c|}{Optimal mapping results} \\
         \hline
         Mapped graph& $[[8,3,3,2]]$ & $[[10,3,5,2]]$ &$[[10,4,4,2]]$ & $[[12,4,6,2]]$\\
         \hline
         Union Jack & [1,3,5,7,0,4,10,6]   & [0,3,1,4,2,6,5,7,8,10] &  [1,3,5,7,0,4,11,8,6,10] & [0,3,1,4,2,6,5,7,11,8,10,9]\\
         \hline
         Square &   [1,3,2,5,0,4,8,7]  & [0,1,3,4,6,5,9,10,7,8]   &[0,1,3,6,2,4,10,7,5,8]& [0,1,3,4,6,7,9,10,2,5,8,11]\\
         \hline
         Heavy-hex    &[0,1,8,3,10,2,5,4] & [4,0,3,1,2,8,10,9,7,6]&  [0,1,6,7,3,2,12,8,10,9]& [3,0,4,1,5,2,11,10,6,7,8,9]\\
         \hline
         Triangular &[0,1,4,5,2,6,9,10] & [0,1,4,5,7,6,3,2,10,11]&  [0,1,4,5,2,6,8,9,11,10]& [0,4,1,5,2,6,3,7,8,9,10,11]\\
         \hline
         Hexagonal &   [0,1,4,3,7,2,9,8]  & [6,0,7,1,8,2,4,5,3,9]&[0,1,5,3,7,2,10,4,8,9]& [6,0,7,1,8,2,9,3,11,5,4,10]\\
         \hline
         Kagome & [0,1,4,3,6,2,8,7]  & [6,0,1,5,2,11,3,9,4,10]  &[0,1,5,3,6,2,9,4,7,8]& [0,1,11,2,5,3,10,4,6,7,8,9]\\
         \hline
         Rigetti & [1,3,0,2,5,4,11,10]  & [0,1,2,3,5,4,8,11,10,12]&[1,3,6,5,2,4,9,8,10,11]& [0,1,2,3,10,4,11,5,6,7,9,8]\\
         \hline
        \end{tabular}
    \caption{The optimal mapping results. The position in $[~\cdot~]$ is the index of physical qubits. The number inside $[~\cdot~]$ is the index of a vertex in the mapped graph as in \cref{tikz:hardware_graphs}. For example, for a $[[10,3,5,2]]$ code mapped to the square layout, the mapping result is $[0,1,3,4,6,5,9,10,7,8]$. The $6$ in the fifth entry means that the physical qubit in row $2$ and column $0$ is mapped to the vertex labeled $6$ in the square layout in \cref{tikz:hardware_graphs}.}
    \label{table:mapping_results}
\end{table*}

\begin{table*}[ht]
\centering
    \begin{tabular}{ |p{2.8cm}||p{1.1cm}|p{1.1cm}|p{1.1cm}|p{1.1cm}|p{1.1cm}|p{1.1cm}|p{1.1cm}|p{1.1cm}| }
     \hline
     \multicolumn{9}{|c|}{Graph properties from the mapping results} \\
     \hline
      Mapped graph & \multicolumn{2}{c|}{$[[8,3,3,2]]$} & \multicolumn{2}{c|}{$[[10,3,5,2]]$} & \multicolumn{2}{c|}{$[[10,4,4,2]]$} & \multicolumn{2}{c|}{$[[12,4,6,2]]$}\\
     \hline
      & total & avg &total & avg &total & avg &total & avg  \\
     \hline
     Union Jack   &  13 & 1   & 16 & 1.067  &  21 & 1.05  & 22 & 1.1\\
     \hline
     Triangular &  15 & 1.154  & 18 & 1.2  &25 & 1.25  & 24 & 1.2\\
     \hline
     Square &  17 & 1.308   & 20 & 1.333  &30 & 1.5  & 24 & 1.2 \\
     \hline
     Kagome & 18 & 1.385  &  21 & 1.4 &  29 & 1.45 & 30 &1.5\\
     \hline
     Hexagonal &  20 & 1.538   & 25 & 1.667 &  32 & 1.6 &34&1.7\\
     \hline
     Rigetti Aspen &  20 & 1.538  &  24 & 1.6 & 34& 1.7 & 36 & 1.8 \\
     \hline
     IBM Heavy-hex  &23 &1.769  & 30 & 2 & 41 & 2.05 & 44 & 2.2\\
     \hline
    \end{tabular}
    \caption{Graph properties from the mapping results. The ``total" columns show the total Manhattan distances, the ``avg'' shows the average Manhattan distances. The mapped graphs appear in order of performance ranking, with the Union Jack graph exhibiting the lowest (i.e., best) total and average Manhattan distance.}
\label{table:mapping_property}
\end{table*}

In \cref{table:mapping_results}, the position in $[~\cdot~]$ represents the index of the physical qubits. We label the physical qubits from left to right and top to bottom, incrementing the index by one at each step. The number inside $[~\cdot~]$ corresponds to the index of a vertex in the mapped graph. For example, for the $[[10,3,5,2]]$ code, when we map it to the square layout, the resulting vector representing the mapping is $[0,1,3,4,6,5,9,10,7,8]$, and the $6$ in the fifth entry of the vector means the physical qubit indexed as $5$ is mapped to the vertex labeled $6$ in the square layout in \cref{tikz:hardware_graphs}.

The properties of the mapped graphs are presented in \cref{table:mapping_property}. The first column lists the total Manhattan distances, while the second column provides the average Manhattan distances. Based on the criteria discussed in \cref{sec:discussion} below, the Union Jack graph outperforms all other layouts across all seven graphs, which can be attributed to its high degree and balanced structure. Despite the fact that the example is relatively small and the induced graph's degree aligns well with the Union Jack layout, larger examples may exhibit slight variations. However, due to its higher degree, the Union Jack graph is better suited for accommodating more complex graphs. At the other end of the spectrum, the heavy-hex graph performs the worst, due to its low degree and lack of balance. 

\section{Examples}
\label{sec:example}

We use the $[[14,6,6,2]]$ and $[[12,4,6,2]]$ codes as examples to provide a step-by-step explanation of the entire pipeline, which includes obtaining the $2$-local logical operators from \cref{eq:log-op-a} using gauge cancellations, constructing the induced graph $G$ based on the obtained logical operators and gauge generators, and mapping the induced graph $G$ to a mapped graph $G_{\mathrm{m}}$.

\subsection{$[[14,6,6,2]]$ code} 

The $[[14,6,6,2]]$ is the $l=1$, $k=3$ case of the $[[4k+2,2k,2k+2l-2,2]]$ subfamily of trapezoid codes. Recall that $l=1$ yields the largest penalty gap and the highest code rate, so this is an example of particular interest.

\subsubsection{Logical operators}\label{subsec:example-logi}
Each entry $1$ in \cref{eq:odd-m-A} represents a physical qubit. These qubits are indexed from left to right and from top to bottom in \cref{eq:A-matrix-number-1} as follows:
\begin{align}
\label{eq:A-matrix-number-1}
&\left[\begin{array}{ccccccc}
      1&2&&&&&\\
    3&&4&&&&\\
    &5&&6&&&\\
    &&7&&8&&\\
    &&&9&&10&\\
    &&&&11&&12\\
    &&&&&13&14\\
\end{array}
   \right].
\end{align}

Logical operators can be represented by the operator matrices in \cref{eq:log-op-a}. Each row is the bitstring representation of a $\hat{X}$ or $\hat{Z}$. A $1$ in the $i$'th position of the bitstring indicates that the $i$'th row or column in \cref{eq:A-matrix-number-1} is involved in the corresponding logical operator. For the $[[14,6,6,2]]$ code, the logical operators are as follows:
\begin{align}
    \hat{X}_1 &= 0101010 = X_2X_5X_6X_9X_{10}X_{13} = X_2X_{13}, \notag\\
    \hat{X}_2 &= 0010101 = X_4X_7X_8X_{11}X_{12}X_{14} = X_4X_{14}, \notag\\
    \hat{X}_3 &= 0001010 = X_6X_9X_{10}X_{13} = X_6X_{13}, \notag\\
    \hat{X}_4 &= 0000101 = X_8X_{11}X_{12}X_{14} = X_8X_{14}, \notag\\
    \hat{X}_5 &= 0000010 = X_{10}X_{13}, \notag\\
    \hat{X}_6 &= 0000001 = X_{12}X_{14}, \notag\\
    \hat{Z}_1 &= 1000000 = Z_1Z_2,\notag\\
    \hat{Z}_2 &= 0100000 = Z_3Z_4,\notag\\
    \hat{Z}_3 &= 1010000 = Z_1Z_2Z_5Z_6 = Z_1Z_6,\notag\\
    \hat{Z}_4 &= 0101000 = Z_3Z_4Z_7Z_8 = Z_3Z_8,\notag\\
    \hat{Z}_5 &= 1010100 = Z_1Z_2Z_5Z_6Z_9Z_{10} = Z_1Z_{10},\notag\\
    \hat{Z}_6 &= 0101010 = Z_3Z_4Z_7Z_8Z_{11}Z_{12} = Z_3Z_{12}.
\end{align}

\begin{equation}\label{tikz:exp-map}
\begin{aligned}
    &\begin{tikzpicture}
    \matrix (m)[
    matrix of math nodes,
    nodes in empty cells,
    left delimiter=\lbrack,
    right delimiter=\rbrack
    ] {
     1&2&&&&&\\
    3&&4&&&&\\
    &5&&6&&&\\
    &&7&&8&&\\
    &&&9&&10&\\
    &&&&11&&12\\
    &&&&&13&14\\
    } ;
   \begin{scope}[fill=green!80!white,fill opacity=0.1,draw=green,thick]
    \filldraw (m-1-2)  circle (10pt);
    \filldraw (m-7-6)  circle (10pt);
   \end{scope}
   \begin{scope}[fill=red!80!white,fill opacity=0.1,draw=red,thick]
    \filldraw (m-2-3)  circle (10pt);
    \filldraw (m-7-7)  circle (10pt);
   \end{scope}
   \begin{scope}[fill=orange!80!white,fill opacity=0.1,draw=orange,thick]
    \filldraw (m-3-4)  circle (8pt);
    \filldraw (m-7-6)  circle (8pt);
   \end{scope}
   \begin{scope}[fill=blue!80!white,fill opacity=0.1,draw=blue,thick]
    \filldraw (m-4-5)  circle (8pt);
    \filldraw (m-7-7)  circle (8pt);
   \end{scope}
   \begin{scope}[fill=black!80!white,fill opacity=0.1,draw=black,thick]
    \filldraw (m-5-6)  circle (6pt);
    \filldraw (m-7-6)  circle (6pt);
   \end{scope}
   \begin{scope}[fill=brown!80!white,fill opacity=0.1,draw=brown,thick]
    \filldraw (m-6-7)  circle (6pt);
    \filldraw (m-7-7)  circle (6pt);
   \end{scope}
    \end{tikzpicture}
    \\
    &\begin{tikzpicture}
    \matrix (m)[
    matrix of math nodes,
    nodes in empty cells,
    left delimiter=\lbrack,
    right delimiter=\rbrack
    ] {
     1&2&&&&&\\
    3&&4&&&&\\
    &5&&6&&&\\
    &&7&&8&&\\
    &&&9&&10&\\
    &&&&11&&12\\
    &&&&&13&14\\
    } ;
   \begin{scope}[fill=green!80!white,fill opacity=0.1,draw=green,thick]
    \filldraw (m-1-1)  circle (10pt);
    \filldraw (m-1-2)  circle (10pt);
   \end{scope}
   \begin{scope}[fill=red!80!white,fill opacity=0.1,draw=red,thick]
    \filldraw (m-2-1)  circle (10pt);
    \filldraw (m-2-3)  circle (10pt);
   \end{scope}
   \begin{scope}[fill=orange!80!white,fill opacity=0.1,draw=orange,thick]
    \filldraw (m-1-1)  circle (8pt);
    \filldraw (m-3-4)  circle (8pt);
   \end{scope}
   \begin{scope}[fill=blue!80!white,fill opacity=0.1,draw=blue,thick]
    \filldraw (m-2-1)  circle (8pt);
    \filldraw (m-4-5)  circle (8pt);
   \end{scope}
   \begin{scope}[fill=black!80!white,fill opacity=0.1,draw=black,thick]
    \filldraw (m-1-1)  circle (6pt);
    \filldraw (m-5-6)  circle (6pt);
   \end{scope}
   \begin{scope}[fill=brown!80!white,fill opacity=0.1,draw=brown,thick]
    \filldraw (m-2-1)  circle (6pt);
    \filldraw (m-6-7)  circle (6pt);
   \end{scope}
    \end{tikzpicture}
\end{aligned}
\end{equation} 

The upper matrix in~\cref{tikz:exp-map} illustrates the connectivity for $\hat{X}$, while the lower matrix corresponds to $\hat{Z}$. Different colors are used to represent distinct logical qubits. Circles of the same color in both matrices represent logical operators associated with the same logical qubit $(\hat{X}_i,\hat{Z}_i)$. For instance, the green circles denote physical operators required to execute $\hat{X}_1$ and $\hat{Z}_1$. Note that circles of the same color from the left and right matrices overlap exactly once, and circles of different colors do not overlap. This arrangement reflects the commutation relation of the logical operators, i.e., $\{\hat{X}_i,\hat{Z}_i\}=0$ and $[\hat{X}_i,\hat{Z}_j]=0$ for $i\ne j$.

To construct $\hat{X}\hat{X}$, we multiply two $\hat{X}$ operators. For example, executing $\hat{X}_1\hat{X}_2$ corresponds to $X_2X_{13}X_4X_{14}$ (represented by two large red and two large green circles). Since $X_{13}X_{14}$ is a gauge we can discard it, allowing us to execute $\hat{X}_1\hat{X}_2$ using just the $2$-local operator $X_2X_4$ [recall \cref{lem:2_local_logical}]. Following similar reasoning, we can derive the following two-body logical operators:
\begin{align}\label{eq:example_XX_hat}
    \hat{X}_1\hat{X}_2 &= X_2X_4,\hat{X}_1\hat{X}_3 = X_2X_6,\hat{X}_1\hat{X}_4 = X_2X_8, \notag\\
    \hat{X}_1\hat{X}_5 &= X_2X_{10},\hat{X}_1\hat{X}_5 = X_2X_{12},\hat{X}_2\hat{X}_3 = X_4X_{6}, \notag\\
    \hat{X}_2\hat{X}_4 &= X_4X_{8},\hat{X}_2\hat{X}_5 = X_4X_{10},\hat{X}_2\hat{X}_6 = X_4X_{12}, \notag\\
    \hat{X}_3\hat{X}_4 &= X_6X_{8},\hat{X}_3\hat{X}_5 = X_6X_{10},\hat{X}_3\hat{X}_6 = X_6X_{12},\notag\\
    \hat{X}_4\hat{X}_5 &= X_8X_{10},\hat{X}_4\hat{X}_6 = X_8X_{12},\hat{X}_5\hat{X}_6 = X_{10}X_{12}.
\end{align}
The $\hat{Z}\hat{Z}$ operators exhibit a similar structure; we simply replace $X$ with $Z$, noting that $Z_1Z_2$ is also a gauge operator. 

\subsubsection{Induced graph $G$}
To implement the universal encoded Hamiltonian~\cref{eq:H-hat-b} on hardware, we need to ensure that the physical qubits supporting all the aforementioned logical operators and gauge operators are coupled. 

We begin by examining the connectivity induced by $\hat{X}$ and $\hat{Z}$. As illustrated in \cref{tikz:exp-map}, for both $\hat{X}$ and $\hat{Z}$, half of the qubits on the superdiagonal must be coupled to qubit $13$ and qubit $1$, while the other half must connect to qubit $14$ and qubit $3$. The resulting connectivity is depicted in the top left panel of \cref{fig:induced-steps}. 

Next, for $\hat{X}\hat{X}$ and $\hat{Z}\hat{Z}$, we observe that they lead to the same connectivity. As shown in \cref{eq:example_XX_hat}, the qubits on the superdiagonal must be coupled. This results in a fully connected graph with $6$ vertices, as illustrated in the top middle panel of \cref{fig:induced-steps}.

Finally, since the penalty Hamiltonian consists of gauge generators, we need to incorporate additional qubits and edges. Specifically, we need to connect qubit $1$ to qubit $3$ and qubit $13$ to qubit $14$. We then add the qubits on the lower diagonal, ensuring that each connects to two other qubits, one in the same row and the other in the same column. The result is shown in the top right panel of \cref{fig:induced-steps}.

Thus, we have the following edge set:
\begin{align}\label{eq:edge-set-example}
    E = \{&(2,13),(4,14),(6,13),(8,14),(10,13), \notag\\
    &(12,14),(1,2),(3,4),(1,6),(3,8),(1,10), \notag\\
    &(3,12),(2,4),(2,6),(2,8),(2,10),(2,12),\notag\\
    &(4,6),(4,8),(4,10),(4,12),(6,8),(6,10),\notag\\
    &(6,12),(8,10),(8,12),(10,12),(1,3),\notag\\
    &(2,5),(4,7),(6,9),(8,11),(5,6),(7,8),\notag\\
    &(9,10),(11,12),(13,14))\}.
\end{align}

\begin{figure*}[ht]
    \centering
\includegraphics[width=0.32\linewidth]{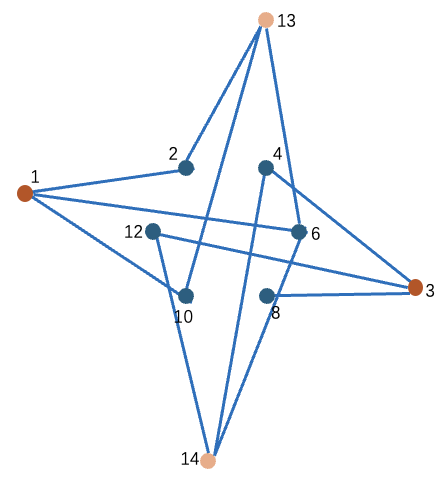}
\includegraphics[width=0.32\linewidth]{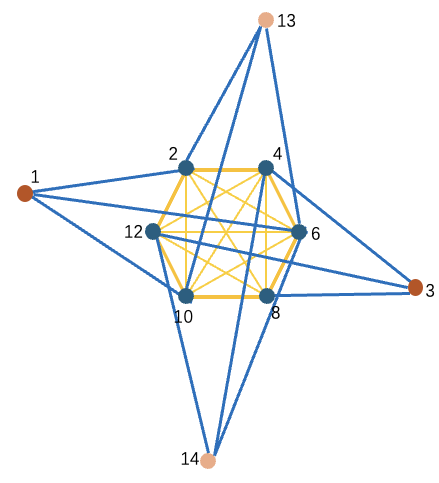}
\includegraphics[width=0.32\linewidth]{figs/induced_l1.png}
\includegraphics[width=0.32\linewidth]{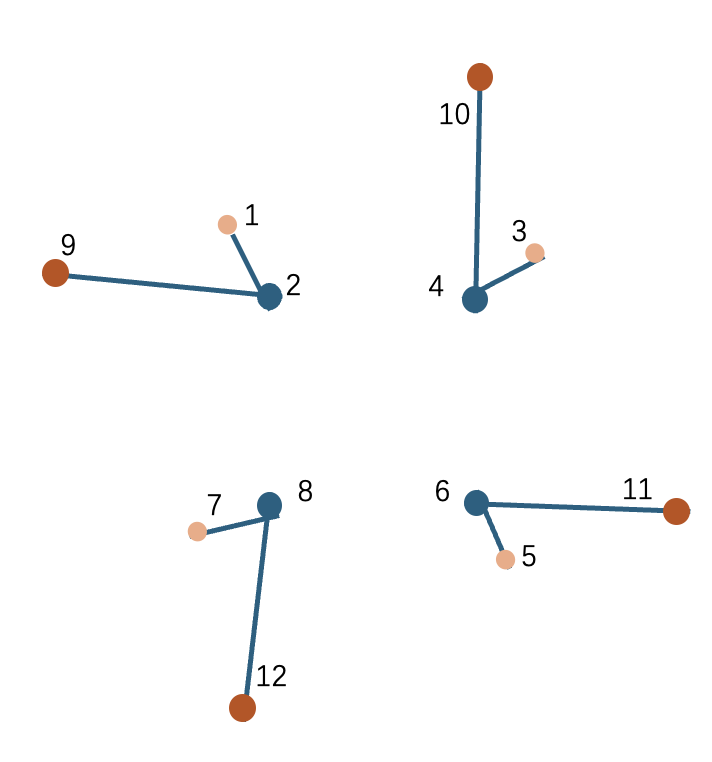}
\includegraphics[width=0.32\linewidth]{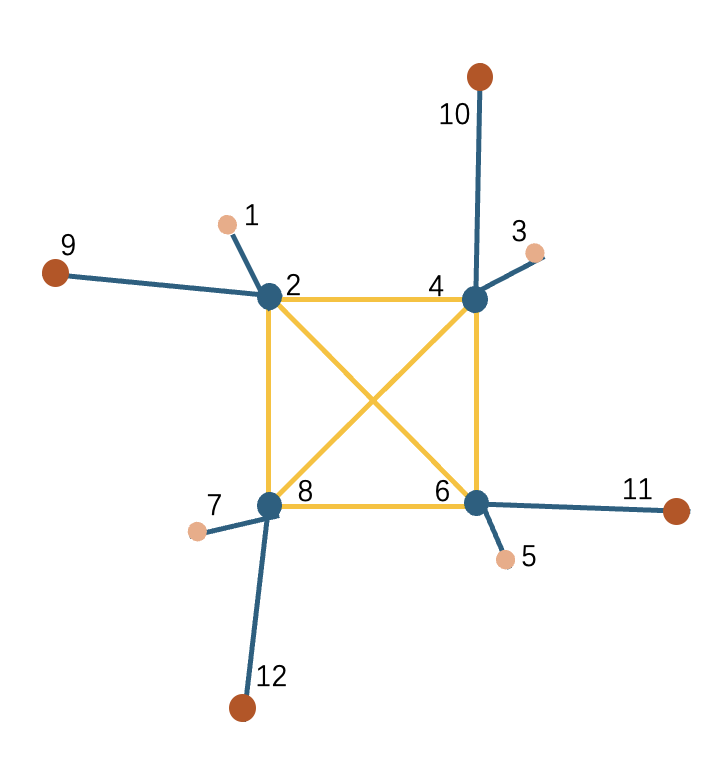}
\includegraphics[width=0.32\linewidth]{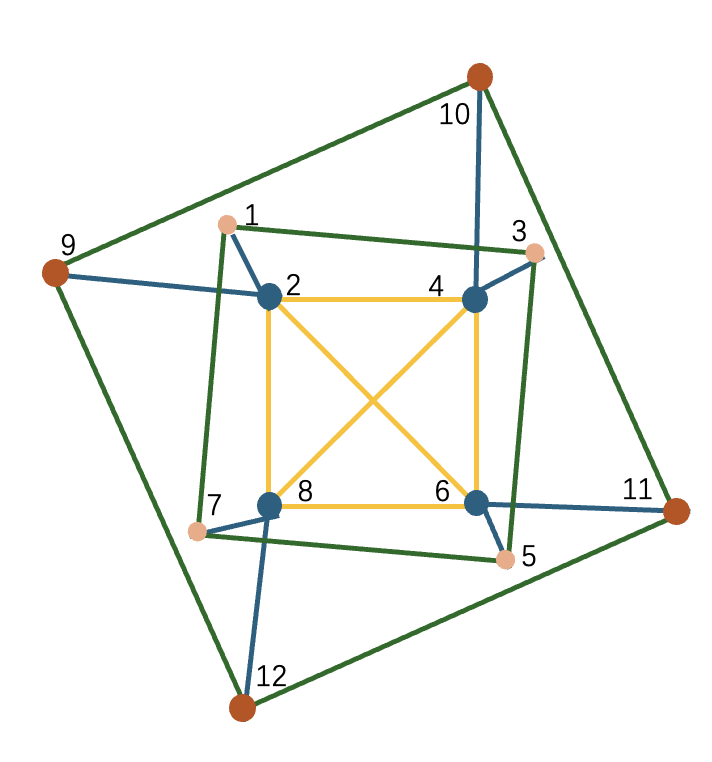}
     \caption{Top row: steps to construct the induced graph $G$ for the $[[14,6,6,2]]$ code. The left panel shows the connectivity required for $\hat{X}$ and $\hat{Z}$, the middle panel for $\hat{X}\hat{X}$ and $\hat{Z}\hat{Z}$, and the right panel for the gauge generators. Bottom row: the same for the $[[12,4,6,2]]$ code.}
     \label{fig:induced-steps}
\end{figure*}

\subsubsection{Mapped graph $G_{\mathrm{m}}$ with Union Jack}

The induced graph (\cref{fig:induced-steps}, right) does not align with any existing hardware connectivity graphs. To illustrate the mapping to a more feasible graph, we select the Union Jack as the target graph and demonstrate how to interpret the results.

The Union Jack layout is shown in the left panel of \cref{tikz:example_mapped_union_jack}. Given the layout, we can obtain the edge set as:
\begin{align}
    E_{\mathrm{m}} = \{&(0,1),(0,4),(0,5),(1,2),(1,4),(1,5),\notag\\
    &(1,6),(2,3),(2,5),(2,6),(2,7),(3,6),\notag\\
    &(3,7),(4,5),(4,8),(4,9),(5,6),(5,8),\notag\\
    &(5,9),(5,10),(6,7),(6,9),(6,10),(6,11),\notag\\
    &(7,10),(7,11),(8,9),(8,12),(8,13),\notag\\
    &(9,10),(9,12),(9,13),(9,14),(10,11),\notag\\
    &(10,13),(10,14),(10,15),(11,14),\notag\\
    &(11,15),(12,13),(13,14),(14,15)\}.
\end{align}
The mapping results are presented by the vector $[8, 6, 13, 14, 1, 5, 15, 10, 0, 4, 11, 9, 2, 7]$, which is of the same format as in \cref{table:mapping_results}. As explained in \cref{subsec:brute-map}, the position $i$ in the vector denotes the index of the physical qubit in the induced graph, while the value of the $i$'th entry denotes the location of that qubit on the mapped graph. Thus, we have the following correspondence:
\begin{widetext}
\begin{center}
\begin{tabular}{ c c c c c c c c c c c c c c}
\centering
 8 & 6 & 13 & 14 & 1 & 5 & 15 & 10 & 0 & 4 & 11 & 9 & 2 & 7 \\ 
 \vline & \vline & \vline & \vline & \vline & \vline & \vline & \vline & \vline & \vline & \vline & \vline & \vline & \vline \\ 
  \textcolor{red}{1} & \textcolor{red}{2} & \textcolor{red}{3} & \textcolor{red}{4} & \textcolor{red}{5} & \textcolor{red}{6} & \textcolor{red}{7} & \textcolor{red}{8} & \textcolor{red}{9} & \textcolor{red}{10} & \textcolor{red}{11} & \textcolor{red}{12} & \textcolor{red}{13}    & \textcolor{red}{14}
\end{tabular}
\end{center}
\end{widetext}
The first row is the mapping result vector, which contains the position of the physical qubits, while the second row shows the indices of physical qubits. We can then label the physical qubits on the mapped graph, as shown in the right panel of \cref{tikz:example_mapped_union_jack}.

\begin{figure*}[ht!]
    \centering

        \includegraphics[width=0.26\linewidth]{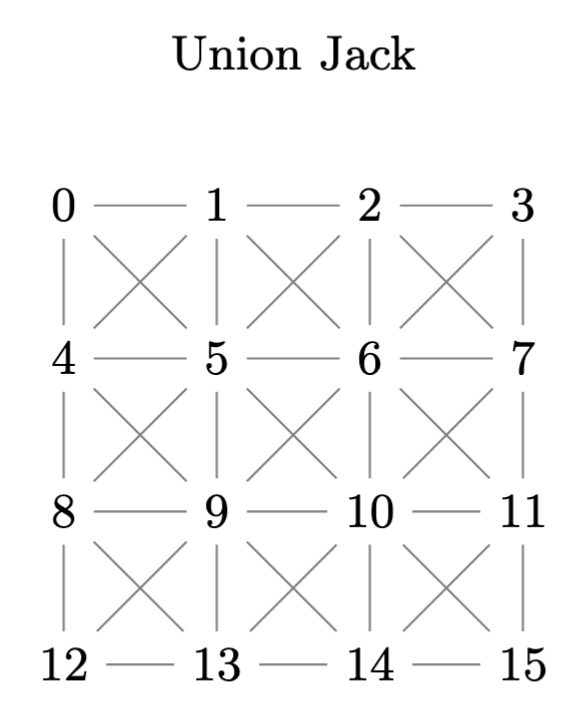}
\includegraphics[width=0.26\linewidth]{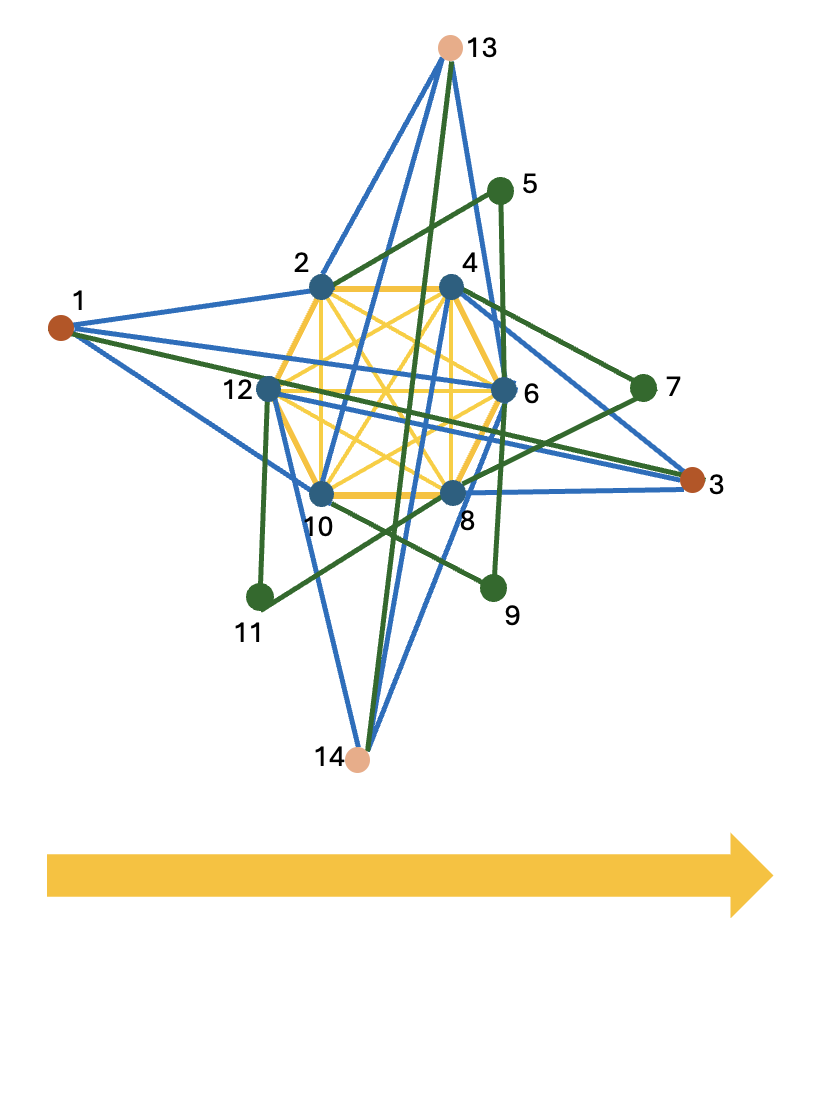}
 \includegraphics[width=0.26\linewidth]{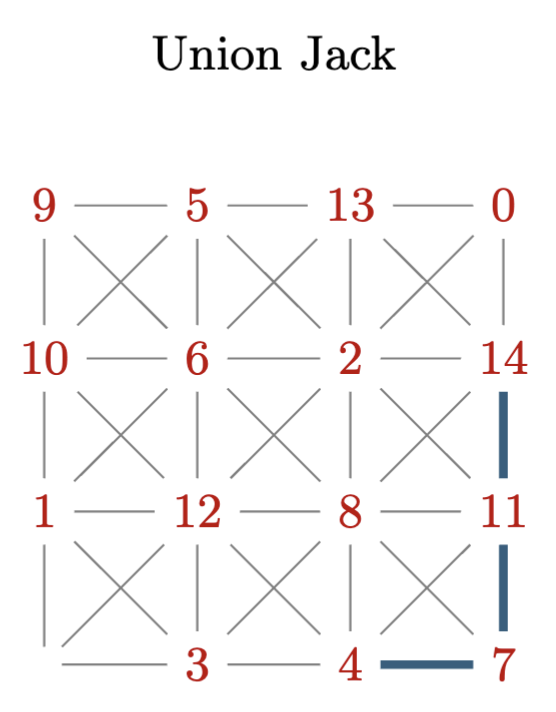}
    \caption{Interpretation of the mapping result for the $[[14,6,6,2]]$ code. The mapping result is $[8,6,13,14,1,5,15,10,0,4,11,9,2,7]$. The red numbers represent the indices of the physical qubits, corresponding to the indices in the left panel, while the black numbers represent the positions of these physical qubits in the mapped graph. Blue edges denote the shortest path to connect qubits $4$ and $14$ in $G_{\mathrm{m}}$.}
    \label{tikz:example_mapped_union_jack}
\end{figure*}

To illustrate how to calculate the Manhattan distance, consider, e.g., b $\hat{X}_2 = X_4X_{14}$. Qubit $4$ is mapped to vertex $14$, and qubit $14$ is mapped to vertex $7$. The induced map contains the edge $4-14$, which support $\hat{X}_2$. In the mapped graph, the Manhattan distance between vertices $7$ and $14$ is used as the metric for the geometric locality of qubits $4$ and $14$ on $G_{\mathrm{m}}$. One of the shortest paths from $7$ to $14$ is $7-11-15-14$, resulting in $m(7,14)=3$. To compute the total Manhattan distance of $G_{\mathrm{m}}$ given $E$, this process is repeated for all the edges $e$ in the set $E$, and the distances are summed.

\subsection{$[[12,4,6,2]]$ code}

We use the $[[12,4,6,2]]$ code as another example to illustrate a different case when $l=k$. We focus on the induced graph $G$ and the mapped graph $G_{\mathrm{m}}$ in this example. The physical qubits in the $A$ matrix for the $[[12,4,6,2]]$ code are indexed as follows:
\begin{align}
\label{eq:A-matrix-number-2}
&\left[\begin{array}{ccccc}
      1&2&&&\\
    3&&4&&\\
    5&&&6&\\
    7&&&&8\\
    &9&10&11&12
\end{array}
   \right].
\end{align}

Following the same steps as \cref{subsec:example-logi}, we obtain the logical operators as:
\begin{align}
    &\hat{X}_1 = X_2X_9, \hat{X}_2 = X_4X_{10},\hat{X}_3 = X_6X_{11}, \notag\\
    &\hat{X}_4 = X_8X_{12}, \hat{Z}_1 = Z_1Z_2, \hat{Z}_2 = Z_3Z_4, \notag\\
    &\hat{Z}_3 = Z_5Z_6, \hat{Z}_4 = Z_7Z_8, \notag \\ 
    &\hat{X}_1\hat{X}_2 = X_2X_4, \hat{X}_1\hat{X}_3 = X_2X_6, \hat{X}_1\hat{X}_4=X_2X_8,\notag\\
    &\hat{X}_2\hat{X}_3 = X_4X_6, \hat{X}_2\hat{X}_4 = X_4X_8, \hat{X}_3\hat{X}_4 = X_6X_8, \notag\\
    &\hat{Z}_1\hat{Z}_2 = Z_2Z_4, \hat{Z}_1\hat{Z}_3 = Z_2Z_6, \hat{Z}_1\hat{Z}_4=Z_2Z_8,\notag\\
    &\hat{Z}_2\hat{Z}_3 = Z_4Z_6, \hat{Z}_2\hat{Z}_4 = Z_4Z_8, \hat{Z}_3\hat{Z}_4 = Z_6Z_8.
\end{align}

In the induced graph $G$, we need to couple pairs of qubits that are necessary to support logical operators and gauge operators. For $\hat{X}$ and $\hat{Z}$, a qubit on the superdiagonal must be coupled to a qubit in the first column and a qubit in the last row, which results in the connectivity shown in the bottom left panel of \cref{fig:induced-steps}. For $\hat{X}\hat{X}$ and $\hat{Z}\hat{Z}$, the qubits on the superdiagonal must be coupled, which leads to a fully connected graph with $4$ vertices, as illustrated in the bottom middle panel of \cref{fig:induced-steps}. Finally, the qubits in the first column and the last row form two cyclic squares to support the gauge operators, as shown in the bottom right panel of \cref{fig:induced-steps}. Thus, the edge set is:
\begin{align}\label{eq:edge-set-example-2}
    E = \{&(2,9),(4,10),(6,11),(8,12),(1,2),(3,4), \notag\\
    &(5,6),(7,8),(2,4),(2,6),(2,8),(4,6),\notag\\
    &(4,8),(6,8),(1,3),(3,5),(5,7),(7,1),\notag\\
    &(9,10),(10,11),(11,12),(12,9)\}.
\end{align}

This $G$ is of low degree and is balanced, which is implementable on current devices. However, to demonstrate the results of larger system size, we map it to a kagome layout. The kagome layout is shown in the left panel of \cref{tikz:example_mapped_kagome}. Given the layout, we can obtain the edge set as:
\begin{align}
    E_{\mathrm{m}} = \{&(0,1),(1,2),(2,3),(3,4),(4,5),\notag\\
    &(5,0),(0,6),(1,6),(1,7),(2,7),\notag\\
    &(2,8),(3,8),(3,9),(4,9),(4,10),\notag\\
    &(5,10),(5,11),(0,11)\}.
\end{align}
The optimal mapping result is give by the vector $[0, 1, 11, 2, 5, 3, 10, 4, 6, 7, 8, 9]$. Thus, we have the following correspondence:
\begin{center}
\begin{tabular}{ c c c c c c c c c c c c}
\centering
 0 & 1 & 11 & 2 & 5 & 3 & 10 & 4 & 6 & 7 & 8 & 9 \\ 
 \vline & \vline & \vline & \vline & \vline & \vline & \vline & \vline & \vline & \vline & \vline & \vline \\ 
  \textcolor{red}{1} & \textcolor{red}{2} & \textcolor{red}{3} & \textcolor{red}{4} & \textcolor{red}{5} & \textcolor{red}{6} & \textcolor{red}{7} & \textcolor{red}{8} & \textcolor{red}{9} & \textcolor{red}{10} & \textcolor{red}{11} & \textcolor{red}{12} 
\end{tabular}
\end{center}

\begin{figure*}[ht!]
    \centering
\includegraphics[width=0.26\linewidth]{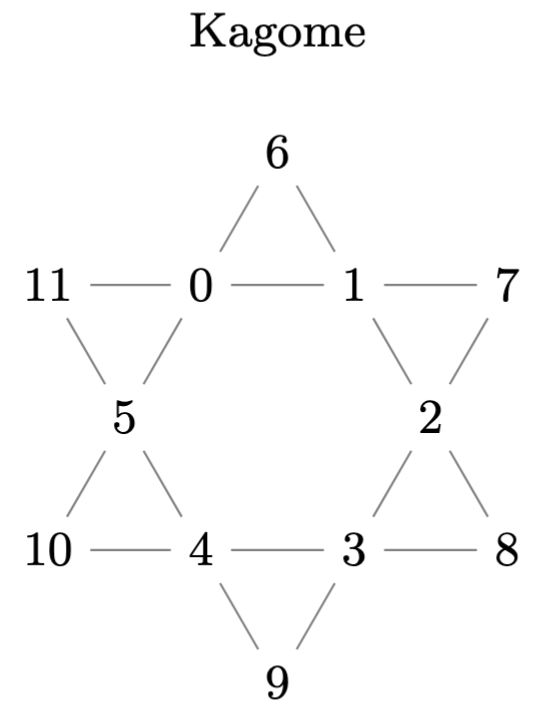}
        \includegraphics[width=0.26\linewidth]{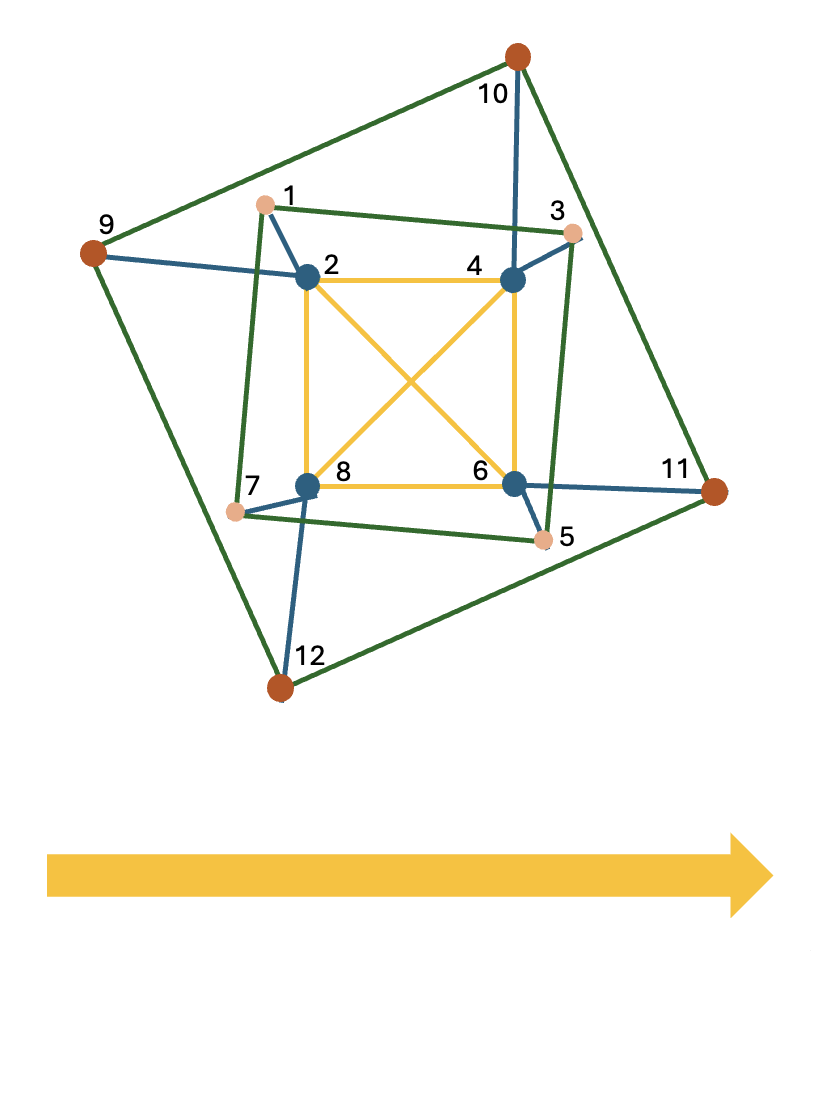}
   \includegraphics[width=0.26\linewidth]{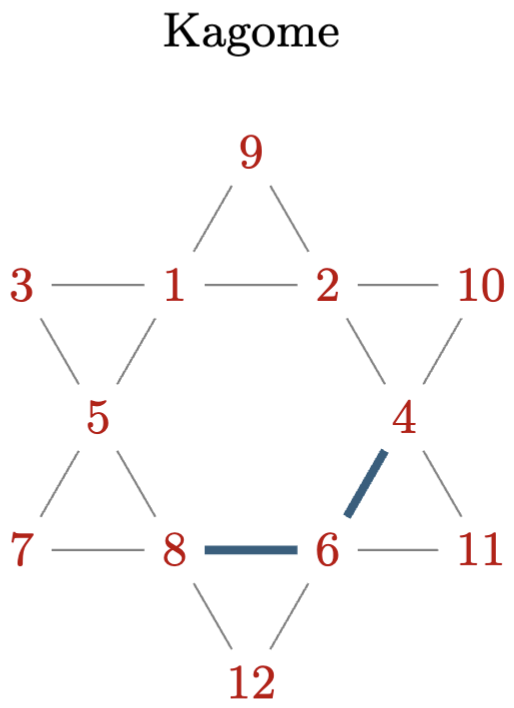}
    \caption{The mapping results. The mapping result is $[0,1,11,2,5,3,10,4,6,7,8,9]$. The red numbers represent the indices of the physical qubits, corresponding to the indices in the left panel, while the black numbers represent the positions of these physical qubits in the mapped graph. Blue edges denote the shortest path to connect qubits $4$ and $8$ in $G_{\mathrm{m}}$.}
    \label{tikz:example_mapped_kagome}
\end{figure*}

To illustrate how to calculate the Manhattan distance, consider, e.g., $\hat{X}_2\hat{X}_4 = X_4X_{8}$. Qubit $4$ is mapped to vertex $2$, and qubit $8$ is mapped to vertex $4$. The induced map contains the edge $4-8$, which support $\hat{X}_2\hat{X}_4$. One of the shortest paths from $2$ to $4$ on $G_{m}$ is $2-3-4$, resulting in $m(2,4)=2$. 

\section{Discussion}
\label{sec:discussion}

We now revisit the five criteria we outlined in \cref{sec:desiderata} and discuss them in light of our results.

\begin{enumerate}

\item Code rate:
For $[[4k+2l,2k,g,2]]$ trapezoid codes, we aim to maximize the code rate $k/(2k+l)$. We showed in \cref{sec:props-rate} that codes with $l=1$ have the maximum code rate of $r=k/(2k+1)$.  

\item Physical Locality:
With our choice of optimal dressed logical operators, all the logical operators required are $2$-body, as detailed in \cref{lem:2_local_logical}. We have also demonstrated that bare logical operators cannot all be $2$-local in \cref{lem:bare-not-2}. 

\item Induced graph properties:
The induced graph is $2$-local. However, the degree of this graph grows linearly with the number of logical qubits, complicating implementation. Ideally, every logical operator and gauge is geometrically $2$-local in a fixed degree mapped graph. However, none of the 2D geometries we explored satisfies this ideal for $l<k$ codes. Codes with $l=k$ ($[[6k,2k,2]]$ codes) inherently have almost-constant degree induced graphs (if we ignore the fully connected component needed for logical $XX$ and $ZZ$ interactions), as shown in~\cite{marvian2019robust}. 
Instead, we use the total and average Manhattan distances as metrics to measure the geometrical locality of the code. The total Manhattan distance can be interpreted as the number of SWAPs that would be needed in the circuit model if we were to apply all dressed logical operators exactly once. 

\item Mapped graph properties:
We provide an explicit mapping scheme to feasible hardware graphs in \cref{sec:optimiza} and illustrate it in detail in \cref{sec:example}. Another approach involves constructing a possibly optimal new graph based on the induced graph with constraints such as having a constant degree, but this method is harder to implement.

\item Penalty Gap:
We numerically calculate the penalty gaps of a subset of trapezoid codes (\cref{fig:gap}). Fixing $m$, for $l=1,2$ the gap closes as a power law; for $l=1$ the penalty Hamiltonian is equivalent to the 1D compass model, for which the gap is known to close as the inverse of the number of (logical) qubits, which is optimal. For $l>2$ the gap closes exponentially in $l$. It also closes exponentially in $m$ at fixed $l$. Recall that the number of logical qubits is $k=\lfloor m/2 \rfloor$ and the number of physical qubits is $4k+2l$ for odd $m$ or $4k+2(l-1)$ for even $m$. Thus, codes with smaller $l$ provide exponentially larger penalty gaps in the number of physical qubits. Trapezoid codes with $l<k$ therefore have larger penalty gaps than the $[[6k,2k,2]]$ code family in Ref.~\cite{marvian2019robust}.

\end{enumerate}

\section{Conclusions}
\label{sec:conclusion}

Analog, Hamiltonian quantum computation can be effectively protected against decoherence via quantum error detection codes implementing energy penalties. Our goal in this work was to identify a family of codes that accomplish this goal subject to common constraints imposed by hardware limitations. This includes no higher than $2$-local interactions, geometric locality, and compatibility with existing qubit connectivity graphs. It is also crucial to maintain as large an energy gap as possible above the ground subspace of the penalty Hamiltonian, in order to prevent deviations from the intended quantum evolution. 

In light of these goals, we studied subsystem codes derived from Bravyi's $A$ matrices, which are naturally suitable for 2D qubit interaction graphs. This led to the identification of a code family we call `trapezoid codes,' which exhibit favorable characteristics compatible with our criteria.  We examined the trade-offs associated with various factors affecting the practical implementation of trapezoid codes, as detailed in \cref{sec:discussion}. In particular, we demonstrated that every code within the trapezoid family possesses all $2$-local one-qubit logical operators and two-qubit dressed logical operators. 

We found that the optimal code rate achievable within this family is $k/(2k+1)$, a rate manifested by the $[[4k+2,2k,g,2]]$ code for any positive integer $k$ (here $g=2k+2l-2$). This code can be mapped to feasible hardware connectivity graphs, as discussed in \cref{sec:optimiza}. Moreover, the $[[4k+2,2k,g,2]]$ code also has the largest penalty gap among the codes in the trapezoid family, as can be seen in \cref{fig:gap}, and this gap scales as $1/(2k)$, ensuring a slower scaling with problem size $k$ than for most problems of interest in AQC.

We used the total Manhattan distance as a performance metric to evaluate the geometric locality of the mapping results and examined a set of common  native hardware graph geometries. Within this set, we found that the Union Jack lattice exhibits the lowest Manhattan distance overall (\cref{table:mapping_property}). 

In conclusion, the $[[4k+2,2k,g,2]]$ code implemented on a native Union Jack lattice presents an optimal combination of favorable properties, and appears to be a promising choice for future implementations.

Our framework is flexible enough to allow each code within the trapezoid family to be evaluated based on other criteria. We hope that the results reported here will contribute to the advancement of error-suppressed, analog Hamiltonian quantum computing.

\acknowledgments
This material is based upon work supported by, or in part by, the U. S. Army Research Laboratory and the U.S. Army Research Office under contract/grant number W911NF2310255, and by the Office of the Director of National Intelligence (ODNI), Intelligence Advanced Research Projects Activity (IARPA) and the Army Research Office, under the Entangled Logical Qubits program through Cooperative Agreement Number W911NF-23-2-0216. The views and conclusions contained in this document are those of the authors and should not be interpreted as representing the official policies, either expressed or implied, of IARPA, the Army Research Office, or the U.S. Government. The U.S. Government is authorized to reproduce and distribute reprints for Government purposes notwithstanding any copyright notation herein.
We thank Seolhwa Kim, Yanlin Cheng, Adam Bistagne, and Larry Gu for providing an early analysis that inspired some of the results presented here. 

\bibliographystyle{apsrev4-2-titles}
\bibliography{refs}

\begin{thebibliography}{58}%
\makeatletter
\providecommand \@ifxundefined [1]{%
 \@ifx{#1\undefined}
}%
\providecommand \@ifnum [1]{%
 \ifnum #1\expandafter \@firstoftwo
 \else \expandafter \@secondoftwo
 \fi
}%
\providecommand \@ifx [1]{%
 \ifx #1\expandafter \@firstoftwo
 \else \expandafter \@secondoftwo
 \fi
}%
\providecommand \natexlab [1]{#1}%
\providecommand \enquote  [1]{``#1''}%
\providecommand \bibnamefont  [1]{#1}%
\providecommand \bibfnamefont [1]{#1}%
\providecommand \citenamefont [1]{#1}%
\providecommand \href@noop [0]{\@secondoftwo}%
\providecommand \href [0]{\begingroup \@sanitize@url \@href}%
\providecommand \@href[1]{\@@startlink{#1}\@@href}%
\providecommand \@@href[1]{\endgroup#1\@@endlink}%
\providecommand \@sanitize@url [0]{\catcode `\\12\catcode `\$12\catcode
  `\&12\catcode `\#12\catcode `\^12\catcode `\_12\catcode `\%12\relax}%
\providecommand \@@startlink[1]{}%
\providecommand \@@endlink[0]{}%
\providecommand \url  [0]{\begingroup\@sanitize@url \@url }%
\providecommand \@url [1]{\endgroup\@href {#1}{\urlprefix }}%
\providecommand \urlprefix  [0]{URL }%
\providecommand \Eprint [0]{\href }%
\providecommand \doibase [0]{https://doi.org/}%
\providecommand \selectlanguage [0]{\@gobble}%
\providecommand \bibinfo  [0]{\@secondoftwo}%
\providecommand \bibfield  [0]{\@secondoftwo}%
\providecommand \translation [1]{[#1]}%
\providecommand \BibitemOpen [0]{}%
\providecommand \bibitemStop [0]{}%
\providecommand \bibitemNoStop [0]{.\EOS\space}%
\providecommand \EOS [0]{\spacefactor3000\relax}%
\providecommand \BibitemShut  [1]{\csname bibitem#1\endcsname}%
\let\auto@bib@innerbib\@empty
\bibitem [{\citenamefont {Farhi}\ \emph {et~al.}(2000)\citenamefont {Farhi},
  \citenamefont {Goldstone}, \citenamefont {Gutmann},\ and\ \citenamefont
  {Sipser}}]{Farhi:00}%
  \BibitemOpen
  \bibfield  {author} {\bibinfo {author} {\bibfnamefont {E.}~\bibnamefont
  {Farhi}}, \bibinfo {author} {\bibfnamefont {J.}~\bibnamefont {Goldstone}},
  \bibinfo {author} {\bibfnamefont {S.}~\bibnamefont {Gutmann}},\ and\ \bibinfo
  {author} {\bibfnamefont {M.}~\bibnamefont {Sipser}},\ }\bibfield  {title}
  {\emph {\bibinfo {title} {Quantum {Computation} by {Adiabatic}
  {Evolution}}},\ }\href {http://arxiv.org/abs/quant-ph/0001106} {\bibfield
  {journal} {\bibinfo  {journal} {arXiv:quant-ph/0001106}\ } (\bibinfo {year}
  {2000})}\BibitemShut {NoStop}%
\bibitem [{\citenamefont {Aharonov}\ \emph {et~al.}(2007)\citenamefont
  {Aharonov}, \citenamefont {van Dam}, \citenamefont {Kempe}, \citenamefont
  {Landau}, \citenamefont {Lloyd},\ and\ \citenamefont
  {Regev}}]{aharonov_adiabatic_2007}%
  \BibitemOpen
  \bibfield  {author} {\bibinfo {author} {\bibfnamefont {D.}~\bibnamefont
  {Aharonov}}, \bibinfo {author} {\bibfnamefont {W.}~\bibnamefont {van Dam}},
  \bibinfo {author} {\bibfnamefont {J.}~\bibnamefont {Kempe}}, \bibinfo
  {author} {\bibfnamefont {Z.}~\bibnamefont {Landau}}, \bibinfo {author}
  {\bibfnamefont {S.}~\bibnamefont {Lloyd}},\ and\ \bibinfo {author}
  {\bibfnamefont {O.}~\bibnamefont {Regev}},\ }\bibfield  {title} {\emph
  {\bibinfo {title} {Adiabatic quantum computation is equivalent to standard
  quantum computation}},\ }\href {https://doi.org/10.1137/S0097539705447323}
  {\bibfield  {journal} {\bibinfo  {journal} {SIAM Journal on Computing}\
  }\textbf {\bibinfo {volume} {37}},\ \bibinfo {pages} {166--194} (\bibinfo
  {year} {2007})}\BibitemShut {NoStop}%
\bibitem [{\citenamefont {Albash}\ and\ \citenamefont
  {Lidar}(2018{\natexlab{a}})}]{Albash-Lidar:RMP}%
  \BibitemOpen
  \bibfield  {author} {\bibinfo {author} {\bibfnamefont {T.}~\bibnamefont
  {Albash}}\ and\ \bibinfo {author} {\bibfnamefont {D.~A.}\ \bibnamefont
  {Lidar}},\ }\bibfield  {title} {\emph {\bibinfo {title} {Adiabatic quantum
  computation}},\ }\href
  {https://link.aps.org/doi/10.1103/RevModPhys.90.015002} {\bibfield  {journal}
  {\bibinfo  {journal} {Reviews of Modern Physics}\ }\textbf {\bibinfo {volume}
  {90}},\ \bibinfo {pages} {015002--} (\bibinfo {year}
  {2018}{\natexlab{a}})}\BibitemShut {NoStop}%
\bibitem [{\citenamefont {Kadowaki}\ and\ \citenamefont
  {Nishimori}(1998)}]{kadowaki_quantum_1998}%
  \BibitemOpen
  \bibfield  {author} {\bibinfo {author} {\bibfnamefont {T.}~\bibnamefont
  {Kadowaki}}\ and\ \bibinfo {author} {\bibfnamefont {H.}~\bibnamefont
  {Nishimori}},\ }\bibfield  {title} {\emph {\bibinfo {title} {Quantum
  annealing in the transverse ising model}},\ }\href
  {https://doi.org/10.1103/PhysRevE.58.5355} {\bibfield  {journal} {\bibinfo
  {journal} {Phys. Rev. E}\ }\textbf {\bibinfo {volume} {58}},\ \bibinfo
  {pages} {5355--5363} (\bibinfo {year} {1998})}\BibitemShut {NoStop}%
\bibitem [{\citenamefont {Hauke}\ \emph {et~al.}(2020)\citenamefont {Hauke},
  \citenamefont {Katzgraber}, \citenamefont {Lechner}, \citenamefont
  {Nishimori},\ and\ \citenamefont {Oliver}}]{Hauke:2019aa}%
  \BibitemOpen
  \bibfield  {author} {\bibinfo {author} {\bibfnamefont {P.}~\bibnamefont
  {Hauke}}, \bibinfo {author} {\bibfnamefont {H.~G.}\ \bibnamefont
  {Katzgraber}}, \bibinfo {author} {\bibfnamefont {W.}~\bibnamefont {Lechner}},
  \bibinfo {author} {\bibfnamefont {H.}~\bibnamefont {Nishimori}},\ and\
  \bibinfo {author} {\bibfnamefont {W.~D.}\ \bibnamefont {Oliver}},\ }\bibfield
   {title} {\emph {\bibinfo {title} {Perspectives of quantum annealing: Methods
  and implementations}},\ }\href
  {https://iopscience.iop.org/article/10.1088/1361-6633/ab85b8} {\bibfield
  {journal} {\bibinfo  {journal} {Reports on Progress in Physics}\ } (\bibinfo
  {year} {2020})}\BibitemShut {NoStop}%
\bibitem [{\citenamefont {Crosson}\ and\ \citenamefont
  {Lidar}(2021)}]{crosson2020prospects}%
  \BibitemOpen
  \bibfield  {author} {\bibinfo {author} {\bibfnamefont {E.~J.}\ \bibnamefont
  {Crosson}}\ and\ \bibinfo {author} {\bibfnamefont {D.~A.}\ \bibnamefont
  {Lidar}},\ }\bibfield  {title} {\emph {\bibinfo {title} {Prospects for
  quantum enhancement with diabatic quantum annealing}},\ }\href
  {https://www.nature.com/articles/s42254-021-00313-6} {\bibfield  {journal}
  {\bibinfo  {journal} {Nature Reviews Physics}\ }\textbf {\bibinfo {volume}
  {3}},\ \bibinfo {pages} {466--489} (\bibinfo {year} {2021})}\BibitemShut
  {NoStop}%
\bibitem [{\citenamefont {Zanardi}\ and\ \citenamefont {Rasetti}(1999)}]{HQC}%
  \BibitemOpen
  \bibfield  {author} {\bibinfo {author} {\bibfnamefont {P.}~\bibnamefont
  {Zanardi}}\ and\ \bibinfo {author} {\bibfnamefont {M.}~\bibnamefont
  {Rasetti}},\ }\bibfield  {title} {\emph {\bibinfo {title} {Holonomic quantum
  computation}},\ }\href
  {http://www.sciencedirect.com/science/article/pii/S0375960199008038}
  {\bibfield  {journal} {\bibinfo  {journal} {Physics Letters A}\ }\textbf
  {\bibinfo {volume} {264}},\ \bibinfo {pages} {94--99} (\bibinfo {year}
  {1999})}\BibitemShut {NoStop}%
\bibitem [{\citenamefont {Duan}\ \emph {et~al.}(2001)\citenamefont {Duan},
  \citenamefont {Cirac},\ and\ \citenamefont {Zoller}}]{Duan:2001ff}%
  \BibitemOpen
  \bibfield  {author} {\bibinfo {author} {\bibfnamefont {L.~M.}\ \bibnamefont
  {Duan}}, \bibinfo {author} {\bibfnamefont {J.~I.}\ \bibnamefont {Cirac}},\
  and\ \bibinfo {author} {\bibfnamefont {P.}~\bibnamefont {Zoller}},\
  }\bibfield  {title} {\emph {\bibinfo {title} {Geometric manipulation of
  trapped ions for quantum computation}},\ }\href
  {http://www.sciencemag.org/content/292/5522/1695.abstract} {\bibfield
  {journal} {\bibinfo  {journal} {Science}\ }\textbf {\bibinfo {volume}
  {292}},\ \bibinfo {pages} {1695--1697} (\bibinfo {year} {2001})}\BibitemShut
  {NoStop}%
\bibitem [{\citenamefont {Zhang}\ \emph {et~al.}(2023)\citenamefont {Zhang},
  \citenamefont {Kyaw}, \citenamefont {Filipp}, \citenamefont {Kwek},
  \citenamefont {Sj{\"o}qvist},\ and\ \citenamefont {Tong}}]{Zhang:2023aa}%
  \BibitemOpen
  \bibfield  {author} {\bibinfo {author} {\bibfnamefont {J.}~\bibnamefont
  {Zhang}}, \bibinfo {author} {\bibfnamefont {T.~H.}\ \bibnamefont {Kyaw}},
  \bibinfo {author} {\bibfnamefont {S.}~\bibnamefont {Filipp}}, \bibinfo
  {author} {\bibfnamefont {L.-C.}\ \bibnamefont {Kwek}}, \bibinfo {author}
  {\bibfnamefont {E.}~\bibnamefont {Sj{\"o}qvist}},\ and\ \bibinfo {author}
  {\bibfnamefont {D.}~\bibnamefont {Tong}},\ }\bibfield  {title} {\emph
  {\bibinfo {title} {Geometric and holonomic quantum computation}},\ }\href
  {https://doi.org/https://doi.org/10.1016/j.physrep.2023.07.004} {\bibfield
  {journal} {\bibinfo  {journal} {Physics Reports}\ }\textbf {\bibinfo {volume}
  {1027}},\ \bibinfo {pages} {1--53} (\bibinfo {year} {2023})}\BibitemShut
  {NoStop}%
\bibitem [{\citenamefont {Bernien}\ \emph {et~al.}(2017)\citenamefont
  {Bernien}, \citenamefont {Schwartz}, \citenamefont {Keesling}, \citenamefont
  {Levine}, \citenamefont {Omran}, \citenamefont {Pichler}, \citenamefont
  {Choi}, \citenamefont {Zibrov}, \citenamefont {Endres}, \citenamefont
  {Greiner}, \citenamefont {Vuleti{\'c}},\ and\ \citenamefont
  {Lukin}}]{Bernien:2017aa}%
  \BibitemOpen
  \bibfield  {author} {\bibinfo {author} {\bibfnamefont {H.}~\bibnamefont
  {Bernien}}, \bibinfo {author} {\bibfnamefont {S.}~\bibnamefont {Schwartz}},
  \bibinfo {author} {\bibfnamefont {A.}~\bibnamefont {Keesling}}, \bibinfo
  {author} {\bibfnamefont {H.}~\bibnamefont {Levine}}, \bibinfo {author}
  {\bibfnamefont {A.}~\bibnamefont {Omran}}, \bibinfo {author} {\bibfnamefont
  {H.}~\bibnamefont {Pichler}}, \bibinfo {author} {\bibfnamefont
  {S.}~\bibnamefont {Choi}}, \bibinfo {author} {\bibfnamefont {A.~S.}\
  \bibnamefont {Zibrov}}, \bibinfo {author} {\bibfnamefont {M.}~\bibnamefont
  {Endres}}, \bibinfo {author} {\bibfnamefont {M.}~\bibnamefont {Greiner}},
  \bibinfo {author} {\bibfnamefont {V.}~\bibnamefont {Vuleti{\'c}}},\ and\
  \bibinfo {author} {\bibfnamefont {M.~D.}\ \bibnamefont {Lukin}},\ }\bibfield
  {title} {\emph {\bibinfo {title} {Probing many-body dynamics on a 51-atom
  quantum simulator}},\ }\href {https://doi.org/10.1038/nature24622} {\bibfield
   {journal} {\bibinfo  {journal} {Nature}\ }\textbf {\bibinfo {volume}
  {551}},\ \bibinfo {pages} {579--584} (\bibinfo {year} {2017})}\BibitemShut
  {NoStop}%
\bibitem [{\citenamefont {King}\ \emph {et~al.}(2022)\citenamefont {King},
  \citenamefont {Suzuki}, \citenamefont {Raymond}, \citenamefont {Zucca},
  \citenamefont {Lanting}, \citenamefont {Altomare}, \citenamefont {Berkley},
  \citenamefont {Ejtemaee}, \citenamefont {Hoskinson}, \citenamefont {Huang},
  \citenamefont {Ladizinsky}, \citenamefont {MacDonald}, \citenamefont
  {Marsden}, \citenamefont {Oh}, \citenamefont {Poulin-Lamarre}, \citenamefont
  {Reis}, \citenamefont {Rich}, \citenamefont {Sato}, \citenamefont
  {Whittaker}, \citenamefont {Yao}, \citenamefont {Harris}, \citenamefont
  {Lidar}, \citenamefont {Nishimori},\ and\ \citenamefont
  {Amin}}]{king2022coherent}%
  \BibitemOpen
  \bibfield  {author} {\bibinfo {author} {\bibfnamefont {A.~D.}\ \bibnamefont
  {King}}, \bibinfo {author} {\bibfnamefont {S.}~\bibnamefont {Suzuki}},
  \bibinfo {author} {\bibfnamefont {J.}~\bibnamefont {Raymond}}, \bibinfo
  {author} {\bibfnamefont {A.}~\bibnamefont {Zucca}}, \bibinfo {author}
  {\bibfnamefont {T.}~\bibnamefont {Lanting}}, \bibinfo {author} {\bibfnamefont
  {F.}~\bibnamefont {Altomare}}, \bibinfo {author} {\bibfnamefont {A.~J.}\
  \bibnamefont {Berkley}}, \bibinfo {author} {\bibfnamefont {S.}~\bibnamefont
  {Ejtemaee}}, \bibinfo {author} {\bibfnamefont {E.}~\bibnamefont {Hoskinson}},
  \bibinfo {author} {\bibfnamefont {S.}~\bibnamefont {Huang}}, \bibinfo
  {author} {\bibfnamefont {E.}~\bibnamefont {Ladizinsky}}, \bibinfo {author}
  {\bibfnamefont {A.~J.~R.}\ \bibnamefont {MacDonald}}, \bibinfo {author}
  {\bibfnamefont {G.}~\bibnamefont {Marsden}}, \bibinfo {author} {\bibfnamefont
  {T.}~\bibnamefont {Oh}}, \bibinfo {author} {\bibfnamefont {G.}~\bibnamefont
  {Poulin-Lamarre}}, \bibinfo {author} {\bibfnamefont {M.}~\bibnamefont
  {Reis}}, \bibinfo {author} {\bibfnamefont {C.}~\bibnamefont {Rich}}, \bibinfo
  {author} {\bibfnamefont {Y.}~\bibnamefont {Sato}}, \bibinfo {author}
  {\bibfnamefont {J.~D.}\ \bibnamefont {Whittaker}}, \bibinfo {author}
  {\bibfnamefont {J.}~\bibnamefont {Yao}}, \bibinfo {author} {\bibfnamefont
  {R.}~\bibnamefont {Harris}}, \bibinfo {author} {\bibfnamefont {D.~A.}\
  \bibnamefont {Lidar}}, \bibinfo {author} {\bibfnamefont {H.}~\bibnamefont
  {Nishimori}},\ and\ \bibinfo {author} {\bibfnamefont {M.~H.}\ \bibnamefont
  {Amin}},\ }\bibfield  {title} {\emph {\bibinfo {title} {Coherent quantum
  annealing in a programmable 2,000 qubit ising chain}},\ }\href
  {https://doi.org/10.1038/s41567-022-01741-6} {\bibfield  {journal} {\bibinfo
  {journal} {Nature Physics}\ }\textbf {\bibinfo {volume} {18}},\ \bibinfo
  {pages} {1324--1328} (\bibinfo {year} {2022})}\BibitemShut {NoStop}%
\bibitem [{\citenamefont {Altman}\ \emph {et~al.}(2021)\citenamefont {Altman},
  \citenamefont {Brown}, \citenamefont {Carleo}, \citenamefont {Carr},
  \citenamefont {Demler}, \citenamefont {Chin}, \citenamefont {DeMarco},
  \citenamefont {Economou}, \citenamefont {Eriksson}, \citenamefont {Fu},
  \citenamefont {Greiner}, \citenamefont {Hazzard}, \citenamefont {Hulet},
  \citenamefont {Koll{\'a}r}, \citenamefont {Lev}, \citenamefont {Lukin},
  \citenamefont {Ma}, \citenamefont {Mi}, \citenamefont {Misra}, \citenamefont
  {Monroe}, \citenamefont {Murch}, \citenamefont {Nazario}, \citenamefont {Ni},
  \citenamefont {Potter}, \citenamefont {Roushan}, \citenamefont {Saffman},
  \citenamefont {Schleier-Smith}, \citenamefont {Siddiqi}, \citenamefont
  {Simmonds}, \citenamefont {Singh}, \citenamefont {Spielman}, \citenamefont
  {Temme}, \citenamefont {Weiss}, \citenamefont {Vu{\v c}kovi{\'c}},
  \citenamefont {Vuleti{\'c}}, \citenamefont {Ye},\ and\ \citenamefont
  {Zwierlein}}]{altmanQuantumSimulatorsArchitectures2019}%
  \BibitemOpen
  \bibfield  {author} {\bibinfo {author} {\bibfnamefont {E.}~\bibnamefont
  {Altman}}, \bibinfo {author} {\bibfnamefont {K.~R.}\ \bibnamefont {Brown}},
  \bibinfo {author} {\bibfnamefont {G.}~\bibnamefont {Carleo}}, \bibinfo
  {author} {\bibfnamefont {L.~D.}\ \bibnamefont {Carr}}, \bibinfo {author}
  {\bibfnamefont {E.}~\bibnamefont {Demler}}, \bibinfo {author} {\bibfnamefont
  {C.}~\bibnamefont {Chin}}, \bibinfo {author} {\bibfnamefont {B.}~\bibnamefont
  {DeMarco}}, \bibinfo {author} {\bibfnamefont {S.~E.}\ \bibnamefont
  {Economou}}, \bibinfo {author} {\bibfnamefont {M.~A.}\ \bibnamefont
  {Eriksson}}, \bibinfo {author} {\bibfnamefont {K.-M.~C.}\ \bibnamefont {Fu}},
  \bibinfo {author} {\bibfnamefont {M.}~\bibnamefont {Greiner}}, \bibinfo
  {author} {\bibfnamefont {K.~R.~A.}\ \bibnamefont {Hazzard}}, \bibinfo
  {author} {\bibfnamefont {R.~G.}\ \bibnamefont {Hulet}}, \bibinfo {author}
  {\bibfnamefont {A.~J.}\ \bibnamefont {Koll{\'a}r}}, \bibinfo {author}
  {\bibfnamefont {B.~L.}\ \bibnamefont {Lev}}, \bibinfo {author} {\bibfnamefont
  {M.~D.}\ \bibnamefont {Lukin}}, \bibinfo {author} {\bibfnamefont
  {R.}~\bibnamefont {Ma}}, \bibinfo {author} {\bibfnamefont {X.}~\bibnamefont
  {Mi}}, \bibinfo {author} {\bibfnamefont {S.}~\bibnamefont {Misra}}, \bibinfo
  {author} {\bibfnamefont {C.}~\bibnamefont {Monroe}}, \bibinfo {author}
  {\bibfnamefont {K.}~\bibnamefont {Murch}}, \bibinfo {author} {\bibfnamefont
  {Z.}~\bibnamefont {Nazario}}, \bibinfo {author} {\bibfnamefont {K.-K.}\
  \bibnamefont {Ni}}, \bibinfo {author} {\bibfnamefont {A.~C.}\ \bibnamefont
  {Potter}}, \bibinfo {author} {\bibfnamefont {P.}~\bibnamefont {Roushan}},
  \bibinfo {author} {\bibfnamefont {M.}~\bibnamefont {Saffman}}, \bibinfo
  {author} {\bibfnamefont {M.}~\bibnamefont {Schleier-Smith}}, \bibinfo
  {author} {\bibfnamefont {I.}~\bibnamefont {Siddiqi}}, \bibinfo {author}
  {\bibfnamefont {R.}~\bibnamefont {Simmonds}}, \bibinfo {author}
  {\bibfnamefont {M.}~\bibnamefont {Singh}}, \bibinfo {author} {\bibfnamefont
  {I.~B.}\ \bibnamefont {Spielman}}, \bibinfo {author} {\bibfnamefont
  {K.}~\bibnamefont {Temme}}, \bibinfo {author} {\bibfnamefont {D.~S.}\
  \bibnamefont {Weiss}}, \bibinfo {author} {\bibfnamefont {J.}~\bibnamefont
  {Vu{\v c}kovi{\'c}}}, \bibinfo {author} {\bibfnamefont {V.}~\bibnamefont
  {Vuleti{\'c}}}, \bibinfo {author} {\bibfnamefont {J.}~\bibnamefont {Ye}},\
  and\ \bibinfo {author} {\bibfnamefont {M.}~\bibnamefont {Zwierlein}},\
  }\bibfield  {title} {\emph {\bibinfo {title} {Quantum simulators:
  Architectures and opportunities}},\ }\href
  {https://doi.org/10.1103/PRXQuantum.2.017003} {\bibfield  {journal} {\bibinfo
   {journal} {PRX Quantum}\ }\textbf {\bibinfo {volume} {2}},\ \bibinfo {pages}
  {017003--} (\bibinfo {year} {2021})}\BibitemShut {NoStop}%
\bibitem [{\citenamefont {Farhi}\ and\ \citenamefont
  {Gutmann}(1998)}]{Farhi:1998aa}%
  \BibitemOpen
  \bibfield  {author} {\bibinfo {author} {\bibfnamefont {E.}~\bibnamefont
  {Farhi}}\ and\ \bibinfo {author} {\bibfnamefont {S.}~\bibnamefont
  {Gutmann}},\ }\bibfield  {title} {\emph {\bibinfo {title} {Quantum
  computation and decision trees}},\ }\href
  {https://doi.org/10.1103/PhysRevA.58.915} {\bibfield  {journal} {\bibinfo
  {journal} {Physical Review A}\ }\textbf {\bibinfo {volume} {58}},\ \bibinfo
  {pages} {915--928} (\bibinfo {year} {1998})}\BibitemShut {NoStop}%
\bibitem [{\citenamefont {Childs}\ \emph {et~al.}(2013)\citenamefont {Childs},
  \citenamefont {Gosset},\ and\ \citenamefont {Webb}}]{Childs:2013kx}%
  \BibitemOpen
  \bibfield  {author} {\bibinfo {author} {\bibfnamefont {A.~M.}\ \bibnamefont
  {Childs}}, \bibinfo {author} {\bibfnamefont {D.}~\bibnamefont {Gosset}},\
  and\ \bibinfo {author} {\bibfnamefont {Z.}~\bibnamefont {Webb}},\ }\bibfield
  {title} {\emph {\bibinfo {title} {Universal computation by multiparticle
  quantum walk}},\ }\href
  {http://science.sciencemag.org/content/339/6121/791.abstract} {\bibfield
  {journal} {\bibinfo  {journal} {Science}\ }\textbf {\bibinfo {volume}
  {339}},\ \bibinfo {pages} {791--794} (\bibinfo {year} {2013})}\BibitemShut
  {NoStop}%
\bibitem [{\citenamefont {M{\"u}lken}\ and\ \citenamefont
  {Blumen}(2011)}]{Mulken:2011aa}%
  \BibitemOpen
  \bibfield  {author} {\bibinfo {author} {\bibfnamefont {O.}~\bibnamefont
  {M{\"u}lken}}\ and\ \bibinfo {author} {\bibfnamefont {A.}~\bibnamefont
  {Blumen}},\ }\bibfield  {title} {\emph {\bibinfo {title} {Continuous-time
  quantum walks: Models for coherent transport on complex networks}},\ }\href
  {https://doi.org/https://doi.org/10.1016/j.physrep.2011.01.002} {\bibfield
  {journal} {\bibinfo  {journal} {Physics Reports}\ }\textbf {\bibinfo {volume}
  {502}},\ \bibinfo {pages} {37--87} (\bibinfo {year} {2011})}\BibitemShut
  {NoStop}%
\bibitem [{\citenamefont {Oreshkov}\ \emph {et~al.}(2009)\citenamefont
  {Oreshkov}, \citenamefont {Brun},\ and\ \citenamefont
  {Lidar}}]{Oreshkov:2009bl}%
  \BibitemOpen
  \bibfield  {author} {\bibinfo {author} {\bibfnamefont {O.}~\bibnamefont
  {Oreshkov}}, \bibinfo {author} {\bibfnamefont {T.~A.}\ \bibnamefont {Brun}},\
  and\ \bibinfo {author} {\bibfnamefont {D.~A.}\ \bibnamefont {Lidar}},\
  }\bibfield  {title} {\emph {\bibinfo {title} {Fault-tolerant holonomic
  quantum computation}},\ }\href
  {http://link.aps.org/doi/10.1103/PhysRevLett.102.070502} {\bibfield
  {journal} {\bibinfo  {journal} {Phys. Rev. Lett.}\ }\textbf {\bibinfo
  {volume} {102}},\ \bibinfo {pages} {070502--} (\bibinfo {year}
  {2009})}\BibitemShut {NoStop}%
\bibitem [{\citenamefont {Jordan}\ \emph {et~al.}(2006)\citenamefont {Jordan},
  \citenamefont {Farhi},\ and\ \citenamefont {Shor}}]{jordan2006error}%
  \BibitemOpen
  \bibfield  {author} {\bibinfo {author} {\bibfnamefont {S.~P.}\ \bibnamefont
  {Jordan}}, \bibinfo {author} {\bibfnamefont {E.}~\bibnamefont {Farhi}},\ and\
  \bibinfo {author} {\bibfnamefont {P.~W.}\ \bibnamefont {Shor}},\ }\bibfield
  {title} {\emph {\bibinfo {title} {Error-correcting codes for adiabatic
  quantum computation}},\ }\href
  {http://link.aps.org/doi/10.1103/PhysRevA.74.052322} {\bibfield  {journal}
  {\bibinfo  {journal} {{Phys. Rev. A}}\ }\textbf {\bibinfo {volume} {74}},\
  \bibinfo {pages} {052322} (\bibinfo {year} {2006})}\BibitemShut {NoStop}%
\bibitem [{\citenamefont {Gottesman}(1996)}]{Gottesman:1996fk}%
  \BibitemOpen
  \bibfield  {author} {\bibinfo {author} {\bibfnamefont {D.}~\bibnamefont
  {Gottesman}},\ }\bibfield  {title} {\emph {\bibinfo {title} {Class of quantum
  error-correcting codes saturating the quantum hamming bound}},\ }\href
  {https://doi.org/10.1103/PhysRevA.54.1862} {\bibfield  {journal} {\bibinfo
  {journal} {{Phys. Rev. A}}\ }\textbf {\bibinfo {volume} {54}},\ \bibinfo
  {pages} {1862} (\bibinfo {year} {1996})}\BibitemShut {NoStop}%
\bibitem [{\citenamefont {Calderbank}\ \emph {et~al.}(1998)\citenamefont
  {Calderbank}, \citenamefont {Rains}, \citenamefont {Shor},\ and\
  \citenamefont {Sloane}}]{Calderbank:98}%
  \BibitemOpen
  \bibfield  {author} {\bibinfo {author} {\bibfnamefont {A.~R.}\ \bibnamefont
  {Calderbank}}, \bibinfo {author} {\bibfnamefont {E.~M.}\ \bibnamefont
  {Rains}}, \bibinfo {author} {\bibfnamefont {P.~M.}\ \bibnamefont {Shor}},\
  and\ \bibinfo {author} {\bibfnamefont {N.~J.~A.}\ \bibnamefont {Sloane}},\
  }\bibfield  {title} {\emph {\bibinfo {title} {Quantum error correction via
  codes over $gf(4)$}},\ }\href {https://ieeexplore.ieee.org/document/681315}
  {\bibfield  {journal} {\bibinfo  {journal} {IEEE Transactions on Information
  Theory}\ }\textbf {\bibinfo {volume} {44}},\ \bibinfo {pages} {1369--1387}
  (\bibinfo {year} {1998})}\BibitemShut {NoStop}%
\bibitem [{\citenamefont {Bravyi}\ and\ \citenamefont
  {Vyalyi}(2005)}]{Bravyi:2003tx}%
  \BibitemOpen
  \bibfield  {author} {\bibinfo {author} {\bibfnamefont {S.}~\bibnamefont
  {Bravyi}}\ and\ \bibinfo {author} {\bibfnamefont {M.}~\bibnamefont
  {Vyalyi}},\ }\bibfield  {title} {\emph {\bibinfo {title} {Commutative version
  of the k-local hamiltonian problem and common eigenspace problem}},\ }\href
  {http://arXiv.org/abs/quant-ph/0308021} {\bibfield  {journal} {\bibinfo
  {journal} {Quantum Inf. and Comp.}\ }\textbf {\bibinfo {volume} {5}},\
  \bibinfo {pages} {187--215} (\bibinfo {year} {2005})}\BibitemShut {NoStop}%
\bibitem [{\citenamefont {Dennis}\ \emph {et~al.}(2002)\citenamefont {Dennis},
  \citenamefont {Kitaev}, \citenamefont {Landahl},\ and\ \citenamefont
  {Preskill}}]{Dennis:02}%
  \BibitemOpen
  \bibfield  {author} {\bibinfo {author} {\bibfnamefont {E.}~\bibnamefont
  {Dennis}}, \bibinfo {author} {\bibfnamefont {A.}~\bibnamefont {Kitaev}},
  \bibinfo {author} {\bibfnamefont {A.}~\bibnamefont {Landahl}},\ and\ \bibinfo
  {author} {\bibfnamefont {J.}~\bibnamefont {Preskill}},\ }\bibfield  {title}
  {\emph {\bibinfo {title} {Topological quantum memory}},\ }\href
  {http://scitation.aip.org/content/aip/journal/jmp/43/9/10.1063/1.1499754}
  {\bibfield  {journal} {\bibinfo  {journal} {Journal of Mathematical Physics}\
  }\textbf {\bibinfo {volume} {43}},\ \bibinfo {pages} {4452--4505} (\bibinfo
  {year} {2002})}\BibitemShut {NoStop}%
\bibitem [{\citenamefont {Alicki}\ \emph {et~al.}(2007)\citenamefont {Alicki},
  \citenamefont {Fannes},\ and\ \citenamefont {Horodecki}}]{Alicki:07a}%
  \BibitemOpen
  \bibfield  {author} {\bibinfo {author} {\bibfnamefont {R.}~\bibnamefont
  {Alicki}}, \bibinfo {author} {\bibfnamefont {M.}~\bibnamefont {Fannes}},\
  and\ \bibinfo {author} {\bibfnamefont {M.}~\bibnamefont {Horodecki}},\
  }\bibfield  {title} {\emph {\bibinfo {title} {A statistical mechanics view on
  kitaev's proposal for quantum memories}},\ }\href
  {https://doi.org/10.1088/1751-8113/40/24/012} {\bibfield  {journal} {\bibinfo
   {journal} {Journal of Physics A: Mathematical and Theoretical}\ }\textbf
  {\bibinfo {volume} {40}},\ \bibinfo {pages} {6451} (\bibinfo {year}
  {2007})}\BibitemShut {NoStop}%
\bibitem [{\citenamefont {Young}\ \emph
  {et~al.}(2013{\natexlab{a}})\citenamefont {Young}, \citenamefont
  {Blume-Kohout},\ and\ \citenamefont {Lidar}}]{Young:2013fk}%
  \BibitemOpen
  \bibfield  {author} {\bibinfo {author} {\bibfnamefont {K.~C.}\ \bibnamefont
  {Young}}, \bibinfo {author} {\bibfnamefont {R.}~\bibnamefont
  {Blume-Kohout}},\ and\ \bibinfo {author} {\bibfnamefont {D.~A.}\ \bibnamefont
  {Lidar}},\ }\bibfield  {title} {\emph {\bibinfo {title} {Adiabatic quantum
  optimization with the wrong hamiltonian}},\ }\href
  {http://link.aps.org/doi/10.1103/PhysRevA.88.062314} {\bibfield  {journal}
  {\bibinfo  {journal} {Phys. Rev. A}\ }\textbf {\bibinfo {volume} {88}},\
  \bibinfo {pages} {062314--} (\bibinfo {year}
  {2013}{\natexlab{a}})}\BibitemShut {NoStop}%
\bibitem [{\citenamefont {Young}\ \emph
  {et~al.}(2013{\natexlab{b}})\citenamefont {Young}, \citenamefont {Sarovar},\
  and\ \citenamefont {Blume-Kohout}}]{Young:13}%
  \BibitemOpen
  \bibfield  {author} {\bibinfo {author} {\bibfnamefont {K.~C.}\ \bibnamefont
  {Young}}, \bibinfo {author} {\bibfnamefont {M.}~\bibnamefont {Sarovar}},\
  and\ \bibinfo {author} {\bibfnamefont {R.}~\bibnamefont {Blume-Kohout}},\
  }\bibfield  {title} {\emph {\bibinfo {title} {Error suppression and error
  correction in adiabatic quantum computation: Techniques and challenges}},\
  }\href {http://link.aps.org/doi/10.1103/PhysRevX.3.041013} {\bibfield
  {journal} {\bibinfo  {journal} {Phys. Rev. X}\ }\textbf {\bibinfo {volume}
  {3}},\ \bibinfo {pages} {041013--} (\bibinfo {year}
  {2013}{\natexlab{b}})}\BibitemShut {NoStop}%
\bibitem [{\citenamefont {Sarovar}\ and\ \citenamefont
  {Young}(2013)}]{Sarovar:2013kx}%
  \BibitemOpen
  \bibfield  {author} {\bibinfo {author} {\bibfnamefont {M.}~\bibnamefont
  {Sarovar}}\ and\ \bibinfo {author} {\bibfnamefont {K.~C.}\ \bibnamefont
  {Young}},\ }\bibfield  {title} {\emph {\bibinfo {title} {Error suppression
  and error correction in adiabatic quantum computation: non-equilibrium
  dynamics}},\ }\href {http://stacks.iop.org/1367-2630/15/i=12/a=125032}
  {\bibfield  {journal} {\bibinfo  {journal} {New J. of Phys.}\ }\textbf
  {\bibinfo {volume} {15}},\ \bibinfo {pages} {125032} (\bibinfo {year}
  {2013})}\BibitemShut {NoStop}%
\bibitem [{\citenamefont {Ganti}\ \emph {et~al.}(2014)\citenamefont {Ganti},
  \citenamefont {Onunkwo},\ and\ \citenamefont {Young}}]{Ganti:13}%
  \BibitemOpen
  \bibfield  {author} {\bibinfo {author} {\bibfnamefont {A.}~\bibnamefont
  {Ganti}}, \bibinfo {author} {\bibfnamefont {U.}~\bibnamefont {Onunkwo}},\
  and\ \bibinfo {author} {\bibfnamefont {K.}~\bibnamefont {Young}},\ }\bibfield
   {title} {\emph {\bibinfo {title} {Family of $[[6k,2k,2]]$ codes for
  practical, scalable adiabatic quantum computation}},\ }\href
  {http://link.aps.org/doi/10.1103/PhysRevA.89.042313} {\bibfield  {journal}
  {\bibinfo  {journal} {Phys. Rev. A}\ }\textbf {\bibinfo {volume} {89}},\
  \bibinfo {pages} {042313--} (\bibinfo {year} {2014})}\BibitemShut {NoStop}%
\bibitem [{\citenamefont {Bookatz}\ \emph {et~al.}(2015)\citenamefont
  {Bookatz}, \citenamefont {Farhi},\ and\ \citenamefont
  {Zhou}}]{Bookatz:2014uq}%
  \BibitemOpen
  \bibfield  {author} {\bibinfo {author} {\bibfnamefont {A.~D.}\ \bibnamefont
  {Bookatz}}, \bibinfo {author} {\bibfnamefont {E.}~\bibnamefont {Farhi}},\
  and\ \bibinfo {author} {\bibfnamefont {L.}~\bibnamefont {Zhou}},\ }\bibfield
  {title} {\emph {\bibinfo {title} {Error suppression in hamiltonian-based
  quantum computation using energy penalties}},\ }\href
  {http://link.aps.org/doi/10.1103/PhysRevA.92.022317} {\bibfield  {journal}
  {\bibinfo  {journal} {{Phys. Rev. A}}\ }\textbf {\bibinfo {volume} {92}},\
  \bibinfo {pages} {022317--} (\bibinfo {year} {2015})}\BibitemShut {NoStop}%
\bibitem [{\citenamefont {Marvian}(2016)}]{Marvian:2016kb}%
  \BibitemOpen
  \bibfield  {author} {\bibinfo {author} {\bibfnamefont {I.}~\bibnamefont
  {Marvian}},\ }\bibfield  {title} {\emph {\bibinfo {title} {Exponential
  suppression of decoherence and relaxation of quantum systems using energy
  penalty}},\ }\href {http://arXiv.org/abs/1602.03251} {\bibfield  {journal}
  {\bibinfo  {journal} {arXiv:1602.03251}\ } (\bibinfo {year}
  {2016})}\BibitemShut {NoStop}%
\bibitem [{\citenamefont {Marvian}\ and\ \citenamefont
  {Lidar}(2017{\natexlab{a}})}]{Marvian:2017aa}%
  \BibitemOpen
  \bibfield  {author} {\bibinfo {author} {\bibfnamefont {M.}~\bibnamefont
  {Marvian}}\ and\ \bibinfo {author} {\bibfnamefont {D.~A.}\ \bibnamefont
  {Lidar}},\ }\bibfield  {title} {\emph {\bibinfo {title} {Error suppression
  for hamiltonian quantum computing in markovian environments}},\ }\href
  {http://link.aps.org/doi/10.1103/PhysRevA.95.032302} {\bibfield  {journal}
  {\bibinfo  {journal} {Physical Review A}\ }\textbf {\bibinfo {volume} {95}},\
  \bibinfo {pages} {032302--} (\bibinfo {year}
  {2017}{\natexlab{a}})}\BibitemShut {NoStop}%
\bibitem [{\citenamefont {Lidar}(2019)}]{Lidar:2019ab}%
  \BibitemOpen
  \bibfield  {author} {\bibinfo {author} {\bibfnamefont {D.~A.}\ \bibnamefont
  {Lidar}},\ }\bibfield  {title} {\emph {\bibinfo {title} {Arbitrary-time error
  suppression for markovian adiabatic quantum computing using stabilizer
  subspace codes}},\ }\href {https://doi.org/10.1103/PhysRevA.100.022326}
  {\bibfield  {journal} {\bibinfo  {journal} {Physical Review A}\ }\textbf
  {\bibinfo {volume} {100}},\ \bibinfo {pages} {022326--} (\bibinfo {year}
  {2019})}\BibitemShut {NoStop}%
\bibitem [{\citenamefont {Mizel}\ and\ \citenamefont
  {Lidar}(2004)}]{Mizel:2004ve}%
  \BibitemOpen
  \bibfield  {author} {\bibinfo {author} {\bibfnamefont {A.}~\bibnamefont
  {Mizel}}\ and\ \bibinfo {author} {\bibfnamefont {D.~A.}\ \bibnamefont
  {Lidar}},\ }\bibfield  {title} {\emph {\bibinfo {title} {Three- and four-body
  interactions in spin-based quantum computers}},\ }\href
  {http://link.aps.org/doi/10.1103/PhysRevLett.92.077903} {\bibfield  {journal}
  {\bibinfo  {journal} {Phys. Rev. Lett.}\ }\textbf {\bibinfo {volume} {92}},\
  \bibinfo {pages} {077903--} (\bibinfo {year} {2004})}\BibitemShut {NoStop}%
\bibitem [{\citenamefont {Gurian}\ \emph {et~al.}(2012)\citenamefont {Gurian},
  \citenamefont {Cheinet}, \citenamefont {Huillery}, \citenamefont {Fioretti},
  \citenamefont {Zhao}, \citenamefont {Gould}, \citenamefont {Comparat},\ and\
  \citenamefont {Pillet}}]{Gurian:2012aa}%
  \BibitemOpen
  \bibfield  {author} {\bibinfo {author} {\bibfnamefont {J.~H.}\ \bibnamefont
  {Gurian}}, \bibinfo {author} {\bibfnamefont {P.}~\bibnamefont {Cheinet}},
  \bibinfo {author} {\bibfnamefont {P.}~\bibnamefont {Huillery}}, \bibinfo
  {author} {\bibfnamefont {A.}~\bibnamefont {Fioretti}}, \bibinfo {author}
  {\bibfnamefont {J.}~\bibnamefont {Zhao}}, \bibinfo {author} {\bibfnamefont
  {P.~L.}\ \bibnamefont {Gould}}, \bibinfo {author} {\bibfnamefont
  {D.}~\bibnamefont {Comparat}},\ and\ \bibinfo {author} {\bibfnamefont
  {P.}~\bibnamefont {Pillet}},\ }\bibfield  {title} {\emph {\bibinfo {title}
  {Observation of a resonant four-body interaction in cold cesium rydberg
  atoms}},\ }\href {https://doi.org/10.1103/PhysRevLett.108.023005} {\bibfield
  {journal} {\bibinfo  {journal} {Physical Review Letters}\ }\textbf {\bibinfo
  {volume} {108}},\ \bibinfo {pages} {023005--} (\bibinfo {year}
  {2012})}\BibitemShut {NoStop}%
\bibitem [{\citenamefont {Dai}\ \emph {et~al.}(2017)\citenamefont {Dai},
  \citenamefont {Yang}, \citenamefont {Reingruber}, \citenamefont {Sun},
  \citenamefont {Xu}, \citenamefont {Chen}, \citenamefont {Yuan},\ and\
  \citenamefont {Pan}}]{Dai:2017aa}%
  \BibitemOpen
  \bibfield  {author} {\bibinfo {author} {\bibfnamefont {H.-N.}\ \bibnamefont
  {Dai}}, \bibinfo {author} {\bibfnamefont {B.}~\bibnamefont {Yang}}, \bibinfo
  {author} {\bibfnamefont {A.}~\bibnamefont {Reingruber}}, \bibinfo {author}
  {\bibfnamefont {H.}~\bibnamefont {Sun}}, \bibinfo {author} {\bibfnamefont
  {X.-F.}\ \bibnamefont {Xu}}, \bibinfo {author} {\bibfnamefont {Y.-A.}\
  \bibnamefont {Chen}}, \bibinfo {author} {\bibfnamefont {Z.-S.}\ \bibnamefont
  {Yuan}},\ and\ \bibinfo {author} {\bibfnamefont {J.-W.}\ \bibnamefont
  {Pan}},\ }\bibfield  {title} {\emph {\bibinfo {title} {Four-body
  ring-exchange interactions and anyonic statistics within a minimal toric-code
  hamiltonian}},\ }\href {https://doi.org/10.1038/nphys4243} {\bibfield
  {journal} {\bibinfo  {journal} {Nature Physics}\ }\textbf {\bibinfo {volume}
  {13}},\ \bibinfo {pages} {1195--1200} (\bibinfo {year} {2017})}\BibitemShut
  {NoStop}%
\bibitem [{\citenamefont {Chirolli}\ \emph {et~al.}(2024)\citenamefont
  {Chirolli}, \citenamefont {Braggio},\ and\ \citenamefont
  {Giazotto}}]{Chirolli:2024aa}%
  \BibitemOpen
  \bibfield  {author} {\bibinfo {author} {\bibfnamefont {L.}~\bibnamefont
  {Chirolli}}, \bibinfo {author} {\bibfnamefont {A.}~\bibnamefont {Braggio}},\
  and\ \bibinfo {author} {\bibfnamefont {F.}~\bibnamefont {Giazotto}},\
  }\bibfield  {title} {\emph {\bibinfo {title} {Cooper quartets in interacting
  hybrid superconducting systems}},\ }\href
  {https://doi.org/10.1103/PhysRevResearch.6.033171} {\bibfield  {journal}
  {\bibinfo  {journal} {Physical Review Research}\ }\textbf {\bibinfo {volume}
  {6}},\ \bibinfo {pages} {033171--} (\bibinfo {year} {2024})}\BibitemShut
  {NoStop}%
\bibitem [{\citenamefont {Marvian}\ and\ \citenamefont
  {Lidar}(2014)}]{Marvian:2014nr}%
  \BibitemOpen
  \bibfield  {author} {\bibinfo {author} {\bibfnamefont {I.}~\bibnamefont
  {Marvian}}\ and\ \bibinfo {author} {\bibfnamefont {D.~A.}\ \bibnamefont
  {Lidar}},\ }\bibfield  {title} {\emph {\bibinfo {title} {Quantum error
  suppression with commuting hamiltonians: Two local is too local}},\ }\href
  {http://link.aps.org/doi/10.1103/PhysRevLett.113.260504} {\bibfield
  {journal} {\bibinfo  {journal} {Phys. Rev. Lett.}\ }\textbf {\bibinfo
  {volume} {113}},\ \bibinfo {pages} {260504--} (\bibinfo {year}
  {2014})}\BibitemShut {NoStop}%
\bibitem [{\citenamefont {Jiang}\ and\ \citenamefont
  {Rieffel}(2017)}]{Jiang:2015kx}%
  \BibitemOpen
  \bibfield  {author} {\bibinfo {author} {\bibfnamefont {Z.}~\bibnamefont
  {Jiang}}\ and\ \bibinfo {author} {\bibfnamefont {E.~G.}\ \bibnamefont
  {Rieffel}},\ }\bibfield  {title} {\emph {\bibinfo {title} {Non-commuting
  two-local hamiltonians for quantum error suppression}},\ }\href
  {https://doi.org/10.1007/s11128-017-1527-9} {\bibfield  {journal} {\bibinfo
  {journal} {{Quant. Inf. Proc.}}\ }\textbf {\bibinfo {volume} {16}},\ \bibinfo
  {pages} {89} (\bibinfo {year} {2017})}\BibitemShut {NoStop}%
\bibitem [{\citenamefont {Marvian}\ and\ \citenamefont
  {Lidar}(2017{\natexlab{b}})}]{Marvian-Lidar:16}%
  \BibitemOpen
  \bibfield  {author} {\bibinfo {author} {\bibfnamefont {M.}~\bibnamefont
  {Marvian}}\ and\ \bibinfo {author} {\bibfnamefont {D.~A.}\ \bibnamefont
  {Lidar}},\ }\bibfield  {title} {\emph {\bibinfo {title} {{Error Suppression
  for Hamiltonian-Based Quantum Computation Using Subsystem Codes}}},\ }\href
  {https://link.aps.org/doi/10.1103/PhysRevLett.118.030504} {\bibfield
  {journal} {\bibinfo  {journal} {Phys. Rev. Lett.}\ }\textbf {\bibinfo
  {volume} {118}},\ \bibinfo {pages} {030504--} (\bibinfo {year}
  {2017}{\natexlab{b}})}\BibitemShut {NoStop}%
\bibitem [{\citenamefont {Poulin}(2005)}]{poulin_stabilizer_2005}%
  \BibitemOpen
  \bibfield  {author} {\bibinfo {author} {\bibfnamefont {D.}~\bibnamefont
  {Poulin}},\ }\bibfield  {title} {\emph {\bibinfo {title} {Stabilizer
  formalism for operator quantum error correction}},\ }\href
  {https://doi.org/10.1103/PhysRevLett.95.230504} {\bibfield  {journal}
  {\bibinfo  {journal} {Phys. Rev. Lett.}\ }\textbf {\bibinfo {volume} {95}},\
  \bibinfo {pages} {230504} (\bibinfo {year} {2005})}\BibitemShut {NoStop}%
\bibitem [{\citenamefont {Marvian}\ and\ \citenamefont
  {Lloyd}(2019)}]{marvian2019robust}%
  \BibitemOpen
  \bibfield  {author} {\bibinfo {author} {\bibfnamefont {M.}~\bibnamefont
  {Marvian}}\ and\ \bibinfo {author} {\bibfnamefont {S.}~\bibnamefont
  {Lloyd}},\ }\bibfield  {title} {\emph {\bibinfo {title} {Robust universal
  hamiltonian quantum computing using two-body interactions}},\ }\href
  {https://arxiv.org/abs/1911.01354} {\bibfield  {journal} {\bibinfo  {journal}
  {arXiv:1911.01354}\ } (\bibinfo {year} {2019})}\BibitemShut {NoStop}%
\bibitem [{\citenamefont {Bravyi}(2011)}]{Bravyi2011}%
  \BibitemOpen
  \bibfield  {author} {\bibinfo {author} {\bibfnamefont {S.}~\bibnamefont
  {Bravyi}},\ }\bibfield  {title} {\emph {\bibinfo {title} {Subsystem codes
  with spatially local generators}},\ }\href
  {https://doi.org/10.1103/PhysRevA.83.012320} {\bibfield  {journal} {\bibinfo
  {journal} {Phys. Rev. A}\ }\textbf {\bibinfo {volume} {83}},\ \bibinfo
  {pages} {012320} (\bibinfo {year} {2011})}\BibitemShut {NoStop}%
\bibitem [{\citenamefont {Biamonte}\ and\ \citenamefont
  {Love}(2008)}]{Biamonte:07}%
  \BibitemOpen
  \bibfield  {author} {\bibinfo {author} {\bibfnamefont {J.~D.}\ \bibnamefont
  {Biamonte}}\ and\ \bibinfo {author} {\bibfnamefont {P.~J.}\ \bibnamefont
  {Love}},\ }\bibfield  {title} {\emph {\bibinfo {title} {Realizable
  hamiltonians for universal adiabatic quantum computers}},\ }\href
  {https://doi.org/10.1103/PhysRevA.78.012352} {\bibfield  {journal} {\bibinfo
  {journal} {Phys. Rev. A}\ }\textbf {\bibinfo {volume} {78}},\ \bibinfo
  {pages} {012352} (\bibinfo {year} {2008})}\BibitemShut {NoStop}%
\bibitem [{\citenamefont {Johnson}\ \emph {et~al.}(2011)\citenamefont
  {Johnson}, \citenamefont {Amin}, \citenamefont {Gildert}, \citenamefont
  {Lanting}, \citenamefont {Hamze}, \citenamefont {Dickson}, \citenamefont
  {Harris}, \citenamefont {Berkley}, \citenamefont {Johansson}, \citenamefont
  {Bunyk}, \citenamefont {Chapple}, \citenamefont {Enderud}, \citenamefont
  {Hilton}, \citenamefont {Karimi}, \citenamefont {Ladizinsky}, \citenamefont
  {Ladizinsky}, \citenamefont {Oh}, \citenamefont {Perminov}, \citenamefont
  {Rich}, \citenamefont {Thom}, \citenamefont {Tolkacheva}, \citenamefont
  {Truncik}, \citenamefont {Uchaikin}, \citenamefont {Wang}, \citenamefont
  {Wilson},\ and\ \citenamefont {Rose}}]{Dwave}%
  \BibitemOpen
  \bibfield  {author} {\bibinfo {author} {\bibfnamefont {M.~W.}\ \bibnamefont
  {Johnson}}, \bibinfo {author} {\bibfnamefont {M.~H.~S.}\ \bibnamefont
  {Amin}}, \bibinfo {author} {\bibfnamefont {S.}~\bibnamefont {Gildert}},
  \bibinfo {author} {\bibfnamefont {T.}~\bibnamefont {Lanting}}, \bibinfo
  {author} {\bibfnamefont {F.}~\bibnamefont {Hamze}}, \bibinfo {author}
  {\bibfnamefont {N.}~\bibnamefont {Dickson}}, \bibinfo {author} {\bibfnamefont
  {R.}~\bibnamefont {Harris}}, \bibinfo {author} {\bibfnamefont {A.~J.}\
  \bibnamefont {Berkley}}, \bibinfo {author} {\bibfnamefont {J.}~\bibnamefont
  {Johansson}}, \bibinfo {author} {\bibfnamefont {P.}~\bibnamefont {Bunyk}},
  \bibinfo {author} {\bibfnamefont {E.~M.}\ \bibnamefont {Chapple}}, \bibinfo
  {author} {\bibfnamefont {C.}~\bibnamefont {Enderud}}, \bibinfo {author}
  {\bibfnamefont {J.~P.}\ \bibnamefont {Hilton}}, \bibinfo {author}
  {\bibfnamefont {K.}~\bibnamefont {Karimi}}, \bibinfo {author} {\bibfnamefont
  {E.}~\bibnamefont {Ladizinsky}}, \bibinfo {author} {\bibfnamefont
  {N.}~\bibnamefont {Ladizinsky}}, \bibinfo {author} {\bibfnamefont
  {T.}~\bibnamefont {Oh}}, \bibinfo {author} {\bibfnamefont {I.}~\bibnamefont
  {Perminov}}, \bibinfo {author} {\bibfnamefont {C.}~\bibnamefont {Rich}},
  \bibinfo {author} {\bibfnamefont {M.~C.}\ \bibnamefont {Thom}}, \bibinfo
  {author} {\bibfnamefont {E.}~\bibnamefont {Tolkacheva}}, \bibinfo {author}
  {\bibfnamefont {C.~J.~S.}\ \bibnamefont {Truncik}}, \bibinfo {author}
  {\bibfnamefont {S.}~\bibnamefont {Uchaikin}}, \bibinfo {author}
  {\bibfnamefont {J.}~\bibnamefont {Wang}}, \bibinfo {author} {\bibfnamefont
  {B.}~\bibnamefont {Wilson}},\ and\ \bibinfo {author} {\bibfnamefont
  {G.}~\bibnamefont {Rose}},\ }\bibfield  {title} {\emph {\bibinfo {title}
  {Quantum annealing with manufactured spins}},\ }\href
  {https://www.nature.com/nature/journal/v473/n7346/full/nature10012.html}
  {\bibfield  {journal} {\bibinfo  {journal} {Nature}\ }\textbf {\bibinfo
  {volume} {473}},\ \bibinfo {pages} {194--198} (\bibinfo {year}
  {2011})}\BibitemShut {NoStop}%
\bibitem [{\citenamefont {Boixo}\ \emph {et~al.}(2014)\citenamefont {Boixo},
  \citenamefont {Ronnow}, \citenamefont {Isakov}, \citenamefont {Wang},
  \citenamefont {Wecker}, \citenamefont {Lidar}, \citenamefont {Martinis},\
  and\ \citenamefont {Troyer}}]{q108}%
  \BibitemOpen
  \bibfield  {author} {\bibinfo {author} {\bibfnamefont {S.}~\bibnamefont
  {Boixo}}, \bibinfo {author} {\bibfnamefont {T.~F.}\ \bibnamefont {Ronnow}},
  \bibinfo {author} {\bibfnamefont {S.~V.}\ \bibnamefont {Isakov}}, \bibinfo
  {author} {\bibfnamefont {Z.}~\bibnamefont {Wang}}, \bibinfo {author}
  {\bibfnamefont {D.}~\bibnamefont {Wecker}}, \bibinfo {author} {\bibfnamefont
  {D.~A.}\ \bibnamefont {Lidar}}, \bibinfo {author} {\bibfnamefont {J.~M.}\
  \bibnamefont {Martinis}},\ and\ \bibinfo {author} {\bibfnamefont
  {M.}~\bibnamefont {Troyer}},\ }\bibfield  {title} {\emph {\bibinfo {title}
  {Evidence for quantum annealing with more than one hundred qubits}},\ }\href
  {https://doi.org/10.1038/nphys2900} {\bibfield  {journal} {\bibinfo
  {journal} {Nat. Phys.}\ }\textbf {\bibinfo {volume} {10}},\ \bibinfo {pages}
  {218--224} (\bibinfo {year} {2014})}\BibitemShut {NoStop}%
\bibitem [{\citenamefont {Boixo}\ \emph {et~al.}(2016)\citenamefont {Boixo},
  \citenamefont {Smelyanskiy}, \citenamefont {Shabani}, \citenamefont {Isakov},
  \citenamefont {Dykman}, \citenamefont {Denchev}, \citenamefont {Amin},
  \citenamefont {Smirnov}, \citenamefont {Mohseni},\ and\ \citenamefont
  {Neven}}]{Boixo:2014yu}%
  \BibitemOpen
  \bibfield  {author} {\bibinfo {author} {\bibfnamefont {S.}~\bibnamefont
  {Boixo}}, \bibinfo {author} {\bibfnamefont {V.~N.}\ \bibnamefont
  {Smelyanskiy}}, \bibinfo {author} {\bibfnamefont {A.}~\bibnamefont
  {Shabani}}, \bibinfo {author} {\bibfnamefont {S.~V.}\ \bibnamefont {Isakov}},
  \bibinfo {author} {\bibfnamefont {M.}~\bibnamefont {Dykman}}, \bibinfo
  {author} {\bibfnamefont {V.~S.}\ \bibnamefont {Denchev}}, \bibinfo {author}
  {\bibfnamefont {M.~H.}\ \bibnamefont {Amin}}, \bibinfo {author}
  {\bibfnamefont {A.~Y.}\ \bibnamefont {Smirnov}}, \bibinfo {author}
  {\bibfnamefont {M.}~\bibnamefont {Mohseni}},\ and\ \bibinfo {author}
  {\bibfnamefont {H.}~\bibnamefont {Neven}},\ }\bibfield  {title} {\emph
  {\bibinfo {title} {Computational multiqubit tunnelling in programmable
  quantum annealers}},\ }\href {http://dx.doi.org/10.1038/ncomms10327}
  {\bibfield  {journal} {\bibinfo  {journal} {Nat Commun}\ }\textbf {\bibinfo
  {volume} {7}} (\bibinfo {year} {2016})}\BibitemShut {NoStop}%
\bibitem [{\citenamefont {Albash}\ and\ \citenamefont
  {Lidar}(2018{\natexlab{b}})}]{Albash:2017aa}%
  \BibitemOpen
  \bibfield  {author} {\bibinfo {author} {\bibfnamefont {T.}~\bibnamefont
  {Albash}}\ and\ \bibinfo {author} {\bibfnamefont {D.~A.}\ \bibnamefont
  {Lidar}},\ }\bibfield  {title} {\emph {\bibinfo {title} {Demonstration of a
  scaling advantage for a quantum annealer over simulated annealing}},\ }\href
  {https://doi.org/10.1103/PhysRevX.8.031016} {\bibfield  {journal} {\bibinfo
  {journal} {Physical Review X}\ }\textbf {\bibinfo {volume} {8}},\ \bibinfo
  {pages} {031016--} (\bibinfo {year} {2018}{\natexlab{b}})}\BibitemShut
  {NoStop}%
\bibitem [{\citenamefont {Mandr{\`a}}\ and\ \citenamefont
  {Katzgraber}(2018)}]{Mandra:2017ab}%
  \BibitemOpen
  \bibfield  {author} {\bibinfo {author} {\bibfnamefont {S.}~\bibnamefont
  {Mandr{\`a}}}\ and\ \bibinfo {author} {\bibfnamefont {H.~G.}\ \bibnamefont
  {Katzgraber}},\ }\bibfield  {title} {\emph {\bibinfo {title} {A deceptive
  step towards quantum speedup detection}},\ }\href
  {https://doi.org/10.1088/2058-9565/aac8b2} {\bibfield  {journal} {\bibinfo
  {journal} {Quantum Sci. Technol.}\ }\textbf {\bibinfo {volume} {3}},\
  \bibinfo {pages} {04LT01} (\bibinfo {year} {2018})}\BibitemShut {NoStop}%
\bibitem [{\citenamefont {Kowalsky}\ \emph {et~al.}(2022)\citenamefont
  {Kowalsky}, \citenamefont {Albash}, \citenamefont {Hen},\ and\ \citenamefont
  {Lidar}}]{kowalsky20213regular}%
  \BibitemOpen
  \bibfield  {author} {\bibinfo {author} {\bibfnamefont {M.}~\bibnamefont
  {Kowalsky}}, \bibinfo {author} {\bibfnamefont {T.}~\bibnamefont {Albash}},
  \bibinfo {author} {\bibfnamefont {I.}~\bibnamefont {Hen}},\ and\ \bibinfo
  {author} {\bibfnamefont {D.~A.}\ \bibnamefont {Lidar}},\ }\bibfield  {title}
  {\emph {\bibinfo {title} {3-regular three-xorsat planted solutions benchmark
  of classical and quantum heuristic optimizers}},\ }\href
  {https://doi.org/10.1088/2058-9565/ac4d1b} {\bibfield  {journal} {\bibinfo
  {journal} {Quantum Science and Technology}\ }\textbf {\bibinfo {volume}
  {7}},\ \bibinfo {pages} {025008} (\bibinfo {year} {2022})}\BibitemShut
  {NoStop}%
\bibitem [{\citenamefont {Knill}\ and\ \citenamefont
  {Laflamme}(1997)}]{Knill:1997kx}%
  \BibitemOpen
  \bibfield  {author} {\bibinfo {author} {\bibfnamefont {E.}~\bibnamefont
  {Knill}}\ and\ \bibinfo {author} {\bibfnamefont {R.}~\bibnamefont
  {Laflamme}},\ }\bibfield  {title} {\emph {\bibinfo {title} {Theory of quantum
  error-correcting codes}},\ }\href
  {http://link.aps.org/doi/10.1103/PhysRevA.55.900} {\bibfield  {journal}
  {\bibinfo  {journal} {Phys. Rev. A}\ }\textbf {\bibinfo {volume} {55}},\
  \bibinfo {pages} {900--911} (\bibinfo {year} {1997})}\BibitemShut {NoStop}%
\bibitem [{\citenamefont {Bacon}(2006)}]{Bacon:05}%
  \BibitemOpen
  \bibfield  {author} {\bibinfo {author} {\bibfnamefont {D.}~\bibnamefont
  {Bacon}},\ }\bibfield  {title} {\emph {\bibinfo {title} {Operator quantum
  error-correcting subsystems for self-correcting quantum memories}},\ }\href
  {https://doi.org/10.1103/PhysRevA.73.012340} {\bibfield  {journal} {\bibinfo
  {journal} {Phys. Rev. A}\ }\textbf {\bibinfo {volume} {73}},\ \bibinfo
  {pages} {012340} (\bibinfo {year} {2006})}\BibitemShut {NoStop}%
\bibitem [{\citenamefont {Aliferis}\ and\ \citenamefont
  {Cross}(2007)}]{Aliferis:07}%
  \BibitemOpen
  \bibfield  {author} {\bibinfo {author} {\bibfnamefont {P.}~\bibnamefont
  {Aliferis}}\ and\ \bibinfo {author} {\bibfnamefont {A.~W.}\ \bibnamefont
  {Cross}},\ }\bibfield  {title} {\emph {\bibinfo {title} {Subsystem fault
  tolerance with the bacon-shor code}},\ }\href
  {https://doi.org/10.1103/PhysRevLett.98.220502} {\bibfield  {journal}
  {\bibinfo  {journal} {Phys. Rev. Lett.}\ }\textbf {\bibinfo {volume} {98}},\
  \bibinfo {pages} {220502} (\bibinfo {year} {2007})}\BibitemShut {NoStop}%
\bibitem [{\citenamefont {Burton}(2018)}]{burton2018spectra}%
  \BibitemOpen
  \bibfield  {author} {\bibinfo {author} {\bibfnamefont {S.}~\bibnamefont
  {Burton}},\ }\bibfield  {title} {\emph {\bibinfo {title} {Spectra of gauge
  code hamiltonians}},\ }\href {https://arxiv.org/abs/1801.03243} {\bibfield
  {journal} {\bibinfo  {journal} {arXiv:1801.03243}\ } (\bibinfo {year}
  {2018})}\BibitemShut {NoStop}%
\bibitem [{\citenamefont {Nussinov}\ and\ \citenamefont {van~den
  Brink}(2015)}]{Nussinov:2015aa}%
  \BibitemOpen
  \bibfield  {author} {\bibinfo {author} {\bibfnamefont {Z.}~\bibnamefont
  {Nussinov}}\ and\ \bibinfo {author} {\bibfnamefont {J.}~\bibnamefont {van~den
  Brink}},\ }\bibfield  {title} {\emph {\bibinfo {title} {Compass models:
  Theory and physical motivations}},\ }\href
  {https://doi.org/10.1103/RevModPhys.87.1} {\bibfield  {journal} {\bibinfo
  {journal} {Reviews of Modern Physics}\ }\textbf {\bibinfo {volume} {87}},\
  \bibinfo {pages} {1--59} (\bibinfo {year} {2015})}\BibitemShut {NoStop}%
\bibitem [{\citenamefont {Pfeuty}(1970)}]{PFEUTY197079}%
  \BibitemOpen
  \bibfield  {author} {\bibinfo {author} {\bibfnamefont {P.}~\bibnamefont
  {Pfeuty}},\ }\bibfield  {title} {\emph {\bibinfo {title} {The one-dimensional
  ising model with a transverse field}},\ }\href
  {https://doi.org/https://doi.org/10.1016/0003-4916(70)90270-8} {\bibfield
  {journal} {\bibinfo  {journal} {Annals of Physics}\ }\textbf {\bibinfo
  {volume} {57}},\ \bibinfo {pages} {79--90} (\bibinfo {year}
  {1970})}\BibitemShut {NoStop}%
\bibitem [{\citenamefont {Brzezicki}\ \emph {et~al.}(2007)\citenamefont
  {Brzezicki}, \citenamefont {Dziarmaga},\ and\ \citenamefont
  {Ole{\'s}}}]{Brzezicki:2007aa}%
  \BibitemOpen
  \bibfield  {author} {\bibinfo {author} {\bibfnamefont {W.}~\bibnamefont
  {Brzezicki}}, \bibinfo {author} {\bibfnamefont {J.}~\bibnamefont
  {Dziarmaga}},\ and\ \bibinfo {author} {\bibfnamefont {A.~M.}\ \bibnamefont
  {Ole{\'s}}},\ }\bibfield  {title} {\emph {\bibinfo {title} {Quantum phase
  transition in the one-dimensional compass model}},\ }\href
  {https://doi.org/10.1103/PhysRevB.75.134415} {\bibfield  {journal} {\bibinfo
  {journal} {Physical Review B}\ }\textbf {\bibinfo {volume} {75}},\ \bibinfo
  {pages} {134415--} (\bibinfo {year} {2007})}\BibitemShut {NoStop}%
\bibitem [{\citenamefont {Brzezicki}\ and\ \citenamefont
  {Ole\ifmmode~\acute{s}\else \'{s}\fi{}}(2013)}]{PhysRevB.87.214421}%
  \BibitemOpen
  \bibfield  {author} {\bibinfo {author} {\bibfnamefont {W.}~\bibnamefont
  {Brzezicki}}\ and\ \bibinfo {author} {\bibfnamefont {A.~M.}\ \bibnamefont
  {Ole\ifmmode~\acute{s}\else \'{s}\fi{}}},\ }\bibfield  {title} {\emph
  {\bibinfo {title} {Symmetry properties and spectra of the two-dimensional
  quantum compass model}},\ }\href {https://doi.org/10.1103/PhysRevB.87.214421}
  {\bibfield  {journal} {\bibinfo  {journal} {Phys. Rev. B}\ }\textbf {\bibinfo
  {volume} {87}},\ \bibinfo {pages} {214421} (\bibinfo {year}
  {2013})}\BibitemShut {NoStop}%
\bibitem [{\citenamefont {Brzezicki}\ and\ \citenamefont
  {Ole\ifmmode~\acute{s}\else \'{s}\fi{}}(2009)}]{PhysRevB.80.014405}%
  \BibitemOpen
  \bibfield  {author} {\bibinfo {author} {\bibfnamefont {W.}~\bibnamefont
  {Brzezicki}}\ and\ \bibinfo {author} {\bibfnamefont {A.~M.}\ \bibnamefont
  {Ole\ifmmode~\acute{s}\else \'{s}\fi{}}},\ }\bibfield  {title} {\emph
  {\bibinfo {title} {Exact solution for a quantum compass ladder}},\ }\href
  {https://doi.org/10.1103/PhysRevB.80.014405} {\bibfield  {journal} {\bibinfo
  {journal} {Phys. Rev. B}\ }\textbf {\bibinfo {volume} {80}},\ \bibinfo
  {pages} {014405} (\bibinfo {year} {2009})}\BibitemShut {NoStop}%
\bibitem [{\citenamefont {Dorier}\ \emph {et~al.}(2005)\citenamefont {Dorier},
  \citenamefont {Becca},\ and\ \citenamefont {Mila}}]{PhysRevB.72.024448}%
  \BibitemOpen
  \bibfield  {author} {\bibinfo {author} {\bibfnamefont {J.}~\bibnamefont
  {Dorier}}, \bibinfo {author} {\bibfnamefont {F.}~\bibnamefont {Becca}},\ and\
  \bibinfo {author} {\bibfnamefont {F.}~\bibnamefont {Mila}},\ }\bibfield
  {title} {\emph {\bibinfo {title} {Quantum compass model on the square
  lattice}},\ }\href {https://doi.org/10.1103/PhysRevB.72.024448} {\bibfield
  {journal} {\bibinfo  {journal} {Phys. Rev. B}\ }\textbf {\bibinfo {volume}
  {72}},\ \bibinfo {pages} {024448} (\bibinfo {year} {2005})}\BibitemShut
  {NoStop}%
\bibitem [{\citenamefont {Nannicini}\ \emph {et~al.}(2022)\citenamefont
  {Nannicini}, \citenamefont {Bishop}, \citenamefont {G{\"u}nl{\"u}k},\ and\
  \citenamefont {Jurcevic}}]{nannicini2022optimal}%
  \BibitemOpen
  \bibfield  {author} {\bibinfo {author} {\bibfnamefont {G.}~\bibnamefont
  {Nannicini}}, \bibinfo {author} {\bibfnamefont {L.~S.}\ \bibnamefont
  {Bishop}}, \bibinfo {author} {\bibfnamefont {O.}~\bibnamefont
  {G{\"u}nl{\"u}k}},\ and\ \bibinfo {author} {\bibfnamefont {P.}~\bibnamefont
  {Jurcevic}},\ }\bibfield  {title} {\emph {\bibinfo {title} {Optimal qubit
  assignment and routing via integer programming}},\ }\href
  {https://dl.acm.org/doi/abs/10.1145/3544563} {\bibfield  {journal} {\bibinfo
  {journal} {ACM Transactions on Quantum Computing}\ }\textbf {\bibinfo
  {volume} {4}},\ \bibinfo {pages} {1--31} (\bibinfo {year}
  {2022})}\BibitemShut {NoStop}%
\end{thebibliography}%

\onecolumn\newpage
\appendix

\section{Equivalence of $A$ matrix via permuting rows and columns}\label{app:A-perm}
There exist multiple approaches to constructing Bravyi's $A$ matrix corresponding to a specified set of code parameters $[[n,k,g,d]]$. \cref{lem:equiv-A} below demonstrates that these various forms are equivalent, as the code parameters are preserved under both row and column permutations. Specifically, while gauge generators undergo modifications through permutations, the gauge group and stabilizers remain invariant.

\begin{mylemma}
\label{lem:equiv-A}
A code associated with a given $A$ matrix is invariant under permutations of rows and columns of the same $A$ matrix. I.e., the gauge group is invariant under permutations. 
\end{mylemma}

\begin{proof}
    
    Any permutations of rows and columns can be decomposed into a series of transpositions, $\prod_{i,j}(\sigma_i,\sigma_j)$, where $\sigma\in\{r,c\}$ denotes a row or a column. Without loss of generality, consider a transposition of two rows, $(r_i,r_j)$. This does not change $X$-type gauge operators because the matrix structure is preserved within the same row. $Z$-type gauge operators are pairs of $Z$ operators in the same column. Let $Z_{(i,a)}Z_{(k\neq i,a)}$ be a gauge operator such that one of its operators is in the $i$'th row. If $k=j$, the transposition $(r_i,r_j)$ does not change the gauge operator. Otherwise, this gauge operator becomes $Z_{(j,a)}Z_{(k\neq j,a)}$. However, the column indices are dummy variables and can be relabelled without affecting physical operations. We recover the original gauge operators after relabelling the indices $i\rightarrow j,j\rightarrow i$. Since any transposition of any rows and columns preserves the gauge group $\mathcal{G}$, the permutations of any rows and columns also preserve the gauge group $\mathcal{G}$.
\end{proof}

To illustrate the result in \cref{lem:equiv-A}, the three matrices in \cref{eq:equiv} are associated with the same code. In this work, we use the `trapezoid' representation on the left hand side. The 1-entries are indexed to show that the gauge group $\mathcal{G}$ is preserved, e.g., $X_aX_b$ is always presented across all permutations.
\begin{equation}
\left[\begin{array}{ccccccc}
    \mbf{1}_a&\mbf{1}_b&\txg{0}&\txg{0}&\txg{0}\\
    \mbf{1}_j&\txg{0}&\mbf{1}_i&\txg{0}&\txg{0}\\
    \txg{0}&\mbf{1}_c&\txg{0}&\mbf{1}_d&\txg{0}\\
    \txg{0}&\txg{0}&\mbf{1}_h&\txg{0}&\mbf{1}_g\\
    \txg{0}&\txg{0}&\txg{0}&\mbf{1}_e&\mbf{1}_f\\
    \end{array}    \right]
\equiv
\left[\begin{array}{ccccccc}
    \txg{0}&\mbf{1}_a&\txg{0}&\txg{0}&\mbf{1}_b\\
    \mbf{1}_f&\txg{0}&\mbf{1}_e&\txg{0}&\txg{0}\\
    \txg{0}&\mbf{1}_j&\txg{0}&\mbf{1}_i&\txg{0}\\
    \txg{0}&\txg{0}&\mbf{1}_d&\txg{0}&\mbf{1}_c\\
    \mbf{1}_g&\txg{0}&\txg{0}&\mbf{1}_h&\txg{0}\\
    \end{array}    \right]
\equiv
 \left[\begin{array}{ccccccc}
    \mbf{1}_a&\mbf{1}_b&\txg{0}&\txg{0}&\txg{0}\\
    \txg{0}&\mbf{1}_c&\mbf{1}_d&\txg{0}&\txg{0}\\
    \txg{0}&\txg{0}&\mbf{1}_e&\mbf{1}_f&\txg{0}\\
    \txg{0}&\txg{0}&\txg{0}&\mbf{1}_g&\mbf{1}_h\\
    \mbf{1}_j&\txg{0}&\txg{0}&\txg{0}&\mbf{1}_i\\
    \end{array}    \right]
    \label{eq:equiv}
\end{equation}
The second matrix is obtained by permuting the columns $(1,2,3,4,5)$ into $(2,5,4,3,1)$ and rows $(1,2,3,4,5)$ into $(1,3,4,5,2)$ from the first matrix. The third matrix is obtained by permuting the columns $(1,2,3,4,5)$ into $(1,2,5,3,4)$ and rows $(1,2,3,4,5)$ into $(1,5,2,4,3)$ from the first matrix. The gauge generators associated with the three matrices are the same, as follows:
\begin{equation}
\mathcal{G}=\<X_aX_b, X_jX_i, X_cX_d, X_hX_g, X_eX_f, Z_aZ_j, Z_bZ_c, Z_iZ_h, Z_dZ_e, Z_gZ_f\>.
\end{equation}

\section{Reduction of search space's cardinality in searching for optimal logical operators}
\label{app:reduce_log}
Recall that \cref{th:x-z-op} and \cref{th:x-z-op-dressed} establish how a bitstring parameterizes bare and dressed logical operators, respectively. Our goal is to construct the $m$ optimal pairs of logical operators, $\hat{X}$ and $\hat{Z}$. By `optimal pairs,' we refer to those minimizing the weight of $\hat{X}\hat{X}$ and $\hat{Z}\hat{Z}$ terms. However, for a code defined by an $m \times m$ $A$ matrix, there are $2^m$ possible configurations of bitstrings of length $m$ that can parameterize logical operators, resulting in an exponentially large search space. 

In this section, we exploit symmetries to reduce the search space and systematically exclude bitstrings that generate operators with locality greater than two. Let $\mathcal{S}_0$ denote the set of all possible bitstrings of length $m$, i.e., $\mathcal{S}_0=\{\vect{a}\in\{0,1\}^m\}$. The search space for logical operators of both $X$- and $Z$-type is given by $\mathcal{S}_0 \times \mathcal{S}_0$, which is exponentially large, as $|\mathcal{S}_0|= 2^m$.

Ideally, we want to achieve the highest code rate with purely $2$-local $\hat{X},\hat{Z},\hat{X}\hat{X}$ and $\hat{Z}\hat{Z}$ operators. However, since Ref.~\cite{marvian2019robust} has shown that the trapezoid code with $l=k$ has two-body interactions for every $\bar{X},\bar{Z},\hat{X}\hat{X}$ and $\hat{Z}\hat{Z}$, it is unavoidable for $l<k$ codes with higher rates to have higher than $2$-local $\hat{X}\hat{X}$ and $\hat{Z}\hat{Z}$ operators, although all $\bar{X},\bar{Z}$ operators are $2$-local.

Using symmetries, we apply two successive reductions: $\mathcal{S}_0\rightarrow\mathcal{S}_1\rightarrow\mathcal{S}_2$, which reduce the cardinality of the search space from $|\mathcal{S}_0 \times \mathcal{S}_0|=2^{2m}$ to $|\mathcal{S}_2 \times \mathcal{S}_2|=m^2(m-1)^2/4$, as summarized below in \cref{lem:XXZZ-size}. Using this result on $l<k$ codes, we search the space to find $m$ optimal pairs of $2$-local logical operators $\hat{X}$ and $\hat{Z}$, such that the weight of $\hat{X}\hat{X}$ and $\hat{Z}\hat{Z}$ terms is minimized.

\begin{mytheorem}
    \label{lem:XXZZ-size}
    The size of the search space for $2$-local logical operators is 
    $|\mathcal{S}_2 \times \mathcal{S}_2|=m^2(m-1)^2/4$, where $\mathcal{S}_2$ is a set of binary bitstrings that parameterize $2$-local logical operators.   
\end{mytheorem}

\begin{proof}
     Using \cref{lem:m-not-l}, there are $m(m-1)/2$ bitstrings that parameterize $2$-local $\hat{Z}$ operators and likewise $\hat{X}$ operators. The reduced search space is therefore $\mathcal{S}_2 \times \mathcal{S}_2$ and its cardinality is $|\mathcal{S}_2 \times \mathcal{S}_2|=m^2(m-1)^2/4$.
\end{proof}

The rest of this section is devoted to building up to and proving \cref{lem:m-not-l}. 

Note that odd-$m$ codes consistently exhibit a higher code rate compared to even-$m$ codes [\cref{eq:rate}], making them more desirable. Consequently, the results in this section focus on odd-$m$ codes. Readers interested in even-$m$ codes can apply analogous reasoning to derive the corresponding results.

\subsection{Reduction by stabilizer symmetry ($\mathcal{S}_0\rightarrow\mathcal{S}_1$)}
Let $\vect{x} = x_1x_2\dots x_m$ and $\vect{z} = z_1z_2\dots z_m$ represent length-$m$ binary bitstrings, where $x_i, z_i \in \{0, 1\}$. We demonstrate that it is sufficient to restrict our consideration to $\vect{x}$ with $x_1 = 0$ and $\vect{z}$ with $z_m = 0$. This restriction resolves the degeneracy arising from equivalence under multiplication by the stabilizers $S_X$ and $S_Z$, as demonstrated in \cref{lem:stab-equiv} below. Imposing this symmetry reduces the search space to $\mathcal{S}_1 \times \mathcal{S}_1 = \{\vect{x} \in \{0,1\}^m | x_1 = 0\} \times \{\vect{z} \in \{0,1\}^m | z_m = 0\}$, which is half the size of the original space: $ |\mathcal{S}_1| = |\mathcal{S}_0| / 2$.

\begin{mylemma}
\label{lem:stab-equiv}
    $\vect{x}$ ($\vect{z}$) and $\tilde{\vect{x}}=1^m\oplus\vect{x}$ $(\tilde{\vect{z}}=1^m\oplus\vect{z}$) parameterize $X$-type ($Z$-type) operators that are equivalent up to multiplication by a stabilizer. 
\end{mylemma}
\begin{proof}
    Without loss of generality, consider an $X$-type operator, $P^X(\vect{x})$, parameterized by a bitstring $\vect{x} = x_1x_2\dots x_m$. Writing it in terms of column operators, $P^X(\vect{x})=\prod_{i=1}^m C_i^{x_i}$. Recall from \cref{lem:stab} that the stabilizer $S_X=\prod_{i=1}^m C_i$, i.e., $S_X=P^X(1^m)$.

    We can also write $P^X(\tilde{\vect{x}})=\prod_{i=1}^m C_i^{\tilde{x}_i}=\prod_{i=1}^m C_i^{1\oplus x_i}=\prod_{i=1}^m C_i\prod_{i=1}^m C_i^{x_i}=S_XP^X(\vect{x})$, which shows that $P^X(\vect{x})$ and $P^X(\tilde{\vect{x}})$ are equivalent up to a stabilizer multiplication. 
    
    The proof for $Z$-type operators follows by replacing the bitstring $\vect{x}$ by $\vect{z}$, $S_X$ by $S_Z$ and the column operator $C_i$ with the row operator $R_i$.
\end{proof}

\subsection{Reduction to $2$-local operators ($\mathcal{S}_1\rightarrow\mathcal{S}_2$)}

Constructing $2$-local operators is desirable since performing operations with higher-weight interactions is less practical. Here, we further constrain the search space to contain only bitstrings that parameterize operators that are $2$-local. \cref{lem:$2$-local-Z} and \cref{lem:$2$-local-X} provide a list of rules for constructing these bitstrings for $Z$- and $X$-type operators, respectively. \cref{app:exp-reduction} provides an example of valid sets of $\hat{Z}$ and $\hat{X}$ logical operators for the $[[16,6,8,2]]$ code constructed using these rules. 

For the convenience of proofs in this section, we rewrite, in \cref{eq:general-A-appB}, the structure of a $m\times m$ $A$ matrix from \cref{eq:general-A}. Likewise, \cref{lem:gauge-Z} explicitly lists out the $Z$ gauge generators, which follows directly from \cref{def:gauge}.

\begin{align}
\label{eq:general-A-appB}
&\qquad\quad A= \notag\\ 
&\begin{array}{llcccccllllll}
    \scriptstyle{1}&\vline&1&1&&&&&&&&&\\
    &\vline&1 &&1&&&&&&\\
    &\vline&\vdots &&&\ddots&&&&&\\
    \scriptstyle{2l}&\vline&1&&&&1&&&&&\\ 
    \scriptstyle{2l+1}&\vline&&1&&&&1&&&\\
    &\vline&&&\ddots&&&&\ddots&&&\\
    &\vline&&&&1&&&&1&\\
    \scriptstyle{m-1}&\vline&&&&&1&&&&1\\
    \scriptstyle{m}&\vline&&&&&&1&\cdots&1&1\\
    \hline
    &\vline&\scriptstyle{1}&&&&&\hspace{-.5cm}\scriptstyle{m-2l+1}&\hspace{.5cm}&&\hspace{-.03cm}\scriptstyle{m}&\\
\end{array}
\end{align}

\begin{mycorollary}
\label{lem:gauge-Z}
 The complete set of $Z$-type gauge generators of the gauge group $\mathcal{G}$ is: 
     \begin{itemize}
        \item first column: $Z_{i,1}Z_{i+1,1}$ for $1 \le i\le 2l-1$ 
        \item columns 2 to $m-2l$: $Z_{j-1,j}Z_{j-1+2l,j}$ for $2 \le j\le m-2l$
        \item columns $m-2l+1$ to $m-1$: $Z_{j-1,j}Z_{m,j}$ for $ m-2l+1\le j \le m-1$
    \end{itemize}
\end{mycorollary}
\begin{proof}
    Following \cref{def:gauge}, we construct $Z$-type gauge generators by writing out pairs of $Z_iZ_j$ in the same column of the $A$ matrix in \cref{eq:general-A-appB}.
    
    The first column has $1$'s in rows $1$ to $2l$. Since $\mathcal{G}$ is closed under multiplication, we only need to consider consecutive pairs of $(i,i+1)$ in this column. In the rest of the columns $i\in[2,m-1]$, the gauge generators are pairs of two $Z$ operators in the same columns. 
    
    Note that the gauge operator $Z_{m-1,m}Z_{m,m}$ in the last column is not a generator because it can be obtained from the multiplication of all the other generators and the $Z$-stabilizer. 
\end{proof}

In \cref{lem:$2$-local-Z}, we show the $2$-local reduction for $\hat{Z}$ operators, and later in \cref{lem:$2$-local-X}, we generalize this result to $\hat{X}$ operators.

\begin{mylemma}
\label{lem:$2$-local-Z}
    For a $[[4k+2l,2k,g,2]]$ trapezoid code, a $2$-local $\hat{Z}$ operator can be constructed from a bitstring $\vect{z}=z_1z_2\dots z_m$ if and only if the bitstring has $1$ entries according to one of the following rules and $0$ entries in the remaining positions.
    \begin{enumerate}
        \item[(1)] $\forall n\in[0,N]$: $z_{i+2nl}=1$ where $i\in[1,m-1]$, $N\in[0,M]$ and $M=\lfloor \frac{m-1-i}{2l}\rfloor$
        \item[(2)] $\forall n_{(i,j)}\in[0,N_{(i,j)}]$: $z_{i+2n_il}=z_{j+2n_jl}=1$ where $i,j\in[1,2l]$, $i<j$, $N_{(i,j)}\in[0,M_{(i,j)}]$, $M_{(i,j)}=\lfloor \frac{m-1-(i,j)}{2l}\rfloor$
    \end{enumerate}
\end{mylemma} 

\begin{proof}
    ($\Rightarrow$) First, we show that if a $\vect{z}$ follows one of the rules, it parameterizes a $2$-local $\hat{Z}$. We  start by considering the special case of each rule, and then generalize it.

    Consider the special case of Rule (1) when $N=0$, a $\vect{z}$ has Hamming weight $1$, hence, $\hat{Z}=\bar{Z}$ is $2$-local by construction. For a general case with $N>0$, $\forall n\in[1,N]$, multiplying by the gauges $Z_{i+2(n-1)l,i+1}Z_{i+2nl,i+1}$ removes the $Z$ operators in the same column. The remaining operators are $Z_{i,1}$ and $Z_{i+2nl,i+2(n+1)l}$, which constructs a $2$-local $\hat{Z}$. 
    
    Next, consider the special case of Rule (2) when $N_i=N_j=0$, $\hat{Z}$ is $2$-local after multiplying by the gauges $Z_{i,1}Z_{j,1}$. For a general case with $N_{i}>0$, $\forall n_j\in[1,N_j]$, multiplying by the gauges $Z_{i+2(n_j-1)l,i+1}Z_{i+2n_jl,i+1}$ removes the $Z$ operators in the columns $i+1$. Then, we can multiply $Z$ gauges in every other $2l$ column according to $n_j$. The same cancellation method applies for $N_j>0$, replacing the index $i\rightarrow j,~j\rightarrow i$. The remaining operators are $Z_{i+2n_jl,i+2(n_j+1)l}$ and $Z_{j+2n_jl,j+2(n_j+1)l}$, which constructs a $2$-local $\hat{Z}$.
    
    ($\Leftarrow$) We  now show that if a $\hat{Z}$ is $2$-local, it must be parameterized by a bitstring that follows one of the rules. We  show this by considering all possible $Z$ cancellations by multiplying $\bar{Z}$ with all possible combinations of $Z$-type gauges in \cref{lem:gauge-Z}.
    
    Recall the reduced search space $\mathcal{S}_1=\{\vect{z}| z_m=0\}$. This means that $\{\vect{z}|z_m=1\}\not\subset \mathcal{S}_1$, i.e., $Z$ operators in the last row are not involved. Hence, we do not need to consider $2$-local $\hat{Z}$ that comes from $Z$ cancellations involving the third type of gauge operators in \cref{lem:gauge-Z}, since those gauges contain $Z_{m,j}$ which is in the last row of the $A$ matrix. 
    
    By construction, the $A$ matrix has two $Z$ physical operators in each remaining row. This implies that a bitstring $\vect{z}$ of Hamming weight $w$ parameterizes a $2w$-local $\bar{Z}$. To construct a $2$-local $\hat{Z}$, we must multiply gauge operators with a total of $2(w-1)$ overlapping physical operations.

    When $w=1$, the gauge operator is not required. This falls into the special case of Rule (1). To use the first type of gauge operators in \cref{lem:gauge-Z}, one requires overlaps in multiples of $2l$, i.e., the $(i+2nl)$'th rows, which is covered in the general case of Rule (1). To use the second type of gauge operators in \cref{lem:gauge-Z}, the overlap is constructed using the special case of Rule (2). When both types of gauge operators are used, one requires overlaps in multiples of $2l$, but starting with both rows $i$'th and $j$'th, hence, falls into the general case of Rule (2).

\end{proof}

To illustrate the procedures in \cref{lem:$2$-local-Z}, \cref{fig:rules} shows a diagram representing the special and general cases of both rules. The example is constructed for a $[[16,6,8,2]]$ code ($m=7,l=2$). The special case of Rule (1) is $2$-local by construction (top left), hence, does not require gauge cancellations. The special case of Rule (2) requires gauge cancellation in the first column (top right). The general case of each rule (bottom row) adopts the gauge cancellation from its special case and requires extra cancellations in each column with two $Z$ operators.

\begin{figure}[t]
\centering
\begin{tikzpicture}\label{fig:rule-1}
    \matrix (m)[
    matrix of math nodes,
    nodes in empty cells,
    left delimiter=\lbrack,
    right delimiter=\rbrack
    ] {
    1&1&&&&&\\
    1&&1&&&&\\
    1&&&1&&&\\
    1&&&&1&&\\
    &1&&&&1&\\
    &&1&&&&1\\
    &&&1&1&1&1\\
    } ;
    \node [below of= m-7-3, xshift=0.4cm, node distance = 0.5cm] {$\vect{z}=(0,1,0,0,0,0,0)$}; 
    \begin{scope}[rounded corners,fill=blue!80!white,fill opacity=0.1,draw=blue,draw opacity=0.6,thick]
    \filldraw (m-2-1.north west)  rectangle (m-2-3.south east);
   \end{scope}
   \begin{scope}[fill=red!80!white,fill opacity=0.1,draw=red,thick]
    \filldraw (m-2-1)  circle (8pt);
    \filldraw (m-2-3)  circle (8pt);
   \end{scope}
    \end{tikzpicture}
    \begin{tikzpicture}\label{fig:rule-2}
    \matrix (m)[
    matrix of math nodes,
    nodes in empty cells,
    left delimiter=\lbrack,
    right delimiter=\rbrack
    ] {
    1&1&&&&&\\
    1&&1&&&&\\
    1&&&1&&&\\
    1&&&&1&&\\
    &1&&&&1&\\
    &&1&&&&1\\
    &&&1&1&1&1\\
    } ;
    \node [below of= m-7-3, xshift=0.4cm, node distance = 0.5cm] {$\vect{z}=(1,0,1,0,0,0,0)$}; 
    \begin{scope}[rounded corners,fill=blue!80!white,fill opacity=0.1,draw=blue,draw opacity=0.6,thick]
    \filldraw (m-1-1.north west)  rectangle (m-1-2.south east);
    \filldraw (m-3-1.north west)  rectangle (m-3-4.south east);
   \end{scope}
    \begin{scope}[rounded corners,fill=orange!80!white,fill opacity=0.1,draw=orange,thick]
    \filldraw (m-1-1.north west)  rectangle (m-3-1.south east);
   \end{scope}
   \begin{scope}[fill=red!80!white,fill opacity=0.1,draw=red,thick]
    \filldraw (m-1-2)  circle (8pt);
    \filldraw (m-3-4)  circle (8pt);
   \end{scope}
    \end{tikzpicture}
    \begin{tikzpicture}\label{fig:rule-3}
    \matrix (m)[
    matrix of math nodes,
    nodes in empty cells,
    left delimiter=\lbrack,
    right delimiter=\rbrack
    ] {
    1&1&&&&&\\
    1&&1&&&&\\
    1&&&1&&&\\
    1&&&&1&&\\
    &1&&&&1&\\
    &&1&&&&1\\
    &&&1&1&1&1\\
    } ;
    \node [below of= m-7-3, xshift=0.4cm, node distance = 0.5cm] {$\vect{z}=(1,0,0,0,1,0,0)$}; 
    \begin{scope}[rounded corners,fill=blue!80!white,fill opacity=0.1,draw=blue,draw opacity=0.6,thick]
    \filldraw (m-1-1.north west)  rectangle (m-1-2.south east);
    \filldraw (m-5-2.north west)  rectangle (m-5-6.south east);
   \end{scope}
    \begin{scope}[rounded corners,fill=orange!80!white,fill opacity=0.1,draw=orange,thick]
    \filldraw (m-1-2.north west)  rectangle (m-5-2.south east);
   \end{scope}
   \begin{scope}[fill=red!80!white,fill opacity=0.1,draw=red,thick]
    \filldraw (m-1-1)  circle (8pt);
    \filldraw (m-5-6)  circle (8pt);
   \end{scope}
    \end{tikzpicture}
    \begin{tikzpicture}\label{fig:rule-4}
    \matrix (m)[
    matrix of math nodes,
    nodes in empty cells,
    left delimiter=\lbrack,
    right delimiter=\rbrack
    ] {
    1&1&&&&&\\
    1&&1&&&&\\
    1&&&1&&&\\
    1&&&&1&&\\
    &1&&&&1&\\
    &&1&&&&1\\
    &&&1&1&1&1\\
    } ;
    \node [below of= m-7-3, xshift=0.4cm, node distance = 0.5cm] {$\vect{z}=(1,1,0,0,1,1,0)$}; 
    \begin{scope}[rounded corners,fill=blue!80!white,fill opacity=0.1,draw=blue,draw opacity=0.6,thick]
    \filldraw (m-1-1.north west)  rectangle (m-1-2.south east);
    \filldraw (m-2-1.north west)  rectangle (m-2-3.south east);
    \filldraw (m-5-2.north west)  rectangle (m-5-6.south east);
    \filldraw (m-6-3.north west)  rectangle (m-6-7.south east);
   \end{scope}
    \begin{scope}[rounded corners,fill=orange!80!white,fill opacity=0.1,draw=orange,thick]
    \filldraw (m-1-1.north west)  rectangle (m-2-1.south east);
    \filldraw (m-2-3.north west)  rectangle (m-6-3.south east);
    \filldraw (m-1-2.north west)  rectangle (m-5-2.south east);
   \end{scope}
   \begin{scope}[fill=red!80!white,fill opacity=0.1,draw=red,thick]
    \filldraw (m-5-6)  circle (8pt);
    \filldraw (m-6-7)  circle (8pt);
   \end{scope}
    \end{tikzpicture}
    \caption{Bitstrings that parameterize $2$-local $\hat{Z}$ operators according to the rules in \cref{lem:$2$-local-Z}. The top row represents special cases, while the bottom row represents general cases. The left column corresponds to Rule (1), and the right column corresponds to Rule (2). The blue horizontal boxes are $\bar{Z}$ operators parameterized by the bitstrings $\vect{z}$. The orange vertical boxes are $Z$-type gauge operators multiplied to $\bar{Z}$ to obtain a $2$-local $\hat{Z}$. The red circles are the remaining $Z$ physical operators of $\hat{Z}$.}
    \label{fig:rules}
\end{figure}
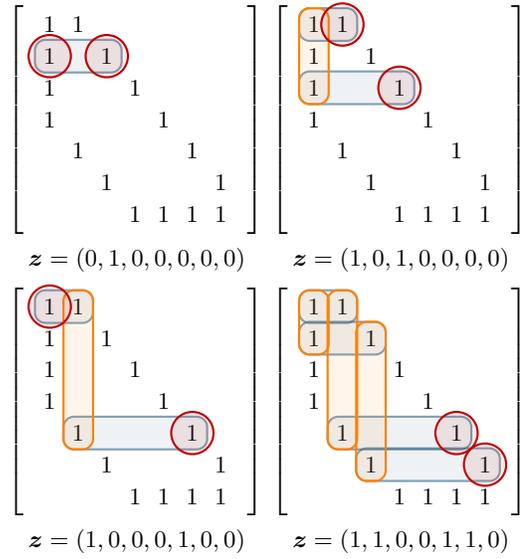

Having established the rules to construct $2$-local $\hat{Z}$ operators in \cref{lem:$2$-local-Z}, we move on to show how to apply this for $\hat{X}$ operators in \cref{lem:$2$-local-X}.

\begin{mylemma}\label{lem:$2$-local-X}
    For a $[[4k+2l,2k,g,2]]$ trapezoid code, a bitstring $\vect{x}$ parameterizes a $2$-local logical operator $\hat{X}$ if and only if it is inverted of a bitstring in \cref{lem:$2$-local-Z}, where the $i$'th component of a length-$m$ inverted bitstring $\vect{a}^{-1}$ is defined as $a_i^{-1}=a_{m-i}$, such that $a_i$ is the $i$'th component of $\vect{a}$.
\end{mylemma} 
\begin{proof}
    Note that the $A$ matrix is symmetric upon reflection about the main diagonal that runs from bottom left to top right. Perform the reflection about this diagonal and follow the proof in \cref{lem:$2$-local-Z} on the rows and columns, which are now the inverted columns and rows, respectively, after reflection. It follows that $\hat{X}_i$ is parameterized by $\vect{x}=\vect{z}^{-1}$, when $\vect{z}$ parameterizes $\hat{Z}_i$ of a logical qubit $i$.
\end{proof}

We have shown that \cref{lem:$2$-local-Z} and \cref{lem:$2$-local-X} have a complete set of bitstrings $\vect{z}$ and $\vect{x}$ that parameterize $2$-local logical operators. The operator $\hat{L}_i\hat{L}_j$ where $L\in \{X,Z\}$ is parameterized by a bitstring $\vect{a}_{i,j}=\vect{a}_i\oplus \vect{a}_j$ where $\vect{a}_i$ and $\vect{a}_j$ parameterize $\hat{L}_i$ and $\hat{L}_j$, respectively. This implies that the bitstring that parameterizes $\hat{Z}\hat{Z}$ and $\hat{X}\hat{X}$ must also follow the rules in \cref{lem:$2$-local-Z} and \cref{lem:$2$-local-X}, respectively.

\subsection{Example}
\label{app:exp-reduction}
This section presents an example of how to construct $2$-local $\hat{Z}$ and $\hat{X}$ of the $[[16,6,8,2]]$ trapezoid code ($m=7,k=3,l=2$). The $A$ matrix of the code is given in \cref{eq:ex-16-6-2}. The row and column are indexed from 1 to $m$. 
\begin{equation}
    \left[\begin{array}{ccccccc}
    \mbf{1}&\mbf{1}&\txg{0}&\txg{0}&\txg{0}&\txg{0}&\txg{0}\\
    \mbf{1}&\txg{0}&\mbf{1}&\txg{0}&\txg{0}&\txg{0}&\txg{0}\\
    \mbf{1}&\txg{0}&\txg{0}&\mbf{1}&\txg{0}&\txg{0}&\txg{0}\\
    \mbf{1}&\txg{0}&\txg{0}&\txg{0}&\mbf{1}&\txg{0}&\txg{0}\\
    \txg{0}&\mbf{1}&\txg{0}&\txg{0}&\txg{0}&\mbf{1}&\txg{0}\\
    \txg{0}&\txg{0}&\mbf{1}&\txg{0}&\txg{0}&\txg{0}&\mbf{1}\\
    \txg{0}&\txg{0}&\txg{0}&\mbf{1}&\mbf{1}&\mbf{1}&\mbf{1}\\
    \end{array}    \right] 
    \label{eq:ex-16-6-2}
\end{equation}

Using \cref{lem:$2$-local-Z} and \cref{lem:$2$-local-X}, a complete set of bitstrings $\vect{z}$ and $\vect{x}$ that parameterizes $2$-local $\hat{Z}$ and $\hat{X}$ operators of the $[[16,6,8,2]]$ code can be listed as in \cref{tab:exp-logical}.

\begin{table}[h]
    \centering
    \begin{tabular}{ccccccc}
    \hline
    \hline
    \multicolumn{7}{c}{$Z$-type}\\  
    \hline
         (1) Special && (2) Special && (1) General && (2) General \\
    \hline
        \textcolor{orange}{\textbf{1} 0 0 0 0 0 0} & \ & $\mbf{1}$ $\mbf{1}$ \txg{0} \txg{0} \txg{0} \txg{0} \txg{0} & \ & 
        \textcolor{orange}{\textbf{1} 0 0 0 \textbf{1} 0 0} & \ & \textbf{1} \textbf{1} \txg{0} \txg{0} \textbf{1} \txg{0} \txg{0} \\
        \textcolor{orange}{0 \textbf{1} 0 0 0 0 0} & \ & $\mbf{1}$ \txg{0} $\mbf{1}$ \txg{0} \txg{0} \txg{0} \txg{0} & \ & 
        \textcolor{orange}{0 \textbf{1} 0 0 0 \textbf{1} 0} & \ & \textbf{1} \textbf{1} \txg{0} \txg{0} \txg{0} \textbf{1} \txg{0} \\
        \textcolor{orange}{0 0 \textbf{1} 0 0 0 0} & \ & \textbf{1} \txg{0} \txg{0} \textbf{1} \txg{0} \txg{0} \txg{0} & \ & 
         & \ & \textbf{1} \textbf{1} \txg{0} \txg{0} \textbf{1} \textbf{1} \txg{0} \\
        \textcolor{orange}{0 0 0 \textbf{1} 0 0 0} & \ & \txg{0} \textbf{1} \textbf{1} \txg{0} \txg{0} \txg{0} \txg{0} & \ & 
         & \ & \textbf{1} \txg{0} \textbf{1} \txg{0} \textbf{1} \txg{0} \txg{0} \\
        \txg{0} \txg{0} \txg{0} \txg{0} \textbf{1} \txg{0} \txg{0} & \ & \txg{0} \textbf{1} \txg{0} \textbf{1} \txg{0} \txg{0} \txg{0} & \ & 
         & \ & \textbf{1} \txg{0} \txg{0} \textbf{1} \textbf{1} \txg{0} \txg{0} \\
        \txg{0} \txg{0} \txg{0} \txg{0} \txg{0} \textbf{1} \txg{0} & \ & \txg{0} \txg{0} \textbf{1} \textbf{1} \txg{0} \txg{0} \txg{0} & \ & 
         & \ & \txg{0} \textbf{1} \textbf{1} \txg{0} \txg{0} \textbf{1} \txg{0} \\
         & \ &  \ & &  \ & & \txg{0} \textbf{1} \txg{0} \textbf{1} \txg{0} \textbf{1} \txg{0} \\
    \hline
    \hline
    \multicolumn{7}{c}{$X$-type}\\  
    \hline
         (1) Special && (2) Special && (1) General && (2) General \\
    \hline
        \textcolor{orange}{0 0 0 0 0 0 \textbf{1}} & \ & \txg{0} \txg{0} \txg{0} \txg{0} \txg{0} \textbf{1} \textbf{1} & \ & \textcolor{orange}{0 0 \textbf{1} 0 0 0 \textbf{1}} & \ & \txg{0} \txg{0} \textbf{1} \txg{0} \txg{0} \textbf{1} \textbf{1} \\ 
        \textcolor{orange}{0 0 0 0 0 \textbf{1} 0} & \ & \txg{0} \txg{0} \txg{0} \txg{0} \textbf{1} \txg{0} \textbf{1} & \ & \textcolor{orange}{0 \textbf{1} 0 0 0 \textbf{1} 0} & \ & \txg{0} \textbf{1} \txg{0} \txg{0} \txg{0} \textbf{1} \textbf{1} \\ 
        \textcolor{orange}{0 0 0 0 \textbf{1} 0 0} & \ & \txg{0} \txg{0} \txg{0} \textbf{1} \txg{0} \txg{0} \textbf{1} & \ & & \ & \txg{0} \textbf{1} \textbf{1} \txg{0} \txg{0} \textbf{1} \textbf{1} \\ 
        \textcolor{orange}{0 0 0 \textbf{1} 0 0 0} & \ & \txg{0} \txg{0} \txg{0} \txg{0} \textbf{1} \textbf{1} \txg{0} & \ & & \ & \txg{0} \txg{0} \textbf{1} \txg{0} \textbf{1} \txg{0} \textbf{1}\\ 
         \txg{0} \txg{0} \textbf{1} \txg{0} \txg{0} \txg{0} \txg{0}& \ & \txg{0} \txg{0} \txg{0} \textbf{1} \txg{0} \textbf{1} \txg{0} & \ &  & \ & \txg{0} \txg{0} \textbf{1} \textbf{1} \txg{0} \txg{0} \textbf{1} \\ 
        \txg{0} \textbf{1} \txg{0} \txg{0} \txg{0} \txg{0} \txg{0} & \ & \txg{0} \txg{0} \txg{0} \textbf{1} \textbf{1} \txg{0} \txg{0} & \ & & \ & \txg{0} \textbf{1} \txg{0} \txg{0} \textbf{1} \textbf{1} \txg{0} \\
        & \ &  \ & &  \ & & \txg{0} \textbf{1} \txg{0} \textbf{1} \txg{0} \textbf{1} \txg{0} \\
    \hline
    \hline
    \end{tabular}
    \caption{A complete set of bitstrings $\vect{z}$ and $\vect{x}$ that parameterizes $2$-local $\hat{Z}$ and $\hat{X}$ operators of the $[[16,6,8,2]]$ code. Orange-colored bitstrings are those chosen as logical operators in \cref{eq:log-op-b}. The column headers refer to the terminology used in \cref{lem:$2$-local-Z}.}
    \label{tab:exp-logical}
\end{table}

Bitstrings listed in \cref{tab:exp-logical} are possible candidates for constructing six pairs of logical operators $(\hat{X},\hat{Z})$ of the $[[16,6,8,2]]$ trapezoid code, however, we need to make sure that the commutation relation $[\hat{X}_i,\hat{Z}_j]=\delta_{ij}$ is satisfied.

For dressed logical operators, we follow \cref{prop:dressed-commutation} and pick a set of bitstrings that satisfies $\vect{z}_j^TA\vect{x}_i+\bar{\vect{z}}_j^TA^T\bar{\vect{x}}_i=\delta_{ij}$. Our choice for the $[[16,6,8,2]]$ code is given in \cref{eq:log-op-b} and \cref{eq:log-op-b-bar}, i.e., $\vect{Z}^TA\vect{X}+\bar{\vect{Z}}^TA^T\bar{\vect{X}}=I$, where $I$ is a $(m-1)\times (m-1)$ identity matrix.

\subsection{Counting the cardinality of $\mathcal{S}_2$}
We have demonstrated that $\mathcal{S}_2$ exclusively contains bitstrings that parameterize $2$-local operators. In this section, we  determine the cardinality of this set and establish the relation that it depends on $m$ but not on $l$. Specifically, trapezoid codes with $A$ matrices of identical dimensions will have $\mathcal{S}_2$ sets of equal size.

\begin{mylemma}\label{lem:m-not-l}
   Let $\mathcal{S}_2$ be the set of bitstrings that parameterize $2$-local $\hat{Z}$ operators. For odd-$m$ codes, the total number of such bitstrings is $|\mathcal{S}_2|=m(m-1)/2$, which is a function of $m$ (or $k$) but not $l$. 
\end{mylemma}

\begin{proof}
    We  prove by induction that a code corresponding to an $m\times m$ $A$ matrix ($m=2k+1$) has $m(m-1)/2$ bitstrings that parameterize $2$-local $\hat{Z}$ operators. We start by showing the base case $(k=1,m=3;l=1)$, then assume that it is true for $k>1$, and lastly, show that it is also true for $k'=k+1$, which corresponds to an $m'\times m'$ $A'$ matrix, where $m'=m+2$. Since bitstrings for $\hat{X}$ operators are the inverted of the $\hat{Z}$ operators, this proof also applies for $\vect{x}$ without loss of generality.
    
    The $A$ matrix for the base case is given by 
    \begin{equation}
        \left[\begin{array}{ccc}
            1 & 1 & 0 \\
            1 & 0 & 1 \\
            0 & 1 & 1
        \end{array}\right].
    \end{equation}
    Using \cref{lem:$2$-local-Z}, we list bitstrings $\vect{z}$ that parametrize $2$-local $\hat{Z}$ operators in \cref{tab:base-mm-1/2}.
    The total number of bitstrings is $3$, which agrees with $m(m-1)/2=3$ when $m=3$.
    \begin{table}[h]
        \centering
        \begin{tabular}{ccccccc}
        \hline
        \hline
            (1) Special && (2) Special && (1) General && (2) General \\
        \hline
            \textbf{1} \txg{0} \txg{0} && \textbf{1} \textbf{1} \txg{0} &&  && \\
            \txg{0} \textbf{1} \txg{0} &&  &&  && \\
        \hline
        \hline
        \end{tabular}
        \caption{A complete set of bitstrings $\vect{z}$ that parameterize $2$-local $\hat{Z}$ operators for the $[[6,2,2,2]]$ code.}
    \label{tab:base-mm-1/2}
    \end{table}

    Next, we assume that the total number of bitstrings is $m(m-1)/2$, where $m=2k+1$, for the code corresponding to the $m\times m$ $A$ matrix, drawn as the inner box in \cref{eq:m7to9}. Consider modifying the code parameter $k\rightarrow k'=k+1$, which corresponds to modifying the matrix $A\rightarrow A'$, drawn as the outer box in \cref{eq:m7to9}. The $A'$ matrix is of size $m'\times m'$, where $m'=m+2$. The gray-colored $1$-entries are not included in the $A'$ matrix.
    \begin{align}
    \label{eq:m7to9}
    &\begin{array}{llccccccccccclccl}
        &\vline&1&1&&&&&&&&\vline&&&\vline\\
        &\vline&1 &&1&&&&&&&\vline&&&\vline\\
        &\vline&\vdots &&&\ddots&&&&&&\vline&&&\vline\\
        \scriptstyle{2l}&\vline&1&&&&1&&&&&\vline&&&\vline\\ 
        \scriptstyle{2l+1}&\vline&&1&&&&1&&&&\vline&&&\vline\\
        &\vline&&&\ddots&&&&\ddots&&&\vline&&&\vline\\
        &\vline&&&&1&&&&1&&\vline&&&\vline\\
        \scriptstyle{m-1}&\vline&&&&&1&&&&1&\vline&&&\vline\\
        \scriptstyle{m}&\vline&&&&&&1&\txg{\cdots}&\txg{1}&\txg{1}&\vline&1&&\vline\\
        \hline
        &\vline&&&&&&&1&&&\vline&&1&\vline\\
        \scriptstyle{m+2}&\vline&&&&&&&&1&\cdots&\vline&1&1&\vline\\
        \hline
        &\vline&&&&&&\hspace{-.1cm}\scriptstyle{m-2l+1}&&&\scriptstyle{m}&\vline&&\hspace{-.2cm}\scriptstyle{m+2}&\vline\\
    \end{array}
    \end{align}

We now consider the addition of bitstrings $\vect{z}$ resulting from this modification, i.e., the additional bitstrings using Rule (1) and Rule (2) from \cref{lem:$2$-local-Z}.

\textbf{Rule (1) contribution:} The special case of Rule (1) contributes as two added rows, $m+1$ and $m+2$. The general case counts the number of Rule (1) bitstrings that have a $1$-entry in the $(m-2l)$'th or $(m+1-2l)$'th position. Only these bitstrings can have an extra $1$-entry in the $m$'th or $(m+1)$'th position because positions of the $1$-entries need to be separated by $2l$. The total number of Rule (1) bitstrings having their last $1$-entry in the $i$'th row is $\lfloor \frac{i-1}{2l}\rfloor$. Hence, the total contribution from Rule (1) is $2+\lfloor \frac{m-1}{2l}\rfloor+\lfloor \frac{m}{2l}\rfloor$.

\textbf{Rule (2) contribution:} To extend Rule (2) bitstrings to $\vect{z}'$ of length $m+2$, the extension can only come from the original length-$m$ $\vect{z}$ bitstring that has $(1,1)$, $(0,1)$ or $(1,0)$ in the $(m-2l,m+1-2l)$'th positions. This is because the positions of each $1$-entry must be $2l$ apart for the gauge operators in the same column to cancel. Let $i_0 = m \mod 2l$ and $i_1=m+1-2l\mod 2l$.  It follows that only bitstrings with the first $1$-entry in the $i_0$ and/or $i_1$ positions can have $(1,1)$, $(0,1)$ or $(1,0)$ in the $(m-2l,m+1-2l)$'th positions. 

Consider a bitstring with a $1$-entry in the $i_0$ position and has $(1,0)$ in the $(m-2l,m+1-2l)$'th positions. The $m$'th row adds another variation of bitstring by adding a $1$-entry in the $m$'th position. 
Rule (2) bitstrings with a $1$-entry in the $i_0$ position also has another $1$-entry in the $i$ position, where $i\in[1,2l]\setminus\{i_0\}$. 
Likewise, for bitstrings with a $1$-entry in the $i_1$ position and $(0,1)$ in the $(m-2l,m+1-2l)$'th positions, the $(m+1)$'th row adds another variation of bitstring by adding a $1$-entry in the $(m+1)$'th position. 
Counting the number of these bitstrings, we have $\sum_{i=i_0,i_1}\sum_{\substack{j=1\\j\neq i}}^{2l}1+\lfloor \frac{m-1-j}{2l} \rfloor$. Lastly, we count the contribution the bitstring with $1$-entries in both $i_0$ and $i_1$ positions and $(1,1)$ in the $(m-2l,m+1-2l)$'th positions.

\textbf{Total contribution:} The total contribution from both Rule (1) and Rule (2) bitstrings is
\begin{equation}
    \begin{aligned}
        &\underbrace{2+\lfloor \frac{m-1}{2l}\rfloor+\lfloor \frac{m}{2l}\rfloor}_{\text{Rule (1)}}+\underbrace{1+\sum_{i=i_0,i_1}\sum_{\substack{j=1\\j\neq i}}^{2l}1+\lfloor \frac{m-1-j}{2l} \rfloor}_{\text{Rule (2)}}\\
        &\quad =1+\lfloor \frac{m-1}{2l}\rfloor+\lfloor \frac{m}{2l}\rfloor-\lfloor \frac{m-1-i_0}{2l} \rfloor -\lfloor \frac{m-1-i_1}{2l} \rfloor+2\left(\sum_{j=1}^{2l} 1+\lfloor \frac{m-1-j}{2l} \rfloor \right)\\
        &\quad =1+2x -2(x-1)+2(m-1)=2m+1,
    \end{aligned}
\end{equation}
where $x=\lfloor \frac{m-1}{2l}\rfloor$. 

To see this, let $\frac{m-1}{2l}=x+\frac{y}{2l}$, where $0\le y<2l$. Using this definition, we can write
\begin{equation}
    \begin{aligned}
        \lfloor \frac{m-1}{2l}\rfloor &= x + \lfloor\frac{y}{2l}\rfloor,\\
        \lfloor \frac{m}{2l}\rfloor&= x + \lfloor\frac{y+1}{2l}\rfloor,\\
        \lfloor \frac{m-1-i_0}{2l} \rfloor &= x-1 + \lfloor\frac{1}{2l}\rfloor,\\
        \lfloor \frac{m-1-i_1}{2l} \rfloor&= x-1 + \lfloor\frac{2}{2l}\rfloor.
    \end{aligned}
\end{equation}
In the second line, note that since $m$ is odd, $y\equiv m-1\mod 2l$ must be even. It follows that $0\le y\le 2l-2$, which implies that $y+1<2l$. In the third and fourth lines, we have used the definition $i_0 = m \mod 2l$, which implies $m-1-i_0 \equiv -1 \mod 2l$ and likewise for $i_1=i_0+1$.

Similarly, $\sum_{j=1}^{2l}\lfloor\frac{m-1-j}{2l}\rfloor=xy+(2l-y)(x-1)=2xl-2l+y=m-1-2l$ because $\frac{m-1-y}{2l}=x$ by definition. Hence, $2(\sum_{j=1}^{2l}1+\lfloor\frac{m-1-j}{2l}\rfloor)=2(2l+m-1-2l)=2(m-1)$.

The total sum of bitstrings is $2m+1$. Hence, we have $m(m-1)/2 + (2m+1) = (m+2)(m+1)/2=m'(m'-1)/2$.

\end{proof}

\section{Penalty Hamiltonian}
\label{app:penalty}

\subsection{Construction of the dressed Hamiltonian}

We follow the proof in Ref.~\cite{marvian2019robust} and apply it to our setting. The goal is to show that in the large penalty limit $H(t)$ [\cref{eq:H_total}] generates the desired computation, i.e., the same computation as generated by $H_S(t)$ in the absence of coupling to a bath. 

\begin{mylemma} \cite[Lemma~1.1]{marvian2019robust}
\label{lem:1.1}
    Let $H_P=-\sum_{g_i\in \mathcal{G}'}g_i$, where $\mathcal{G}'$ is a set of $X$-type and $Z$-type gauge operators that  generates the gauge group $\mathcal{G}$. Denote the projector to the ground subspace of $H_P$ by $\Pi_0$. Then
\begin{enumerate}
    \item Any ground state of $H_P$ is stabilized by the stabilizers of the subsystem code.
    \item $\Pi_0 g_i \Pi_0$ is proportional to $\Pi_0$, i.e., $\Pi_0g_i\Pi_0=\alpha_i\Pi_0$.
\end{enumerate}
\end{mylemma}

\begin{mycorollary}\cite[Eq.~(6)]{marvian2019robust}
    For $H_{SB}=\sum \sigma_\alpha\otimes B_\alpha$, where $\sigma_\alpha$ is a $1$-local system operator, $\Pi_0H_{SB}\Pi_0=0$.
\end{mycorollary}
\begin{proof}
The first part of \cref{lem:1.1} guarantees that the ground subspace of the penalty Hamiltonian is in the
codespace and hence can detect any $1$-local error $\sigma_\alpha$. The codespace projector can be written as $\Pi_0=\prod_{i=1}^{n-k} \frac{1}{2}(I + S_i)$. For a stabilizer code, errors are detectable if and only if they anticommute with at least one stabilizer generator, e.g., $S_j$. Therefore, $(I + S_j)\sigma_\alpha(I + S_j) = (I + S_j)(I - S_j)\sigma_\alpha = 0$. Thus, $\Pi_0\sigma_\alpha\Pi_0=0,\forall \alpha$, and $\Pi_0H_{SB}\Pi_0=0$.
\end{proof}

\begin{mylemma}\cite[Eq.~(8)]{marvian2019robust}
    For $\bar{H}_S(t)$ and $\hat{H}_S(t)$ in \cref{eq:H-bar} and \cref{eq:H-hat-b}, respectively, $\Pi_0\hat{H}_S(t)\Pi_0=\Pi_0\bar{H}_S(t)$.
\end{mylemma}

\begin{proof}
Using $\hat{H}_S(t)$ in the form of \cref{eq:H-hat} and applying the projectors $\Pi_0$ on both sides, we can write: 
\bes
    \begin{align}
            &\Pi_0\hat{H}_S(t)\Pi_0\notag\\
            &\quad=\Pi_0[\sum_i a_i \Pi_0\bar{X}_i g_i^x\Pi_0 +\sum_i b_i \Pi_0\bar{Z}_i g_i^z\Pi_0+\sum_{(i,j)\in E^{XX}_S} c_{ij} \Pi_0\bar{X}_i\bar{X}_j g^x_{ij}\Pi_0 +\sum_{(i,j)\in E^{ZZ}_S} d_{ij} \Pi_0\bar{Z}_i\bar{Z}_j g^z_{ij}\Pi_0 ]\\
            &\quad=\Pi_0[\sum_i a_i \bar{X}_i \Pi_0g_i^x\Pi_0 +\sum_i b_i \bar{Z}_i \Pi_0g_i^z\Pi_0+\sum_{(i,j)\in E^{XX}_S} c_{ij} \bar{X}_i\bar{X}_j \Pi_0g^x_{ij}\Pi_0 +\sum_{(i,j)\in E^{ZZ}_S} d_{ij} \bar{Z}_i\bar{Z}_j \Pi_0g^z_{ij}\Pi_0 ],
    \end{align}
\ees
where to go from the first to the second equality, we used the fact that by definition, bare logical operators leave the ground subspace of the penalty Hamiltonian invariant, so that $[\Bar{X}_i,\Pi_0]=[\Bar{Z}_i,\Pi_0]=0$.

    The situation differs slightly when we use dressed logical operators. 
    The optimal set of logical operators we constructed in \cref{sec:optimal_logi} are dressed logical operators $\hat{X}_i$ and $\hat{Z}_i$, which can be represented by products of bare logical operators and gauges $g^x_i$ and $g^z_i$. Note that the gauges $g^x_i$ and $g^z_i$ might be products of $2$-local gauge generators, which also belong to the gauge group $\mathcal{G}$. One special case is when $n=l$, as shown in \cite{marvian2019robust}. They use the same logical operators, but they have $\hat{X}_i=\bar{X}_i$ and $\hat{Z}_i=\bar{Z}_i$ due to $n=l$. However, their $\hat{X}_i\hat{X}_j$ and $\hat{Z}_i\hat{Z}_j$ are dressed logical operators, which corresponds to the case where $\alpha^x_i=\alpha^z_i=1$. 
     The bare logical operator commutes with $\Pi_0$, so we arrive at the second line assuming that $g^x_i$ is as discussed above (i.e., a product of 2-local gauge
generators). The same argument applies to $\hat{Z}_i$.

Using the second part of \cref{lem:1.1}, $\Pi_0 g^{\{x,z\}}_i\Pi_0=\alpha^{\{x,z\}}_i\Pi_0$, we arrive at:
    \begin{equation}
        \begin{aligned}
            \Pi_0\hat{H}_S(t)\Pi_0
            &=\Pi_0[\sum_i\alpha_i^x a'_i \bar{X}_i +\sum_i \alpha_i^zb'_i \bar{Z}_i\\
            &+\sum_{(i,j)\in E^{XX}_S} \alpha_{ij}^x c'_{ij} \bar{X}_i\bar{X}_j g^x_{ij}+\sum_{(i,j)\in E^{ZZ}_S} \alpha_{ij}^z d'_{ij} \bar{Z}_i\bar{Z}_j g^z_{ij} ],
        \end{aligned}
    \end{equation}
    which is $\Bar{H}_S(t)$ in \cref{eq:H-bar} when the constants are set to $a'_i=a_i/\alpha^x_i, b'_i=b_i/\alpha^z_i,c'_{ij}=c_{ij}/\alpha^x_{ij}$ and $d'_{ij}=d_{ij}/\alpha^z_{ij}$.
\end{proof}

\begin{table*}[t]
    \centering
    \begin{tabular}{cccccc}
    \hline
    \hline
        $g^x$ & operators & index & $g^z$ & operators & index  \\
    \hline
        $g^x_{a,i}$&$X_{i,1}X_{i,i+1}$&$1\le i\le 2l$&$g^z_{a,i}$&$Z_{m+i-2l,m_g+i}Z_{m,m_g+i}$&$0\le i\le 2l-1$\\
        $g^x_{b,i}$&$X_{i,i+1}X_{i,i+1+2l}$&$2l+1\le i\le m-1$&$g^z_{b,i}$&$Z_{i,i+1}Z_{i+2l,i+1}$&$1\le i\le m_g-2$\\
        $g^x_{c,i}$&$X_{m,m_g}X_{m,m_g+i}$&$1\le i\le 2l-1$&$g^z_{c,i}$&$Z_{m_l,1}Z_{i,1}$&$1\le i\neq m_l\le 2l$\\
    \hline
    \hline
    \end{tabular}
    \caption{A complete set of $X$-type and $Z$-type gauge operators, denoted by $g^{(x,z)}_{(a,b,c),i}$, and their corresponding physical operators on the $A$ matrix. There is a total of $n_g=m-2+2l$ gauges for each $X$- and $Z$-type operator.}
    \label{tab:gauge-from-phys}
\end{table*}

\subsection{Simplification of the penalty gap calculation}
\label{app:gap}

Let each qubit be indexed according to its row and column. For an $A$ matrix of size $m\times m$ parameterizing the code $[[4k+2l,2k,g,2]]$, the penalty Hamiltonian $H_P=-\sum_{g_i\in \mathcal{G}}g_i$ can be written explicitly as
\begin{equation}\label{eq:ham-penalty}
    H_P=-\sum_{i=1}^{n_g}g^x_i-\sum_{i=1}^{n_g}g^z_i,
\end{equation} 
where $n_g=m-2+2l$ is the number of gauge generators of each type. To see this, consider the $A$ matrix and note that for $X$-type gauge generators, the first $m-1$ rows each contributes one $g^x_i$ because there are two $X$ operators in each of the first $m-1$ row. The last row has $2l$ operators, hence, it contributes $2l-1$ non-repeating gauge generators. The total is then $(m-1)+(2l-1)=m-2+2l$. A similar argument holds for the $Z$-type generators. The physical $Z$ operators composing the $Z$-type gauge generators are given in \cref{lem:gauge-Z}. 

In this subsection, we use the notation $g^{(x,z)}_{(a,b,c),i}$ to label the gauge generators according to their shape and position in the $A$ matrix. An example of this indexing is drawn in \cref{fig:g-XZ-abc} on the $A$ matrix of a $[[16,6,8,2]]$ code. The definition is given in \cref{tab:gauge-from-phys}, where there are $2l$, $m-1-2l$, and $2l-1$ operators of type $a$, $b$, and $c$, respectively, which adds up to $m-2+2l=n_g$ gauges as required. Let $m_g=m+1-2l$ be the column index of the leftmost $1$-entry of the last row of $A$.

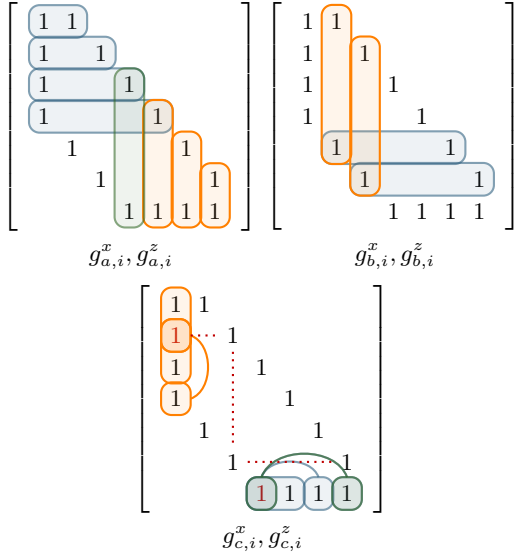
\begin{figure}[h!]
\centering
\begin{tikzpicture}
    \matrix (m)[
    matrix of math nodes,
    nodes in empty cells,
    left delimiter=\lbrack,
    right delimiter=\rbrack
    ] {
    1&1&&&&&\\
    1&&1&&&&\\
    1&&&1&&&\\
    1&&&&1&&\\
    &1&&&&1&\\
    &&1&&&&1\\
    &&&1&1&1&1\\
    } ;
    \node [below of= m-7-3, xshift=0.4cm, node distance = 0.5cm] {$g^x_{a,i},g^z_{a,i}$}; 
    \begin{scope}[rounded corners,fill=blue!80!white,fill opacity=0.1,draw=blue,draw opacity=0.6,thick]
    \filldraw (m-1-1.north west)  rectangle (m-1-2.south east);
    \filldraw (m-2-1.north west)  rectangle (m-2-3.south east);
    \filldraw (m-3-1.north west)  rectangle (m-3-4.south east);
    \filldraw (m-4-1.north west)  rectangle (m-4-5.south east);
   \end{scope}
   \begin{scope}[rounded corners,fill=orange!80!white,fill opacity=0.1,draw=orange,thick]
    \filldraw (m-4-5.north west)  rectangle (m-7-5.south east);
    \filldraw (m-5-6.north west)  rectangle (m-7-6.south east);
    \filldraw (m-6-7.north west)  rectangle (m-7-7.south east);
   \end{scope}
   \begin{scope}[rounded corners,fill=green!80!white,fill opacity=0.1,draw=green,draw opacity=0.6,thick]
    \filldraw (m-3-4.north west)  rectangle (m-7-4.south east);
   \end{scope}
    \end{tikzpicture}
    \begin{tikzpicture}\label{fig:rule-2}
    \matrix (m)[
    matrix of math nodes,
    nodes in empty cells,
    left delimiter=\lbrack,
    right delimiter=\rbrack
    ] {
    1&1&&&&&\\
    1&&1&&&&\\
    1&&&1&&&\\
    1&&&&1&&\\
    &1&&&&1&\\
    &&1&&&&1\\
    &&&1&1&1&1\\
    } ;
    \node [below of= m-7-3, xshift=0.4cm, node distance = 0.5cm] {$g^x_{b,i},g^z_{b,i}$}; 
    \begin{scope}[rounded corners,fill=blue!80!white,fill opacity=0.1,draw=blue,draw opacity=0.6,thick]
    \filldraw (m-5-2.north west)  rectangle (m-5-6.south east);
    \filldraw (m-6-3.north west)  rectangle (m-6-7.south east);
   \end{scope}
    \begin{scope}[rounded corners,fill=orange!80!white,fill opacity=0.1,draw=orange,thick]
    \filldraw (m-1-2.north west)  rectangle (m-5-2.south east);
    \filldraw (m-2-3.north west)  rectangle (m-6-3.south east);
   \end{scope}
    \end{tikzpicture}
    \begin{tikzpicture}\label{fig:rule-3}
    \matrix (m)[
    matrix of math nodes,
    nodes in empty cells,
    left delimiter=\lbrack,
    right delimiter=\rbrack
    ] {
    1&1&&&&&\\
    \textcolor{red}{1}&&1&&&&\\
    1&&&1&&&\\
    1&&&&1&&\\
    &1&&&&1&\\
    &&1&&&&1\\
    &&&\textcolor{red}{1}&1&1&1\\
    } ;
    \node [below of= m-7-3, xshift=0.4cm, node distance = 0.5cm] {$g^x_{c,i},g^z_{c,i}$}; 
    \begin{scope}[rounded corners,fill=blue!80!white,fill opacity=0.1,draw=blue,draw opacity=0.6,thick]
    \filldraw (m-7-4.north west)  rectangle (m-7-5.south east);
    \filldraw (m-7-4.north west)  rectangle (m-7-4.south east);
    \filldraw (m-7-6.north west)  rectangle (m-7-6.south east);
    \filldraw (m-7-7.north west)  rectangle (m-7-7.south east);
    \draw (m-7-4.north) to[out=70,in=110] (m-7-6.north);
    \draw (m-7-4.north) to[out=70,in=110] (m-7-7.north);
   \end{scope}
    \begin{scope}[rounded corners,fill=orange!80!white,fill opacity=0.1,draw=orange,thick]
    \filldraw (m-1-1.north west)  rectangle (m-2-1.south east);
    \filldraw (m-2-1.north west)  rectangle (m-2-1.south east);
    \filldraw (m-2-1.north west)  rectangle (m-3-1.south east);
    \filldraw (m-4-1.north west)  rectangle (m-4-1.south east);
    \draw (m-2-1.east) to[out=-20,in=20] (m-4-1.east);
   \end{scope}
   \begin{scope}[rounded corners,fill=green!90!white,fill opacity=0.1,draw=green,draw opacity=0.6,thick]
    \filldraw (m-7-4.north west)  rectangle (m-7-4.south east);
    \filldraw (m-7-7.north west)  rectangle (m-7-7.south east);
    \draw (m-7-4.north) to[out=70,in=110] (m-7-7.north);
   \end{scope}
   \begin{scope}[draw=red,thick]
    \draw[dotted] (m-2-1)  -- (m-2-3) -- (m-6-3) -- (m-6-7) ;
   \end{scope}
    \end{tikzpicture}
    \caption{Gauge operators for evaluating the penalty Hamiltonian. Blue (orange) boxes are $X$-type ($Z$-type) operators. The green boxes are the two operators chosen to be written in terms of the product of other gauges in \cref{eq:chosen-g}.}
   \label{fig:g-XZ-abc}
\end{figure}

\begin{table*}[]
    \centering
    \begin{tabular}{ccc}
    \hline
    \hline
        row index for $\hat{X}_i,\hat{Z}_i$ & $X$-type gauges & $Z$-type gauges  \\
    \hline
        $1\le i\neq m_l\le 2l$&$g^x_{a,i}$&$g^z_{a,2l-1}g^z_{c,i}g^z_b(n_l)$\\
        $i= m_l$&$g^x_{a,i}$&$g^z_{a,2l-1}g^z_b(n_l)$\\
        $2l+1\le i\le m-2$&$g^x_{b,i}$&$g^z_{a,2l-1}g^z_{c,\fM(i)}g^z_b(n_l)\prod_{n=1}^{\fN_{\uparrow}(i)} g^z_{b,i-2nl}$\\
        $i=m-1$&$g^x_{b,i}$&$g^z_{a,2l-1}$\\
        $m\le i\le n_g-1$&$g^x_{c,i-m+1}$&$g^z_{a,i-m+1}g^z_{a,2l-1}g^z_{c,\fM(i+1)}g^z_b(n_l)\prod_{n=2}^{\fN_{\uparrow}(i+1)}g^z_{b,i+1-2nl}$\\
    \hline
    \hline
    \end{tabular}
    \caption{The set of reduced operators in terms of gauge operators whose physical operators are in \cref{tab:gauge-from-phys}. Here $g^z_b(n_l)=\prod_{n=0}^{n_l} g^z_{b,m_l+2nl}$.}
    \label{tab:gauge-to-log}
\end{table*}

The most straightforward way to calculate the penalty gap of $H_P$ in \cref{eq:ham-penalty} is to evaluate $H_P$ in the entire Hilbert space of $n=4k+2l$ qubits, which scales as $\mathcal{O}(2^n)$ and quickly becomes intractable. However, it turns out that we can alleviate this problem by constructing a new, reduced operator space from the gauge operators, which removes one of the gauge operators. As a result, evaluating $H_P$ in the correspondingly reduced Hilbert space quadratically reduces the complexity to $\mathcal{O}(2^{n/2})$.

\cref{tab:gauge-to-log} defines the reduced operators in terms of gauge operators $g^{(x,z)}_{(a,b,c),i}$. There are $n_g-1$ reduced operators for each type ($X$ and $Z$). To simplify the notation, we define the following functions. The row and column number is counted from $1$ to $m$.
\bes
\begin{align}
   \fM(x)&=x\mod (2l),\\
   \fN_{\downarrow}(x)&=\lfloor\frac{m-1-x}{2l}\rfloor-1,\\
    \fN_{\uparrow}(x)&=\lfloor\frac{x}{2l}\rfloor .
\end{align}
\ees
We also define two variables that appear often: $m_l=\fM(m-1)$ and $n_l=\fN_{\downarrow}(m_l)$. Without loss of generality, we choose to define the $X$-type operators in a simpler format than the $Z$-type operators. 

Using the reduced operators defined in \cref{tab:gauge-to-log}, we can invert the relation and rewrite the gauge operators in terms of the reduced operators instead. \cref{tab:log-to-gauge} shows the gauge operators written in terms of the new reduced operators. There are $n_g$ gauge operators but $n_g-1$ reduced operators, hence, only $n_g-1$ gauge operators are defined in terms of the reduced operators. The two remaining gauge operators can be written as the product of the stabilizer and other logical operators, so that the dimension of the Hilbert space can be reduced by half, i.e.,
\begin{equation}\label{eq:chosen-g}
    g^x_{n_g}=S^X \prod_{i=1}^{n_g-1}\hat{X}_i\quad\mathrm{and}\quad g^z_{n_g}=S^Z \hat{Z}_{m-2l}\prod_{i=n_g-1}^{m}\hat{Z}_i,
\end{equation}
where these (arbitrary) chosen gauge operators are colored in green in \cref{fig:g-XZ-abc}.

\begin{table*}[t]
    \centering
    \begin{tabular}{cccccc}
    \hline
    \hline
        $g^x$ & operators & index & $g^z$ & operators & index  \\
    \hline
        $g^x_{a,i}$&$\hat{X}_{i}$&$1\le i\le 2l$&$g^z_{a,i}$&$\hat{Z}_{m-1+i}\hat{Z}_{m+i-2l}$&$1\le i\le 2l-2$\\
        &&&&$\hat{Z}_{m-1}$&$i=2l-1$\\
        $g^x_{b,i}$&$\hat{X}_{i}$&$2l+1\le i\le m-1$&$g^z_{b,i}$&$\hat{Z}_i\hat{Z}_{i+2l}$&$1\le i\le m_g-2$\\
        $g^x_{c,i}$&$\hat{X}_{i+m-1}$&$1\le i\le 2l-1$&$g^z_{c,i}$&$\hat{Z}_{m_l}\hat{Z}_i$&$1\le i\neq m_l\le 2l$\\
    \hline
    \hline
    \end{tabular}
    \caption{The gauge operators in terms of the new set of reduced operators.}
    \label{tab:log-to-gauge}
\end{table*}

Using the gauge operators in this form, we can rewrite the penalty Hamiltonian in \cref{eq:ham-penalty} as
\begin{equation}
    \begin{aligned}
        H_P
        &=-\sum_{i=1}^{n_g-1}\hat{X}_{i} -S^X \prod_{i=1}^{n_g-1}\hat{X}_i-S^Z \hat{Z}_{m-2l}\prod_{i=n_g-1}^{m}\hat{Z}_i\\
        &\qquad\qquad-\hat{Z}_{m-1}-\sum_{\substack{i=1\\i\neq m_l}}^{2l}\hat{Z}_{i}\hat{Z}_{m_l} - \sum_{i=1}^{m_g-2}\hat{Z}_{i}\hat{Z}_{i+2l}- \sum_{i=1}^{2l-2}\hat{Z}_{m-1+i}\hat{Z}_{m+i-2l},
    \end{aligned}
\end{equation}
which has a quadratically smaller Hilbert space dimension compared to \cref{eq:ham-penalty} and is the form we used to calculate the penalty gap in \cref{fig:gap}.

\section{Scaling of penalty Hamiltonian gap}
\label{app:scaling}

To determine the scaling of the data presented in \cref{fig:gap}, we fit the results to a power-law model 
\begin{equation}
    y_\text{power}=a_\text{power} x^{-\nu_\text{power}},
\end{equation}
and to an exponential model
\begin{equation}
    y_\text{expo}=a_\text{expo} e^{-\nu_\text{expo} x} .
\end{equation}
Here $y_\text{model}$ is the numerically computed penalty gap and $x_\text{model}$ is the code parameter $m$ ($l$) on the left (right) of \cref{fig:gap}. The fitting parameters are $a_\text{model}$ and $\nu_\text{model}$, where $\nu_\text{model}$ is the scaling parameter that we are most interested in.

\subsection{Penalty gap as a function of $m$}

\begin{figure*}[t]
\centering
\begin{tikzpicture}[scale=0.7]
        \begin{scope}
            \node (1) at (0,0) {1};
            \node (2) at (1,0) {2};
            \node (5) at (2,0) {5};
            \node (6) at (3,0) {6};
            \node (9) at (4,0) {9};
            \node (10) at (5,0) {10};
            \node (13) at (6,0) {13};
            
            \node (3) at (0,-1) {3};
            \node (4) at (1,-1) {4};
            \node (7) at (2,-1) {7};
            \node (8) at (3,-1) {8};
            \node (11) at (4,-1) {11};
            \node (12) at (5,-1) {12};
            \node (14) at (6,-1) {14};
        \end{scope}

        \begin{scope}[draw=orange, very thick]
            \draw[-] (1) -- (2);
            \draw[-] (3) -- (4);
            \draw[dotted] (10) -- (13);
            \draw[dotted] (12) -- (14);
            \draw[-, opacity=0.3] (10) -- (13);
            \draw[-, opacity=0.3] (12) -- (14);
        \end{scope}

        \begin{scope}[draw=blue, very thick]
            \draw[-] (5) -- (6);
            \draw[-] (7) -- (8);
            \draw[-] (9) -- (10);
            \draw[-] (11) -- (12);
            \draw[dotted] (2) -- (5);
            \draw[dotted] (4) -- (7);
            \draw[dotted] (6) -- (9);
            \draw[dotted] (8) -- (11);
            \draw[-, opacity=0.3] (2) -- (5);
            \draw[-, opacity=0.3] (4) -- (7);
            \draw[-, opacity=0.3] (6) -- (9);
            \draw[-, opacity=0.3] (8) -- (11);
        \end{scope}

        \begin{scope}[draw=red, very thick]
            \draw[dotted] (1) -- (3);
            \draw[-, opacity=0.3] (1) -- (3);
            \draw[-] (13) -- (14);
        \end{scope}
        
        \node (l) at (3,-2) {$l=1$} ;
    \end{tikzpicture}
    \qquad\qquad
    \begin{tikzpicture}[scale=0.7]
        \begin{scope}
            \node (1) at (0,1) {1};
            \node (2) at (1,1) {2};
            \node (9) at (2,1) {9};
            \node (10) at (3,1) {10};
            \node (15) at (4,1) {15};
            
            \node (3) at (0,0) {3};
            \node (4) at (1,0) {4};
            \node (11) at (2,0) {11};
            \node (12) at (3,0) {12};
            \node (16) at (4,0) {16};
            
            \node (5) at (0,-1) {5};
            \node (6) at (2,-1) {6};
            \node (13) at (4,-1) {13};
            
            \node (7) at (0,-2) {7};
            \node (8) at (2,-2) {8};
            \node (14) at (4,-2) {14};
        \end{scope}

        \begin{scope}[draw=orange, very thick]
            \draw[-] (1) -- (2);
            \draw[-] (3) -- (4);
            \draw[-] (5) -- (6);
            \draw[-] (7) -- (8);
            \draw[dotted] (10) -- (15);
            \draw[dotted] (12) -- (16);
            \draw[dotted] (6) -- (13);
            \draw[dotted] (8) -- (14);
            \draw[-, opacity=0.3] (10) -- (15);
            \draw[-, opacity=0.3] (12) -- (16);
            \draw[-, opacity=0.3] (6) -- (13);
            \draw[-, opacity=0.3] (8) -- (14);
        \end{scope}

        \begin{scope}[draw=blue, very thick]
            \draw[-] (9) -- (10);
            \draw[-] (11) -- (12);
            \draw[dotted] (2) -- (9);
            \draw[dotted] (4) -- (11);
            \draw[-, opacity=0.3] (2) -- (9);
            \draw[-, opacity=0.3] (4) -- (11);
        \end{scope}

        \begin{scope}[draw=red, very thick]
            \draw[dotted] (1) -- (3) -- (5) -- (7);
            \draw[-, opacity=0.3] (1) -- (3) -- (5) -- (7);
            \draw[-] (15) -- (16) -- (13) -- (14);
        \end{scope}
        
        \node (l) at (2,-3) {$l=2$} ;
    \end{tikzpicture}
    \qquad\qquad
    \begin{tikzpicture}[scale=0.7]
        \begin{scope}
            \node (1) at (0,2) {1};
            \node (2) at (1,2) {2};
            \node (13) at (2,2) {13};
            
            \node (3) at (0,1) {3};
            \node (4) at (1,1) {4};
            \node (14) at (2,1) {14};
            
            \node (5) at (0,0) {5};
            \node (6) at (1,0) {6};
            \node (15) at (2,0) {15};
            
            \node (7) at (0,-1) {7};
            \node (8) at (1,-1) {8};
            \node (16) at (2,-1) {16};

            \node (9) at (0,-2) {9};
            \node (10) at (1,-2) {10};
            \node (17) at (2,-2) {17};
            
            \node (11) at (0,-3) {11};
            \node (12) at (1,-3) {12};
            \node (18) at (2,-3) {18};
        \end{scope}

        \begin{scope}[draw=orange, very thick]
            \draw[-] (1) -- (2);
            \draw[-] (3) -- (4);
            \draw[-] (5) -- (6);
            \draw[-] (7) -- (8);
            \draw[-] (9) -- (10);
            \draw[-] (11) -- (12);
            \draw[dotted] (2) -- (13);
            \draw[dotted] (4) -- (14);
            \draw[dotted] (6) -- (15);
            \draw[dotted] (8) -- (16);
            \draw[dotted] (10) -- (17);
            \draw[dotted] (12) -- (18);
            \draw[-, opacity=0.3] (2) -- (13);
            \draw[-, opacity=0.3] (4) -- (14);
            \draw[-, opacity=0.3] (6) -- (15);
            \draw[-, opacity=0.3] (8) -- (16);
            \draw[-, opacity=0.3] (10) -- (17);
            \draw[-, opacity=0.3] (12) -- (18);
        \end{scope}

        \begin{scope}[draw=red, very thick]
            \draw[dotted] (1) -- (3) -- (5) -- (7) -- (9) -- (11);
            \draw[-, opacity=0.3] (1) -- (3) -- (5) -- (7) -- (9) -- (11);
            \draw[-] (13) -- (14) -- (15) -- (16) -- (17) -- (18);
        \end{scope}
        
        \node (l) at (1,-4) {$l=3$} ;
    \end{tikzpicture}
    \qquad
    \begin{tikzpicture}[scale=0.7]
        \begin{scope}
            \node (1) at (0,0) {1};
            \node (2) at (1,0) {2};
            \node (5) at (2,0) {5};
            \node (6) at (3,0) {6};
            \node (9) at (4,0) {9};
            \node (10) at (5,0) {10};
            \node (13) at (6,0) {13};
            \node (14) at (7,0) {14};
            \node (17) at (8,0) {17};
            
            \node (3) at (0,-1) {3};
            \node (4) at (1,-1) {4};
            \node (7) at (2,-1) {7};
            \node (8) at (3,-1) {8};
            \node (11) at (4,-1) {11};
            \node (12) at (5,-1) {12};
            \node (15) at (6,-1) {15};
            \node (16) at (7,-1) {16};
            \node (18) at (8,-1) {18};
        \end{scope}

        \begin{scope}[draw=orange, very thick]
            \draw[-] (1) -- (2);
            \draw[-] (3) -- (4);
            \draw[dotted] (14) -- (17);
            \draw[dotted] (16) -- (18);
            \draw[-, opacity=0.3] (14) -- (17);
            \draw[-, opacity=0.3] (16) -- (18);
        \end{scope}

        \begin{scope}[draw=blue, very thick]
            \draw[-] (5) -- (6);
            \draw[-] (7) -- (8);
            \draw[-] (9) -- (10);
            \draw[-] (11) -- (12);
            \draw[-] (13) -- (14);
            \draw[-] (15) -- (16);
            \draw[dotted] (2) -- (5);
            \draw[dotted] (4) -- (7);
            \draw[dotted] (6) -- (9);
            \draw[dotted] (8) -- (11);
            \draw[dotted] (10) -- (13);
            \draw[dotted] (12) -- (15);
            \draw[-, opacity=0.3] (2) -- (5);
            \draw[-, opacity=0.3] (4) -- (7);
            \draw[-, opacity=0.3] (6) -- (9);
            \draw[-, opacity=0.3] (8) -- (11);
            \draw[-, opacity=0.3] (10) -- (13);
            \draw[-, opacity=0.3] (12) -- (15);
        \end{scope}

        \begin{scope}[draw=red, very thick]
            \draw[dotted] (1) -- (3);
            \draw[-, opacity=0.3] (1) -- (3);
            \draw[-] (17) -- (18);
        \end{scope}
        
        \node (l) at (4,-2) {$l=1$} ;
    \end{tikzpicture}
    \quad
    \begin{tikzpicture}[scale=0.7]
        \begin{scope}
            \node (1) at (0,1) {1};
            \node (2) at (1,1) {2};
            \node (9) at (2,1) {9};
            \node (10) at (3,1) {10};
            \node (17) at (4,1) {15};
            
            \node (3) at (0,0) {3};
            \node (4) at (1,0) {4};
            \node (11) at (2,0) {11};
            \node (12) at (3,0) {12};
            \node (18) at (4,0) {18};
            
            \node (5) at (0,-1) {5};
            \node (6) at (1,-1) {6};
            \node (13) at (2,-1) {13};
            \node (14) at (3,-1) {14};
            \node (19) at (4,-1) {19};
            
            \node (7) at (0,-2) {7};
            \node (8) at (1,-2) {8};
            \node (15) at (2,-2) {15};
            \node (16) at (3,-2) {16};
            \node (20) at (4,-2) {20};
        \end{scope}

        \begin{scope}[draw=orange, very thick]
            \draw[-] (1) -- (2);
            \draw[-] (3) -- (4);
            \draw[-] (5) -- (6);
            \draw[-] (7) -- (8);
            \draw[dotted] (10) -- (17);
            \draw[dotted] (12) -- (18);
            \draw[dotted] (14) -- (19);
            \draw[dotted] (16) -- (20);
            \draw[-, opacity=0.3] (10) -- (17);
            \draw[-, opacity=0.3] (12) -- (18);
            \draw[-, opacity=0.3] (14) -- (19);
            \draw[-, opacity=0.3] (16) -- (20);
        \end{scope}

        \begin{scope}[draw=blue, very thick]
            \draw[-] (9) -- (10);
            \draw[-] (11) -- (12);
            \draw[-] (13) -- (14);
            \draw[-] (15) -- (16);
            \draw[dotted] (2) -- (9);
            \draw[dotted] (4) -- (11);
            \draw[dotted] (6) -- (13);
            \draw[dotted] (8) -- (15);
            \draw[-, opacity=0.3] (2) -- (9);
            \draw[-, opacity=0.3] (4) -- (11);
            \draw[-, opacity=0.3] (6) -- (13);
            \draw[-, opacity=0.3] (8) -- (15);
        \end{scope}

        \begin{scope}[draw=red, very thick]
            \draw[dotted] (1) -- (3) -- (5) -- (7);
            \draw[-, opacity=0.3] (1) -- (3) -- (5) -- (7);
            \draw[-] (17) -- (18) -- (19) -- (20);
        \end{scope}
        
        \node (l) at (2,-3) {$l=2$} ;
    \end{tikzpicture}
    \quad
    \begin{tikzpicture}[scale=0.7]
        \begin{scope}
            \node (1) at (0,2) {1};
            \node (2) at (1,2) {2};
            \node (13) at (2,2) {13};
            \node (14) at (3,2) {14};
            \node (21) at (4,2) {21};
            
            \node (3) at (0,1) {3};
            \node (4) at (1,1) {4};
            \node (15) at (2,1) {15};
            \node (16) at (3,1) {16};
            \node (22) at (4,1) {22};
            
            \node (5) at (0,0) {5};
            \node (6) at (2,0) {6};
            \node (17) at (4,0) {17};
            
            \node (7) at (0,-1) {7};
            \node (8) at (2,-1) {8};
            \node (18) at (4,-1) {18};

            \node (9) at (0,-2) {9};
            \node (10) at (2,-2) {10};
            \node (19) at (4,-2) {19};
            
            \node (11) at (0,-3) {11};
            \node (12) at (2,-3) {12};
            \node (20) at (4,-3) {20};
        \end{scope}

        \begin{scope}[draw=orange, very thick]
            \draw[-] (1) -- (2);
            \draw[-] (3) -- (4);
            \draw[-] (5) -- (6);
            \draw[-] (7) -- (8);
            \draw[-] (9) -- (10);
            \draw[-] (11) -- (12);
            \draw[dotted] (14) -- (21);
            \draw[dotted] (16) -- (22);
            \draw[dotted] (6) -- (17);
            \draw[dotted] (8) -- (18);
            \draw[dotted] (10) -- (19);
            \draw[dotted] (12) -- (20);
            \draw[-, opacity=0.3] (14) -- (21);
            \draw[-, opacity=0.3] (16) -- (22);
            \draw[-, opacity=0.3] (6) -- (17);
            \draw[-, opacity=0.3] (8) -- (18);
            \draw[-, opacity=0.3] (10) -- (19);
            \draw[-, opacity=0.3] (12) -- (20);
        \end{scope}

        \begin{scope}[draw=blue, very thick]
            \draw[-] (13) -- (14);
            \draw[-] (15) -- (16);
            \draw[dotted] (2) -- (13);
            \draw[dotted] (4) -- (15);
            \draw[-, opacity=0.3] (2) -- (13);
            \draw[-, opacity=0.3] (4) -- (15);
        \end{scope}

        \begin{scope}[draw=red, very thick]
            \draw[dotted] (1) -- (3) -- (5) -- (7) -- (9) -- (11);
            \draw[-, opacity=0.3] (1) -- (3) -- (5) -- (7) -- (9) -- (11);
            \draw[-] (21) -- (22) -- (17) -- (18) -- (19) -- (20);
        \end{scope}
        
        \node (l) at (2,-4) {$l=3$} ;
    \end{tikzpicture}
    \quad
    \begin{tikzpicture}[scale=0.7]
        \begin{scope}
            \node (1) at (0,2) {1};
            \node (2) at (1,2) {2};
            \node (17) at (2,2) {17};
            
            \node (3) at (0,1) {3};
            \node (4) at (1,1) {4};
            \node (18) at (2,1) {18};
            
            \node (5) at (0,0) {5};
            \node (6) at (1,0) {6};
            \node (19) at (2,0) {19};
            
            \node (7) at (0,-1) {7};
            \node (8) at (1,-1) {8};
            \node (20) at (2,-1) {20};

            \node (9) at (0,-2) {9};
            \node (10) at (1,-2) {10};
            \node (21) at (2,-2) {21};
            
            \node (11) at (0,-3) {11};
            \node (12) at (1,-3) {12};
            \node (22) at (2,-3) {22};

            \node (13) at (0,-4) {13};
            \node (14) at (1,-4) {14};
            \node (23) at (2,-4) {23};

            \node (15) at (0,-5) {15};
            \node (16) at (1,-5) {16};
            \node (24) at (2,-5) {24};
        \end{scope}

        \begin{scope}[draw=orange, very thick]
            \draw[-] (1) -- (2);
            \draw[-] (3) -- (4);
            \draw[-] (5) -- (6);
            \draw[-] (7) -- (8);
            \draw[-] (9) -- (10);
            \draw[-] (11) -- (12);
            \draw[-] (13) -- (14);
            \draw[-] (15) -- (16);
            \draw[dotted] (2) -- (17);
            \draw[dotted] (4) -- (18);
            \draw[dotted] (6) -- (19);
            \draw[dotted] (8) -- (20);
            \draw[dotted] (10) -- (21);
            \draw[dotted] (12) -- (22);
            \draw[dotted] (14) -- (23);
            \draw[dotted] (16) -- (24);
            \draw[-, opacity=0.3] (2) -- (17);
            \draw[-, opacity=0.3] (4) -- (18);
            \draw[-, opacity=0.3] (6) -- (19);
            \draw[-, opacity=0.3] (8) -- (20);
            \draw[-, opacity=0.3] (10) -- (21);
            \draw[-, opacity=0.3] (12) -- (22);
            \draw[-, opacity=0.3] (14) -- (23);
            \draw[-, opacity=0.3] (16) -- (24);
        \end{scope}

        \begin{scope}[draw=red, very thick]
            \draw[dotted] (1) -- (3) -- (5) -- (7) -- (9) -- (11) -- (13) -- (15);
            \draw[-, opacity=0.3] (1) -- (3) -- (5) -- (7) -- (9) -- (11) -- (13) -- (15);
            \draw[-] (17) -- (18) -- (19) -- (20) -- (21) -- (22) -- (23) -- (24);
        \end{scope}
        
        \node (l) at (1,-6) {$l=4$} ;
    \end{tikzpicture}
    \caption{Gauge connectivity for $k=3$ and $k=4$ codes according to \cref{tab:gauge-from-phys} and \cref{fig:g-XZ-abc}. Solid (dashed) lines are $XX$ ($ZZ$) gauge operators. Orange, blue, and red colors denote $g^{(x,z)}_{a,i},g^{(x,z)}_{b,i}$, and $g^{(x,z)}_{c,i}$ in the notation in \cref{tab:gauge-from-phys}, respectively. The first row of this figure is the same configuration as the second row in \cref{tikz:gauge_graphs-xz-color} with a different color code.}
    \label{tikz:gauge_graphs}
\end{figure*}
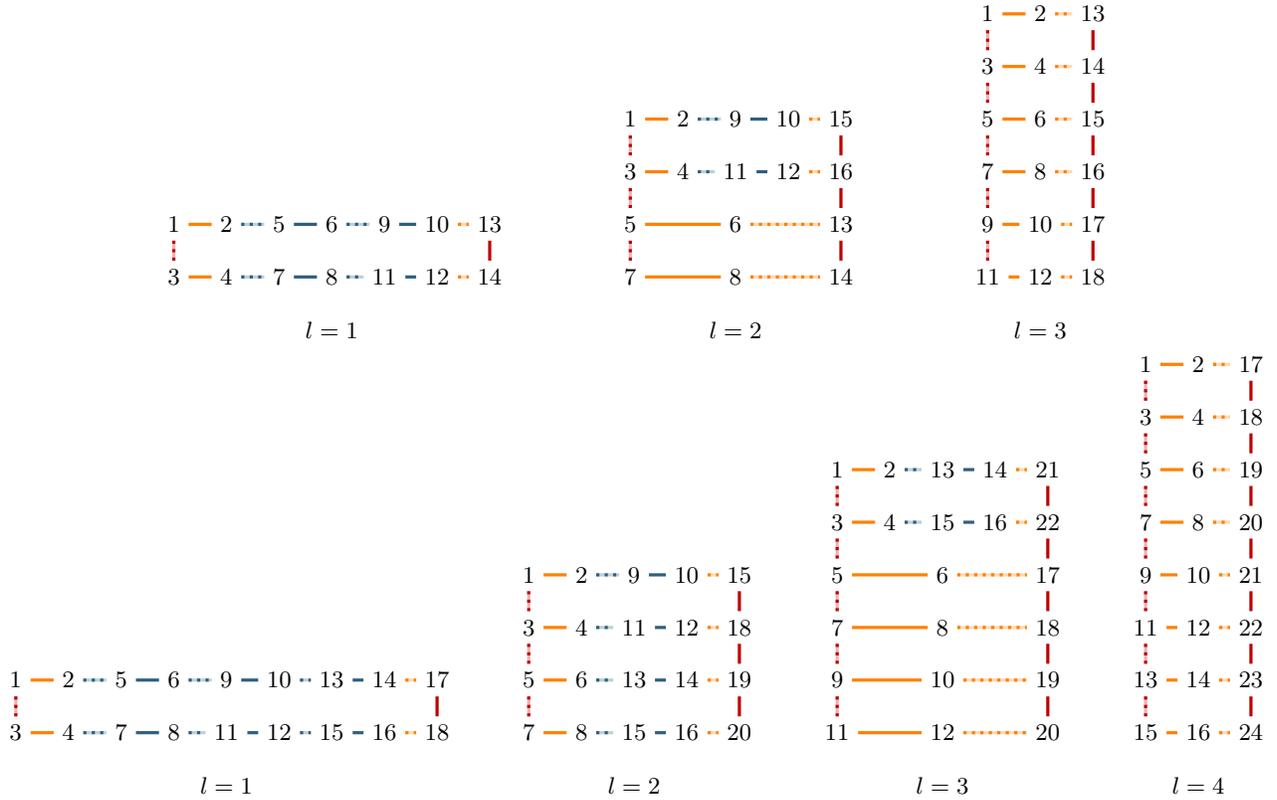

\begin{table*}[t]
\centering
\begin{tabular}{|c|c|c|c|c|c|c|c|c|}
\hline
$l$ & $a_\text{power}$ & $a_\text{power}$ (SE) & $\nu_\text{power}$ & $\nu_\text{power}$ (SE) & $a_\text{expo}$ & $a_\text{expo}$ (SE) & $\nu_\text{expo}$ & $\nu_\text{expo}$ (SE) \\
\hline
1 & 1.683 & 0.008407 & \cellcolor{gray!25}1.032 & 0.003665 & 1.036 & 0.09426 & 0.1966 & 0.01909 \\
2 & 1.107 & 0.04963 & \cellcolor{gray!25}1.032 & 0.02311 & 0.3899 & 0.01495 & 0.1234 & 0.004746 \\
3 & 1.419 & 0.1472 & 1.377 & 0.04799 & 0.256 & 0.005904 & \cellcolor{gray!25}0.1371 & 0.002445 \\
4 & 1.545 & 0.2443 & 1.728 & 0.06736 & 0.1322 & 0.003552 & \cellcolor{gray!25}0.1478 & 0.002448 \\
5 & 1.565 & 0.3098 & 2.098 & 0.0791 & 0.05851 & 0.001213 & \cellcolor{gray!25}0.1576 & 0.001647 \\
6 & 1.405 & 0.1859 & 2.432 & 0.05058 & 0.02522 & 0.0009054 & \cellcolor{gray!25}0.1699 & 0.002589 \\
\hline
\end{tabular}
\caption{Fitting parameters $a_\text{model}$ and $\nu_\text{model}$ for the penalty gap as a function of $m$ for fixed $l$ values from \cref{fig:gap} (left) and their standard error (SE). Grey boxes denote the values of the fit parameters corresponding to the better fits, as determined by the AIC and $R^2$ values reported in \cref{tab:each-l-by-m}.}
\label{tab:params-by-m}
\end{table*}

\begin{table}[h!]
\centering
\begin{tabular}{|c|c|c|c|c|}
\hline
$l$ & AIC power & AIC expo & $R^2$ power & $R^2$ expo \\
\hline
1 & \cellcolor{gray!25}-165.6 & -50.15 & \cellcolor{gray!25}0.9999 & 0.9631 \\
2 & \cellcolor{gray!25}-125.3 & -110.9 & \cellcolor{gray!25}0.9983 & 0.996 \\
3 & -122.1 & \cellcolor{gray!25}-144.6 & 0.9967 & \cellcolor{gray!25}0.9993 \\
4 & -131.3 & \cellcolor{gray!25}-155 & 0.9969 & \cellcolor{gray!25}0.9995 \\
5 & -140.1 & \cellcolor{gray!25}-169.2 & 0.998 & \cellcolor{gray!25}0.9999 \\
6 & -119.1 & \cellcolor{gray!25}-123.7 & 0.9998 & \cellcolor{gray!25}0.9999 \\
\hline
\end{tabular}
\caption{Akaike Information Criterion (AIC) and $R^2$ values for each $l$ as a function of $m$ reported to 
$4$ significant figures. The better values are highlighted in gray, i.e., lower is better for AIC, and higher is better for $R^2$.}
\label{tab:each-l-by-m}
\end{table}

Recall that $m$ denotes the dimension of the matrix $A$. An $A$ matrix of dimension $m \times m$ can support a family of $k$ codes, where $k = \lfloor m/2 \rfloor$. These codes are indexed by $l$, ranging from $l=1$ to $l = k$.

\cref{tab:params-by-m} presents the fitting parameters for both models, along with their associated standard errors as determined by the NonlinearModelFit function in Mathematica. \cref{tab:each-l-by-m} provides the Akaike Information Criterion (AIC) and $R^2$ values for both models. A lower AIC and higher $R^2$ values indicate a better fit to the data between the two models.

The results indicate a transition in behavior from $l=2$ to $l=3$, shifting from power-law decay to exponential decay of the penalty gap. 

\begin{table}[h!]
\centering
\begin{tabular}{|c|c|c|c|c|}
\hline
$l'$ & $a_\text{expo}$ & $a_\text{expo}$ (SE) & $\nu_\text{expo}$ & $\nu_\text{expo}$ (SE) \\
\hline
0 & 2.053 & 0.07242 & 0.4482 & 0.009973 \\
1 & 2.101 & 0.1999 & 0.3767 & 0.01671 \\
2 & 2.886 & 0.4887 & 0.3623 & 0.02195 \\
3 & 4.434 & 1.275 & 0.3573 & 0.02964 \\
4 & 7.864 & 2.807 & 0.3628 & 0.03058 \\
5 & 14.13 & 7.021 & 0.3651 & 0.03644 \\
\hline
\end{tabular}
\caption{Fitting parameters $a_\text{expo}$ and $\nu_\text{expo}$ for the penalty gap as a function of $m$ for different $l=k-l'$ values from \cref{fig:gap} (middle) and their standard error (SE).}
\label{tab:params-by-m-lp}
\end{table}

We also consider $l$ as a function of $k$, i.e., the scaling of the penalty gap when $l=k-l'$ where we fix $l'=0,1,2,\dots$. The gap scaling in this case exhibits an exponential decay (rather than a power law decay) as shown in \cref{fig:gap} (middle). We provide the fitting parameters for $l=k-l'$ where $l'=0,1,\dots,5$ in \cref{tab:params-by-m-lp}.

\subsection{Penalty gap as a function of $l$}

\begin{table*}[t]
\centering
\begin{tabular}{|c|c|c|c|c|c|c|c|c|}
\hline
$m$ & $a_\text{power}$ & $a_\text{power}$ (SE) & $\nu_\text{power}$ & $\nu_\text{power}$ (SE) & $a_\text{expo}$ & $a_\text{expo}$ (SE) & $\nu_\text{expo}$ & $\nu_\text{expo}$ (SE) \\
\hline
10 & 0.1645 & 0.01971 & 0.9866 & 0.2072 & 0.2675 & 0.02311 & \cellcolor{gray!25}0.5036 & 0.04889 \\
11 & 0.1494 & 0.01798 & 1.016 & 0.2136 & 0.2478 & 0.0218 & \cellcolor{gray!25}0.522 & 0.05077 \\
12 & 0.1375 & 0.01576 & 1.123 & 0.2089 & 0.2346 & 0.01798 & \cellcolor{gray!25}0.5541 & 0.04504 \\
13 & 0.1268 & 0.01444 & 1.147 & 0.2127 & 0.2196 & 0.01707 & \cellcolor{gray!25}0.5691 & 0.04645 \\
14 & 0.1179 & 0.01263 & 1.217 & 0.2048 & 0.2078 & 0.0146 & \cellcolor{gray!25}0.5887 & 0.04264 \\
15 & 0.1099 & 0.01163 & 1.244 & 0.2082 & 0.1974 & 0.01409 & \cellcolor{gray!25}0.6071 & 0.04418 \\
16 & 0.1027 & 0.0105 & 1.277 & 0.2079 & 0.1886 & 0.01297 & \cellcolor{gray!25}0.6286 & 0.0435 \\
17 & 0.0961 & 0.01009 & 1.263 & 0.2196 & 0.1798 & 0.0132 & \cellcolor{gray!25}0.6437 & 0.0473 \\
18 & 0.09057 & 0.009235 & 1.292 & 0.2195 & 0.1728 & 0.0123 & \cellcolor{gray!25}0.6629 & 0.04664 \\
19 & 0.08533 & 0.00907 & 1.26 & 0.235 & 0.1651 & 0.01334 & \cellcolor{gray!25}0.6733 & 0.05379 \\
20 & 0.08091 & 0.008369 & 1.289 & 0.2349 & 0.1596 & 0.01239 & \cellcolor{gray!25}0.6921 & 0.0525 \\
\hline
\end{tabular}
\caption{Fitting parameters $a_\text{model}$ and $\nu_\text{model}$ as a function of $l$ for different $m$ values and their standard error (SE).}
\label{tab:params-by-l}
\end{table*}

The penalty gap can also be analyzed as a function of $l$ for fixed values of $m$. \cref{tab:params-by-l} and \cref{tab:each-m-by-l} present fits analogous to those in the previous subsection, but are now fitted as a function of $l$.

\begin{table}[h!]
\centering
\begin{tabular}{|c|c|c|c|c|}
\hline
$m$ & AIC power & AIC expo & $R^2$ power & $R^2$ expo \\
\hline
10 & -20.84 & \cellcolor{gray!25}-29.47 & 0.9703 & \cellcolor{gray!25}0.9947 \\
11 & -21.77 & \cellcolor{gray!25}-30.41 & 0.9696 & \cellcolor{gray!25}0.9946 \\
12 & -28.51 & \cellcolor{gray!25}-40.06 & 0.9619 & \cellcolor{gray!25}0.9944 \\
13 & -29.58 & \cellcolor{gray!25}-41.06 & 0.9619 & \cellcolor{gray!25}0.9944 \\
14 & -37.05 & \cellcolor{gray!25}-51.19 & 0.957 & \cellcolor{gray!25}0.9943 \\
15 & -38.22 & \cellcolor{gray!25}-52.24 & 0.9575 & \cellcolor{gray!25}0.9943 \\
16 & -39.68 & \cellcolor{gray!25}-54.05 & 0.9598 & \cellcolor{gray!25}0.9948 \\
17 & -33.94 & \cellcolor{gray!25}-46 & 0.966 & \cellcolor{gray!25}0.9955 \\
18 & -35 & \cellcolor{gray!25}-47.34 & 0.9675 & \cellcolor{gray!25}0.9958 \\
19 & -28.7 & \cellcolor{gray!25}-38.32 & 0.9736 & \cellcolor{gray!25}0.9961 \\
20 & -29.51 & \cellcolor{gray!25}-39.44 & 0.9747 & \cellcolor{gray!25}0.9965 \\
\hline
\end{tabular}
\caption{Akaike Information Criterion (AIC) and $R^2$ values for each $m$ as a function of $l$. The exponential model is always a better fit.}
\label{tab:each-m-by-l}
\end{table}

In contrast to the gap scaling as a function of $m$, which exhibited a transition from power law to exponential, the decay as a function of $l$ is exponential for all values of $m$, as seen in \cref{tab:each-m-by-l}. 

Fits for smaller values of $m$ are not provided due to the limited number of data points. Specifically, an $A$ matrix of size $m\times m$ can accommodate at most $\lfloor m/2 \rfloor$ codes within the family (e.g., for $m=9$, there are $\lfloor 9/2 \rfloor=4$ codes in the family). We only provide fits for $m$ values with at least $5$ data points.

\section{Mixed Integer Program for Optimization}
\label{app:MIPO}

The mixed integer program is inspired by Ref.~\cite{nannicini2022optimal}, but we ignore the temporal sequencing of gates. We define the following binary variables:
\begin{align}
    w_{q,i}=1 \text{ iff } q\in V \text{ resides at node } i\in V_{\mathrm{m}} .
\end{align}
We then have the following two constraints:
\begin{enumerate}
    \item $\sum_{j\in V_{\mathrm{m}}} w_{q,j} = 1$ $\forall q\in V$ to make sure each qubit is located at exactly one node,
    \item $\sum_{q\in V} w_{q,j} = 1$ $\forall j\in V_{\mathrm{m}}$ to make sure each node can host exactly one qubit.
\end{enumerate}
We calculate the Manhattan distance table from every node to every node using shortest path algorithms and denote the table by $M$, where $M_{ij}=m(i,j)$. The cost function is:
\begin{align}
    \sum_{(p,q)\in E}\sum_{i,j\in V_{\mathrm{m}}} w_{p,i} \times w_{q,j} \times M_{ij} .
\end{align}
Then we can use the following mixed integer programming to define the problem:
\bes
\begin{align}
    \min\limits_{w} & \sum_{(p,q)\in E}\sum_{i,j\in V_{\mathrm{m}}} w_{p,i} \times w_{q,j} \times M_{ij} \\
    \mathrm{s.t.} &~ \sum_{j\in V_{\mathrm{m}}} w_{q,j} = 1 ~ \forall q\in V \\
    &~ \sum_{q\in V} w_{q,j} = 1 ~ \forall j\in V_{\mathrm{m}}
\end{align}
\ees

Since the variable $w$ is binary, we can further reduce the problem to a linear program by defining another binary variable:
\begin{align}
    y_{(p,q)(i,j)} = 1 \text{ iff } (p,q)\in E \text{ with }&p\in V\text{ locate at }i\in V_{\mathrm{m}} \notag\\
    \text{and } & q\in V\text{ locate at }j\in V_{\mathrm{m}}
\end{align}
Thus, we have the following constraints:
\begin{enumerate}
    \item $\sum\limits_{i,j\in V_{\mathrm{m}}} y_{(p,q)(i,j)}=1$ $\forall (p,q) \in E$ to make sure each edge is implemented exactly once 
    \item $y_{(p,q)(i,j)}=w_{p,i}\times w_{q,j}$
\end{enumerate}
Using McCormick's inequality, we can reduce the second constraint to:
\begin{align}
    \{w_{p,i},w_{q,j}\} \geq y_{(p,q)(i,j)} \geq \{0,w_{p,i}+w_{q,j}-1\}
\end{align}
Then the linear program is:
\bes
\begin{align}
    \min\limits_{w} & \sum_{(p,q)\in E}\sum_{i,j\in V_{\mathrm{m}}} y_{(p,q)(i,j)} \times M_{ij} \\
    \mathrm{s.t.} &~ \sum_{j\in V_{\mathrm{m}}} w_{q,j} = 1 ~ \forall q\in V \\
    &~ \sum_{q\in V} w_{q,j} = 1 ~ \forall j\in V_{\mathrm{m}} \\
    &~ \sum_{i,j\in V_{\mathrm{m}}} y_{(p,q)(i,j)}=1 ~ \forall (p,q) \in E \\
    &~ y_{(p,q)(i,j)} \leq w_{(p,i)} ~ \forall (p,q)\in E, \forall i,j\in V_{\mathrm{m}} \\
    &~ y_{(p,q)(i,j)} \leq w_{(q,j)} ~ \forall (p,q)\in E, \forall i,j\in V_{\mathrm{m}} \\
    &~ y_{(p,q)(i,j)} \geq w_{(p,i)}+w_{(q,j)}-1 ~ \forall (p,q)\in E, \forall i,j\in V_{\mathrm{m}}
\end{align}
\ees

\end{document}